\documentclass[11pt]{article}

\usepackage[letterpaper, margin=1in]{geometry}
\usepackage{amsmath, amssymb, amsfonts, amsthm}
\usepackage{bm,bbm}
\usepackage{mathtools}
\usepackage{mathrsfs}
\usepackage{enumitem}
\usepackage{setspace}
\usepackage[utf8x]{inputenc}
\usepackage{algpseudocode,algorithmicx,algorithm}
\usepackage[T1]{fontenc}
\usepackage{lmodern}
\usepackage{graphicx}
\usepackage[natural, dvipsnames, pdftex]{xcolor}
\usepackage{tikz}
\usepackage{verbatim}
\usepackage[colorlinks,
            linkcolor=blue,
            anchorcolor=blue,
            citecolor=blue
            ]{hyperref}
\usepackage{cancel}
\usepackage[all]{xy}
\usepackage{mdframed}
\usepackage[capitalize]{cleveref}

\newif\ifmynotes
\mynotestrue
\newcommand{\ZGnote}[1]{\ifmynotes\textcolor{purple}{[Zeyu: #1]}\else\fi}

\newcommand{\ZZnote}[1]{\ifmynotes\textcolor{violet}{[Zihan: #1]}\else\fi}

\theoremstyle{plain}
\newtheorem{thm}{Theorem}[section]
\newtheorem{lemma}[thm]{Lemma}

\newtheorem{fa}[thm]{Fact}
\newtheorem{Op}[thm]{Open Problem}
\newtheorem{prop}[thm]{Proposition}
\newtheorem{cor}[thm]{Corollary}

\theoremstyle{definition}
\newtheorem{rmka}[thm]{Remark}
\newtheorem{defn}[thm]{Definition}
\newtheorem{conj}[thm]{Conjecture}

\theoremstyle{remark}

\newcommand{\ldmrd}[1]{\operatorname{LD-MRD}({\leq}\,#1)}

\newcommand{\N}{\mathbb{N}}

\newcommand{\F}{\mathbb{F}}
\newcommand{\K}{\mathbb{K}}

\newcommand{\cS}{\mathcal{S}}

\newcommand{\eps}{\varepsilon}

\newcommand{\spa}{\mathrm{span}}

\newcommand{\abs}[1]{\left| #1 \right|}

\newcommand{\bv}{\bm{v}}
\newcommand{\bu}{\bm{u}}
\newcommand{\vm}{\bm{m}}

\newcommand{\bw}{\bm{w}}
\newcommand{\by}{\bm{y}}
\newcommand{\bx}{\bm{x}}
\newcommand{\bz}{\bm{z}}

\newcommand{\bc}{\bm{c}}

\newcommand{\rank}{\mathrm{rank}}
\renewcommand{\epsilon}{\varepsilon}

\def \mC {\mathcal{C}}

\def \mG {\mathcal{G}}
\def \mC {\mathcal{C}}

\def \mF {\mathcal{F}}

\usepackage{pdfpages}

\begin{document}

\title{Random Gabidulin Codes Achieve List Decoding Capacity in the Rank Metric}
\date{}
\author{
Zeyu Guo\thanks{Department of Computer Science and Engineering, The Ohio State University. \href{mailto:zguotcs@gmail.com}{\texttt{zguotcs@gmail.com}}}  
\and
Chaoping Xing\thanks{School of Electronic Information and Electric Engineering, Shanghai Jiao Tong
University. \href{mailto:xingcp@sjtu.edu.cn}{\texttt{xingcp@sjtu.edu.cn}}}  
\and
Chen Yuan\thanks{School of Electronic Information and Electric Engineering, Shanghai Jiao Tong
University. \href{mailto: chen\_yuan@sjtu.edu.cn}{\texttt{chen\_yuan@sjtu.edu.cn}}}  
\and
Zihan Zhang\thanks{Department of Computer Science and Engineering, The Ohio State University. \href{mailto:zhang.13691@buckeyemail.osu.edu}{\texttt{zhang.13691@osu.edu}}}  
}
\maketitle
\thispagestyle{empty}

\begin{center}
\fbox{%
\begin{minipage}{0.93\textwidth}
\small
\textbf{Author note (July 2026).}
The proof of Theorem~4.7 in the original version contains an error:
the dimension of an intersection of subspaces may increase under a
linear projection.  A corrected proof is included in the erratum
appended to this version.  The corrected theorem requires the additional
assumption \(q\ge m-1\).  Consequently, Theorems~1.17 and~1.18 require
\(q\ge k-1\), while Theorem~1.3 and Corollary~1.4 require
\(q\ge n-k-1\).

Proposition~5.1 admits a new direct proof independent of Theorem~4.7.
Therefore, the generic-intersection formula, the equivalence results in
Sections~5 and~6, and the corresponding statements in
Theorems~1.13, 1.14, and~1.16 and Corollary~1.15 remain valid without
any restriction on \(q\).

The same correction implies that the main positive result of the
follow-up paper \emph{Gabidulin Codes Achieve List Decoding Capacity
with an Order-Optimal Column-To-Row Ratio} requires the assumption
\(q\ge n-k-1\).
\end{minipage}%
}
\end{center}
\medskip

\begin{abstract}
Gabidulin codes, serving as the rank-metric counterpart of Reed--Solomon codes, constitute an important class of maximum rank distance (MRD) codes. However, unlike the fruitful positive results about the list decoding of Reed--Solomon codes, results concerning the list decodability of Gabidulin codes in the rank metric are all negative so far. For example, in contrast to Reed--Solomon codes, which are always list decodable up to the Johnson bound in the Hamming metric,
Raviv and Wachter-Zeh (IEEE TIT, 2016 and 2017) constructed a class of Gabidulin codes that are not even combinatorially list decodable beyond the unique decoding radius in the rank metric. Proving the existence of Gabidulin codes with good combinatorial list decodability in the rank metric has remained a long-standing open problem.

In this paper, we resolve the aforementioned open problem by showing that, with high probability, random Gabidulin codes over sufficiently large alphabets attain the optimal generalized Singleton bound for list decoding in the rank metric. In particular, they achieve list decoding capacity in the rank metric.

Our work is significantly influenced by the recent breakthroughs in the combinatorial list decodability of Reed--Solomon codes, especially the work by Brakensiek, Gopi, and Makam (STOC 2023). Our major conceptual and technical contributions, which may hold independent interest, consist of the following: (1) We initiate the study of ``higher order MRD codes'' and provide a novel unified theory, which runs parallel to the theory of ``higher order MDS codes'' developed by Brakensiek, Gopi, and Makam. (2) We prove a natural analog of the GM-MDS theorem, proven by Lovett (FOCS 2018) and Yildiz and Hassibi (IEEE TIT, 2019), which we call the GM-MRD theorem. In particular, our GM-MRD theorem for Gabidulin codes is strictly stronger than the GM-MDS theorem for Gabidulin codes proven by Yildiz and Hassibi.


\end{abstract}
\newpage
\tableofcontents
\thispagestyle{empty}
\newpage
\setcounter{page}{1}
\section{Introduction}
A rank-metric code is a collection of matrices in $\F_q^{m\times n}$ with $m \geq n$, where the distance between two matrices $A$ and $B$ is defined to be the rank of $A-B$.
Introduced by Delsarte \cite{Del} as a combinatorial curiosity, rank-metric codes have since developed into a field of study with applications and connections spanning network coding \cite{KK07,SKK08,kK08,SK09}, space-time coding \cite{LGB03, LK05}, cryptography \cite{Gib96,Gib95,Loi10,Loi17}, and pseudorandomness \cite{FS12,FG15,GWX16,GRX21,GVJZ23}.

A rank-metric code $C$, with rate $R$ and relative minimum distance $\delta$, must satisfy the Singleton bound $1-R \geq \delta$. If the equality $1-R = \delta$ is attained, then the code $C$ is called a Maximum Rank Distance (MRD) code. Gabidulin codes \cite{Del,G85,Rot91}, an important class of MRD codes, can be viewed as the linearized versions of Reed--Solomon codes. They are defined through the evaluation of linearized polynomials within a subspace. This perspective positions Gabidulin codes as a perfect analogy to Reed--Solomon codes in the rank metric setting. Similar to Reed--Solomon codes, one can design highly efficient encoding and unique decoding algorithms for Gabidulin codes by generalizing the Berlekamp--Welch algorithm \cite{kK08}. However, regarding the list decoding regime, it is unknown whether exist any Gabidulin codes that can be list decoded beyond the unique decoding radius. In comparison, the celebrated Guruswami--Sudan list decoding algorithm \cite{GS99} can list decode any Reed--Solomon code up to the Johnson bound. Thus, a long-standing open problem remains for rank-metric codes:
\begin{Op}
Are Gabidulin codes algorithmically or combinatorially list decodable\footnote{Combinatorial list decodability refers to the condition where the output list of candidate codewords, within the list decoding radius, is small, whereas algorithmic list decodability further requires that there exists an efficient algorithm outputting this list.} beyond the unique decoding radius in the rank metric?
\end{Op}
In this paper, we provide a positive answer to the ``combinatorial'' aspect of this problem. 
Namely, we prove that Gabidulin codes with random evaluation subspaces are combinatorially list decodable up to the generalized Singleton bound with high probability in the rank metric. Our approach is inspired by the recent progress regarding the list decodability of random Reed--Solomon codes.
\subsection{List Decodability of Reed--Solomon Codes}
In the Hamming metric, the Singleton bound \cite{Sin64} states that any code with a rate $R$ and relative minimum distance $\delta$ must satisfy $R + \delta \leq 1$. Codes that attain this bound are called Maximally Distance Separable (MDS) codes. Reed--Solomon codes are a class of Hamming-metric codes that attain the Singleton bound. This bound was recently generalized to the list decoding setting by investigating the list decoding of Reed--Solomon codes. Shangguan and Tamo\footnote{Before the work of Shangguan and Tamo, Rudra and Wootters \cite{rudra2014every} were the first to show that random RS codes are list decodable beyond the Johnson radius for some parameter regimes. Their methodology is purely analytic instead of algebraic.} \cite{shangguan2020combinatorial} proposed the generalized Singleton bound $\rho \leq \frac{L}{L+1}(1-R)$, where $\rho$ is the list decoding radius and $L$ is the list size. Moreover, they showed that such a bound is achievable for $L=2,3$. Since then, there have been some efforts to prove the tightness of this bound \cite{GSTW21,FKS22,GST23}. A major breakthrough was made by Brakensiek, Gopi, and Makam \cite{brakensiek2023generic} in proving that this generalized Singleton bound holds for any list size $L$. In particular, their argument demonstrated that Reed--Solomon codes with random evaluation points defined over $\F_q$ (where $q$ is exponential in the length of the code) can attain the generalized Singleton bound with high probability. 

The exponential size of the field is inevitable \cite{brakensiek2022lower}. If we accept a $\epsilon$ gap from this bound, i.e., $\rho = \frac{L}{L+1}(1-R-\epsilon)$, then the size of the field can be reduced to $O_{\epsilon}(n^2)$ \cite{guo2023randomly} and further to $O_{\epsilon}(n)$ \cite{alrabiah2023randomly}. In the context of constant-sized field sizes, Brakensiek, Dhar, Gopi, and Zhang \cite{brakensiek2023ag} showed that algebraic geometry codes with random evaluation points defined over $\F_q$, with $q = \exp(O(1/\epsilon^2))$, can be list decoded up to the radius $\frac{L}{L+1}(1-R-\epsilon)$. The exponential dependence of $q$ on $1/\epsilon$ was proved to be necessary ($q \geq \exp(\Omega(1/\epsilon))$) in \cite{alrabiah2023ag}. Also, based on the frameworks built by \cite{brakensiek2023generic,guo2023randomly,alrabiah2023randomly}, Ron-Zewi, Venkitesh, and Wootters \cite{ron2024efficient} recently showed similar list decodability phenomenon holds for polynomial ideal codes, which includes several well-studied families of error-correcting codes such as Reed--Solomon codes, folded Reed--Solomon codes, and multiplicity codes.

Our understanding of the generalized Singleton bound is nearly complete, thanks to the recent progress mentioned above. One might ask whether this generalized Singleton bound holds for codes in other metrics, such as the rank metric. We note that a code with a minimum rank distance of $d$ is also a code with a minimum Hamming distance of at least $d$. Thus, the argument for the upper bound of the generalized Singleton bound can be applied straightforwardly. However, proving the tightness of this upper bound appears to be highly non-trivial.

\subsection{List Decodability of Rank-Metric codes and Gabidulin Codes}
Let us first review some results concerning the list decoding of rank-metric codes. Ding \cite{ding2014list} proposed the Gilbert-Varshamov bound (GV) for the list decoding of rank-metric codes. Specifically, she showed that, with high probability, a random rank-metric code can be list decodable up to the GV bound, and any rank-metric code cannot be list decodable beyond this GV bound. In \cite{ding2014list}, the list size for random linear rank-metric codes is given as $O\left((\frac{1}{\epsilon})^{\frac{1}{\epsilon}}\right)$, a result of the limited randomness available for these codes.
Guruswami and Resch \cite{GR18} adopted the ideas from \cite{GHK10} to further reduce the list size of random linear rank-metric codes to a constant $O(1/\epsilon)$.
The results mentioned above are not explicit.
Meanwhile, Guruswami, Wang, and Xing \cite{GWX16} presented the first explicit class of rank-metric codes with efficient list decoding algorithms capable of decoding up to the Singleton bound. Their construction involved carefully selecting a subcode of the Gabidulin code using a tool known as subspace designs \cite{GK13}, thereby reducing the list size to a constant. Later on, Xing and Yuan \cite{XY18} introduced another explicit class of rank-metric codes that can be list decoded up to the Singleton bound. Their construction borrowed ideas from folded Reed--Solomon codes to fold Gabidulin codes in a similar manner. To approach the Singleton bound, the column-to-row ratio $\frac{n}{m}$ should be close to $0$. In the regime where the column-to-row ratio is constant, the construction from \cite{XY18} remains applicable and can correct up to a $\frac{2}{3}(1-2R)$ fraction of errors. By modifying this approach, Liu, Xing, and Yuan \cite{LYX23} presented an explicit construction of list decodable rank-metric codes with a ratio of $2/3$. An open question remains: Can we explicitly construct a class of rank-metric codes with a column-to-row ratio of $1$ that is list decodable beyond the unique decoding radius? For comparison, the Gilbert--Varshamov (GV) bound argument suggests that, with high probability, there exist random rank-metric codes with a column-to-row ratio of $1$ that can be list decoded up to a radius of $1-\sqrt{R}$, which is strictly larger than the unique decoding radius of $\frac{1-R}{2}$.

Although there are a few positive results for the list decoding of rank-metric codes, the findings regarding the list decoding of Gabidulin codes have been predominantly negative so far. Wachter-Zeh \cite{WZ13} proved that any square Gabidulin codes---where ``square'' indicates that the column-to-row ratio is $1$---cannot be list decoded beyond a radius of $1-\sqrt{R}$. Furthermore, Raviv and Wachter-Zeh \cite{raviv2016some,raviv2017correction} constructed a class of Gabidulin codes that cannot be list decoded to any radius beyond the unique decoding radius. Despite these negative results, our work demonstrates the existence of Gabidulin codes that can be list decoded up to the generalized Singleton bound. We briefly review some applications based on the list decoding of rank-metric codes or Gabidulin codes and discuss the implications of our results. 
\paragraph{Code-based cryptography.} Rank-metric code-based cryptosystems have been studied since 1991 \cite{PT91,CS96}. Rank-metric codes play an important role in designing post-quantum cryptosystems. Two well-known systems, RQC \cite{MAB2019b} and ROLLO \cite{MAB2019a}, have been considered for the NIST standardization. The main advantage of rank-metric codes is that the hard problems in the rank metric, such as the generic decoding problem, seem harder to solve than their counterparts in the Hamming metric. The random syndrome decoding (RSD) problem, in particular, assumes that it is computationally hard to decode a random $[n,k]_{\F_{q^m}}$-linear code\footnote{See \cref{Prelim} for the definition of $[n,k]_{\F_{q^m}}$-linear codes.} in the rank metric. Recent developments in attacks based on G\"obner bases \cite{BBBGNRT20a,BBCGPSTV20b} have significantly impacted the security parameters of cryptosystems relying on the hardness of the search RSD problem. Given that the search RSD problem is not as hard as previously believed, the list search version of the RSD problem, which is strictly harder than the search RSD problem, might be considered. This problem outputs a list of codewords within a given radius of a random $[n,k]_{\F_{q^m}}$-linear code. The Gabidulin code list search RSD problem was thought to be more challenging than the list search RSD problem \cite{RJBPLW20}, mainly because only negative results were available for the list decoding of Gabidulin codes. Faure and Loidreau \cite{FL05} proposed the Faure–Loidreau (FL) cryptosystem, which is connected with the list decoding of Gabidulin codes. The original FL cryptosystem was compromised by exploiting the list decoding of interleaved Gabidulin codes. However, the revised FL cryptosystem, LIGA \cite{RPW21}, which relies on the hardness of list decoding Gabidulin codes, is resistant to this attack. The hardness assumption for the list decoding of Gabidulin codes suggests that it is computationally difficult to find all codewords of a Gabidulin code within the list decoding radius $1-\sqrt{R}$. Prior to our work, it was widely believed that the output list for this list decoding problem should be exponential. Thus, our work enhances the understanding of the list decoding of Gabidulin codes, which could significantly impact the design of rank-metric code-based cryptosystems.

\paragraph{Pseudorandomness.}

Rank-metric codes and Gabidulin codes have also found applications and connections in the field of pseudorandomness. In \cite{FG15}, Forbes and Guruswami studied various objects related to ``linear-algebraic pseudorandomness.'' They specifically proved that bilinear lossless two-source rank condensers are equivalent to linear rank-metric codes. Consequently, they showed that Gabidulin codes translate into optimal two-source rank condensers.
Inspired by the explicit constructions of subcodes of Gabidulin codes that are list decodable in the rank metric \cite{GX13, GWX16}, Guruswami, Resch, and Xing \cite{GRX21} presented an explicit construction of dimension expanders, which can be seen as the linear-algebraic analogs of expander graphs. This construction achieves excellent ``lossless'' expansion.
Cheraghchi, Didier, and Shokrollahi \cite{MDS11} used Gabidulin codes to construct explicit affine extractors for a restricted family of affine sources over large fields, which have applications in wiretap protocols.
In \cite{GVJZ23}, Guo, Volk, Jalan, and Zuckerman considered $(\epsilon, e)$-biased sources over $\F_p$. These are random sources $\mathcal{X}$ over $\F_p^n$ that are $\epsilon$-biased against all but a subgroup $H$ of characters, where $|H| \leq e$, thereby generalizing affine sources of small codimension. Using Gabidulin codes, they constructed deterministic extractors that extract almost all the min-entropy from such sources. This construction was further utilized as a component in constructing extractors for other algebraic sources.

Rank-metric codes have also been connected to other problems studied in theoretical computer science, such as low-rank recovery \cite{FS12}. Despite these applications, we believe that the potential use of rank-metric codes in theoretical computer science has not been fully explored. In particular, given the intimate connection between pseudorandomness and error-correcting codes in the Hamming metric, especially those with good list decodability or list recoverability \cite{GUV09, Vad12}, it is conceivable that rank-metric codes and Gabidulin codes may find similar applications. 

\subsection{Our Results} 

Our main results can be divided into three parts:
\begin{enumerate}
    \item The optimal (combinatorial) list decodability of random Gabidulin codes over sufficiently large alphabets in the rank metric.
    \item The formulation of three notions of ``higher order MRD codes'', which we denote as $\operatorname{GKP}(\ell)$, $\operatorname{MRD}(\ell)$, and $\operatorname{LD-MRD}(\ell)$, and the demonstration of their equivalence.
    \item A result we call the $\operatorname{GM-MRD}$ theorem, which states that symbolic Gabidulin codes satisfy $\operatorname{GKP}(\ell)$ for all $\ell$.
\end{enumerate} 

The first two items constitute a theory parallel to that of Brakensiek, Gopi, and Makam \cite{brakensiek2023generic}, who demonstrated the optimal list decodability of random Reed-Solomon codes over sufficiently large alphabets and established the equivalence among three notions of ``higher order MDS codes.'' The third item, i.e., the GM-MRD theorem, serves as a rank-metric analog to the \emph{GM-MDS theorem}, which was conjectured in \cite{dau2014existence} and subsequently proved in \cite{lovett2018mds,yildiz2019optimum}.
In the following, we will explain each of these three items in detail.

\subsubsection{Optimal List Decodability of Random Gabidulin Codes} 

Recall that for $\rho\in [0,1]$, a code $C\subseteq\Sigma^n$ over an alphabet $\Sigma$  is said to be \emph{$(\rho,\ell)$-list decodable} if for any $\by\in\F_q^n$, it holds that 
\[
|\{\bx\in C: d(\bx,\by)\leq \rho n\}|\leq \ell.
\]
where $d(\bx,\by)$ denotes the distance between $\bx$ and $\by$. Here $\rho$ is called the list decoding radius and $\ell$ is called the list size.

In \cite{shangguan2020combinatorial}, Shangguan and Tamo proved the \emph{generalized Singleton bound} for list decoding, generalizing the classical Singleton bound for unique decoding. For linear codes, this generalized Singleton bound states that 
if $C\subseteq\F_q^n$ is an $[n,k]$-linear code that is $(\rho,\ell)$-list decodable in the Hamming metric, then it holds that $\rho\leq \frac{\ell}{\ell+1}\left(1-\frac{k}{n}\right)$.

Note that the rank distance $d_R(\bx,\by)$ between $\bx,\by\in\F_{q^m}^n$ is always bounded by their the Hamming distance $d_H(\bx,\by)$ from above. This is because $\bx-\by$, viewed as an $m\times n$ matrix over $\F_q$, has at most $d_H(\bx,\by)$ nonzero columns. It follows that a $(\rho,\ell)$-list decodable code $C\subseteq\F_{q^m}^n$ in the rank metric remains $(\rho,\ell)$-list decodable in the Hamming metric. This immediately implies that the generalized Singleton bound  continues to hold in the rank metric:

\begin{lemma}[Generalized Singleton bound for rank-metric codes]
Let $C\subseteq \F_{q^m}^n$ be an $[n,k]_{\F_{q^m}}$-linear code that is $(\rho,\ell)$-list decodable in the rank metric. Then it holds that
\[
\rho\leq \frac{\ell}{\ell+1}\left(1-\frac{k}{n}\right).
\]
\end{lemma}

The central objects studied in this paper are \emph{Gabidulin codes}.
For integers $m\geq n\geq k$ and elements $\alpha_1,\dotsc, \alpha_n\in\F_{q^m}$ that are linearly independent over $\F_q$, the corresponding Gabidulin code over $\F_{q^m}$ is defined to be the $[n,k]_{\F_{q^m}}$-linear code
\[
\mG_{n,k}(\alpha_1,\dotsc, \alpha_n) := \left\{(f(\alpha_1),\ldots,f(\alpha_n)):\; f(X)=\sum_{i=1}^{k} c_i X^{q^{i-1}}, c_1,\dots,c_k\in \F_q^m\right\}\subseteq \F_{q^m}^n.
\]

Our main theorems state that, for random $\alpha_1,\dots,\alpha_n$ and sufficiently large $m$, with high probability, the Gabidulin code $\mG_{n,k}(\alpha_1,\dotsc, \alpha_n)$ over $\F_{q^m}$ are list decodable up to a radius that exactly attains the generalized Singleton bound.

\begin{thm}[Informal version of \cref{thm:list-decodability}]\label{mm1}
Let $(\alpha_1,\dots,\alpha_n)$ be uniformly distributed over the set of all vectors in $\F_{q^m}^n$ whose coordinates are linearly independent over $\F_q$.\footnote{While we assume that $(\alpha_1,\dots,\alpha_n)$ is sampled from the set of all vectors with $\F_q$-linearly independent coordinates, one can also sample it from the whole set $\F_{q^m}^n$ given that most vectors in this set have $\F_q$-linearly independent coordinates. This distinction is unimportant.}
Suppose $m\geq cn(n-k)\ell+\log_q(1/\delta)$, where $c$ is a large enough absolute constant.
Then it holds with probability at least $1-\delta$ that the Gabidulin code $\mG_{n,k}(\alpha_1,\dots, \alpha_n)$ over $\F_{q^m}$ is $\left(\frac{L}{L+1}\left(1-{k}/{n}\right), L\right)$-list decodable for all $L\in [\ell]$ in the rank metric.
\end{thm}
\begin{cor}[Informal version of \cref{cor:mm1}]\label{mm2}
Let $(\alpha_1,\dots,\alpha_n)$ be uniformly distributed over the set of all vectors in $\F_{q^m}^n$ whose coordinates are linearly independent over $\F_q$. 
Suppose $m\geq cn(n-k)/\eps+\log_q(1/\delta)$, where $c$ is a large enough absolute constant.
Then it holds with probability at least $1-\delta$ that the Gabidulin code $\mG_{n,k}(\alpha_1,\dots, \alpha_n)$ over $\F_{q^m}$ is $\left(1-R-\epsilon,\frac{1-R}{\epsilon}\right)$-list decodable, where $R=k/n$ is the rate of the code. 
\end{cor}

In fact, we prove the stronger statement that $\mG_{n,k}(\alpha_1,\dots, \alpha_n)$ is, with high probability, \emph{average-radius list decodable} with the parameters stated in \cref{mm1} and \cref{mm2}. See \cref{thm:list-decodability} and \cref{cor:mm1} for details. Average-radius list decodability is stronger than standard list decodability, which we will discuss shortly when defining the notion $\operatorname{LD-MRD}(\ell)$.

\paragraph{Field size lower bound.} 
\cref{mm1} shows that Gabidulin codes over $\F_{q^m}$ can attain the Singleton bound (even in the sense of average-radius list decodability) for some $m=O_\ell(n^2)$.
To complement this upper bound, we establish a matching lower bound on the field size via a technique developed in \cite{alrabiah2023ag}. 
\begin{thm}[Informal version of \cref{A.1}]\label{thm:informal-lower-bound}
Let $\ell\geq 2$.
Let $C\subseteq\F_{q^{m}}^{n}$ be a rank-metric code over $\F_{q^m}$ of rate $R$. If $C$ is $\left(\frac{\ell}{\ell+1}(1-R),\ell\right)$-average-radius list decodable (see \cref{defn:average}) and $R\in [c, 1-c-\ell/n]$ for some constant $c>0$, then $m=\Omega_{\ell}(n^2)$.
\end{thm}

\subsubsection{Higher Order MRD Codes} 

Maximum rank distance (MRD) codes are the counterparts of maximum distance separable (MDS) codes in the rank metric.
Recall for $m\geq n\geq k$, an $[n,k]_{\F_{q^m}}$-linear rank-metric code $C$ is a MRD code if its minimum (rank) distance $d(C)$ attains the Singleton bound, i.e., $d(C)=n-k+1$.

Brakensiek, Gopi, and Makam \cite{brakensiek2023generic} studied three different definitions of ``higher order MDS codes,'' which they call $\operatorname{GZP}(\ell)$, $\operatorname{MDS}(\ell)$, and $\operatorname{LD-MDS}(\ell)$ codes. Amazingly, they proved that these three notions are equivalent in a certain rigorous sense.
This equivalence was crucially used in their proof that generic Reed--Solomon codes achieve list decoding capacity.

This raises the questions of whether similar notions of ``higher order MRD codes'' exist, and if so, whether there is also an equivalence among them. In this paper, we show that the answer to both questions is yes.

Specifically, we introduce the notions of $\operatorname{GKP}(\ell)$, $\operatorname{MRD}(\ell)$, and $\operatorname{LD-MRD}(\ell)$ codes.
Just as the MRD property strengthens the MDS property, each of these three notions strengthens their counterpart in the Hamming metric, as studied in \cite{brakensiek2023generic}. Furthermore, we establish an equivalence among these three notions. This equivalence plays a crucial role in our proof that, with high probability, random Gabidulin codes achieve list decoding capacity in the rank metric.

\paragraph{Rank-metric codes over a general field $\F/\F_q$.}
Before defining the three notions of ``higher order MRD codes,'' we note that, for convenience, we will present these definitions over a general extension field $\F$ of $\F_q$, which may be infinite. 
This generality allows us to easily discuss the ``symbolic Gabidulin code,'' defined over a function field $\F_q(Z_1,\dots,Z_n)$.
For infinite $\F$, while it is no longer possible to view a vector $\bv=(v_1,\dots,v_n)\in \F^n$ as a matrix over $\F_q$ and discuss its rank (since $m=[\F:\F_q]$ is infinite), we can still define the rank of $\bv$ as
\begin{equation}\label{eq:rank}
\rank_{\F_q}(\bv):=\dim_{\F_q}(\spa_{\F_q}\{v_1,\ldots,v_n\}).
\end{equation}
With this revised definition of rank, the concept of linear rank-metric codes over $\F_{q^m}$ can easily be extended to those over a general field $\F/\F_q$, and fortunately, all necessary facts and properties continue to hold. For a more comprehensive discussion about linear rank-metric codes over a general field $\F/\F_q$, we refer the readers to \cref{Prelim}.


\subparagraph{$\blacktriangleright\operatorname{GKP}(\ell)$.}

Given $k\leq n$, a \emph{zero pattern} is a tuple of sets $\mathcal{S}=(S_1,\dots,S_k)$ with $S_i\subseteq [n]$.
For an $[n,k]_{\F}$ linear code $C$ with a generator matrix $G\subseteq \F^{n\times k}$, we say $C$ \emph{attains} a zero pattern $\mathcal{S}$ if there exists an invertible matrix $M_{\mathcal{S}}\in\F^{k\times k}$ such that for $(i,j)\in [k]\times [n]$, the $(i,j)$-th entry of $M_{\mathcal{S}} G$ is zero if $j\in S_i$, i.e., $M_{\mathcal{S}} G$ exhibits the zero pattern $\mathcal{S}$.

This notion was driven by applications such as those discussed in \cite{dau2014simple,yan2013algorithms,yan2014weakly} in order to find linear MDS codes with sparse generator matrices.
A natural question emerges: What are the zero patterns that linear MDS codes such as Reed--Solomon codes can attain? To understand this, the notion of \emph{generic zero patterns} (GZPs) was defined in \cite{brakensiek2023generic}, which originated from \cite{dau2014simple}: A zero pattern $\mathcal{S}=(S_1,\ldots,S_k)$ is called a generic zero pattern if 
\begin{equation}\label{eq:gzp}
\left|\bigcap_{i\in \Omega} S_i\right|\leq k-|\Omega| \qquad \text{for all nonempty } \Omega\subseteq [k].
\end{equation}

It is not difficult to prove that \eqref{eq:gzp} is a necessary condition for a linear MDS code to attain the zero pattern $\mathcal{S}$. It was conjectured in \cite{dau2014simple} that for any zero pattern $\mathcal{S}$, there exist MDS codes with alphabet size $q\ge n+k-1$ attaining $\mathcal{S}$.
This conjecture, known as the \emph{$\operatorname{GM-MDS}$ conjecture}, was later proven by Lovett \cite{lovett2018mds} and independently by Yildiz and Hassibi \cite{yildiz2019optimum}, and has become the GM-MDS theorem.
In fact, their proofs imply the stronger statement that generic Reed--Solomon codes attain all GZPs. This theorem was further generalized in \cite{brakensiek2023generalized} to any polynomial codes, whereas the original version only deals with the Reed--Solomon codes. Another generalization was recently proved in \cite{ron2024efficient} to establish the near-optimal list decodability of folded Reed--Solomon codes.

We define a more general notion, which we call generic kernel patterns (GKPs). It appears to be more suitable for studying rank-metric codes, including MRD codes.

\begin{defn}[Generic kernel pattern (GKP)]\label{def:GKP}
Given $k\leq n$ and a finite field $\F_q$, a \emph{kernel pattern} 
is a tuple $\mathcal{V}=(V_1,\ldots,V_k)$ of $\F_q$-linear subspaces of $\F_q^n$,
We say a kernel pattern $\mathcal{V}$ is a \emph{generic kernel pattern} (GKP) if for any nonempty $\Omega\subseteq [k]$, it holds that
\begin{equation}
\dim\left(\bigcap_{i\in \Omega} V_i\right)\leq k-|\Omega|.
\end{equation}
In addition, we say a kernel pattern $\mathcal{V}$ has \emph{order $\ell$} if $\mathcal{V}$ has exactly $\ell$ distinct nonzero subspaces. 
\end{defn}


Next, we define a linear code attaining a kernel pattern.

\begin{defn}
Let $C$ be an $[n,k]_{\F}$-linear code with a generator matrix $G\times \F^{k\times n}$.
Let $\mathcal{V}=(V_1,\ldots,V_k)$ be a kernel pattern and $(A_1,\ldots,A_k)$ be a tuple of $k$ full rank matrices such that $A_i\in\F_q^{n\times \dim V_i}$ with $\langle A_i\rangle=V_i$\footnote{For a matrix $H\in\F_q^{n\times \ell}$, we denote $\langle H\rangle$ to be the linear subspace of $\F_q^n$ spanned by the columns of $H$.}. 
We say $C$ \emph{attains} the kernel pattern $\mathcal{V}$ if there is an invertible matrix $M\in\F^{k\times k}$ such that $\vm_i GA_i=0$ for all $i\in[k]$, where $\vm_i$ is the $i$-th row of $M$. 
\end{defn}
The paper \cite{brakensiek2023generic} formulated a notion of higher order MDS codes called GZP($\ell$). A linear code is GZP($\ell$) if it is MDS and attains all GZPs of order at most $\ell$.
We now formulate a rank-metric counterpart and a strengthening of this notion:
\begin{defn}[$\mathrm{GKP}(\ell)$]\label{olkernel1}
Given a $[n,k]_{\F}$-linear code $C\subseteq\F^n$ with a generator matrix $G\in\F^{k\times n}$, $C$ is said to be $\mathrm{GKP}_q(\ell)$, or simply $\mathrm{GKP}(\ell)$, if $C$ is an MRD code and attains all GKPs of order at most $\ell$. 
\end{defn}
\begin{rmka}\label{rmk:GZP}
GKPs can be seen as a generalization of GZPs. For a GZP $\mathcal{S}=(S_1,\ldots,S_k)$ with $S_i\subseteq [n]$, define $\mathcal{V}=(V_1,\dots,V_k)$, where
\[
V_i=\left\{(v_1,\dots,v_n)\in \F_q^n: v_j=0\text{ for }j\in [n]\setminus S_i\right\}.
\] 
Then $\mathcal{S}$ is a GZP if and only if $\mathcal{V}$ is a GKP. Moreover, an $[n,k]_{\F}$-linear code $C\subseteq\F^n$ over a field $\F/\F_q$ attains $\mathcal{S}$ if and only if it attains $\mathcal{V}$. Consequently, GKP($\ell$) codes are also GZP($\ell$).
\end{rmka}
We also prove an analog of the GM-MDS theorem, which we call the \emph{$\operatorname{GM-MRD}$ theorem}. Roughly speaking, it states that Gabidulin codes attain all GKPs. This theorem is crucial in proving our main result that random Gabidulin codes has the optimal list decodability in the rank metric. A detailed discussion about the GM-MRD theorem is given at the end of this section. 

\subparagraph{$\blacktriangleright\operatorname{MRD}(\ell)$.} 

One characterization of an $[n,k]_{\F}$-linear code $C\subseteq\F^n$ being MDS is that its generator matrix $G\subseteq\F^{k\times n}$ is an $\operatorname{MDS}$ matrix, meaning that any $k$ columns of $G$ are linearly independent.
Strengthening this condition, for any $\ell\geq 1$, Brakensiek et al. \cite{brakensiek2023generic} defined an $\operatorname{MDS}(\ell)$ code to be a linear code with a generator matrix $G\subseteq\F^{k\times n}$ such that for any subsets $S_1,\dots,S_{\ell}\subseteq [n]$, each of size at most $k$, it holds that
\begin{equation}\label{eq:mds-ell-condition}
\dim\left(\bigcap_{i=1}^\ell G_{S_i}\right)=\dim \left(\bigcap_{i=1}^\ell W_{S_i}\right),
\end{equation}
where $W$ is the symbolic matrix $\left(Z_{i,j}\right)_{i\in [k], j\in [n]}$ over the function field in the variables $Z_{1,1},\dots,Z_{k,n}$, and $G_S$ (resp. $W_S$) denotes the span of the columns of $G$ (resp. $W$) with indices in $S$.
Note that by definition, $\operatorname{MDS}(\ell)$ codes are also $\operatorname{MDS}(\ell')$ for $\ell'\leq \ell$, and $\operatorname{MDS}(1)$ codes are just (linear) $\operatorname{MDS}$ codes due to the fact that the symbolic matrix $W$ is an MDS matrix.

To define the rank-metric counterpart of $\operatorname{MDS}(\ell)$, we first express the column span $G_S$ in a linear-algebraic manner for $S\subseteq [n]$ of size at most $k$.
Let $I_S$ denote the $n\times |S|$ matrix that, when restricted to the subset of rows with indices in $S$, becomes the identity matrix, and contains only zeros outside these rows. Then $G I_S$ is precisely the $k\times |S|$ submatrix of $G$ formed by the columns of $G$ with indices in $S$. Thus, $G_S$ is just the column space of $G I_S$.

Suppose $\F$ is an extension field of $\F_q$.
We extend the notion of $\operatorname{MDS}(\ell)$ by replacing the matrix $I_S$ by an arbitrary full-rank matrix $A\in \F_q^{n\times d}$, where $d\leq k$.
For convenience, we introduce the notations $G_A$ and $G_V$ as follows:
Define $G_A:=G A\in\F^{k\times d}$.
Let $V\subseteq\F_q^n$ be the column space of $A$ over $\F_q$, and denote by $G_V\subseteq \F^{k}$ the column space of $G_A$ over $\F$. 
This is well-defined as $G_V$ depends only on $V$, not on $A$. Indeed, $G_V$ equals the $\F$-span of $\sigma_G(V)$, where $\sigma_G: \F^n\to\F^k$ is the linear map $\bv\mapsto G\bv$.

\begin{defn}[$\operatorname{MRD}(\ell)$]
Let $\F$ be an extension field of $
\F_q$. 
We say an $[n,k]_{\F}$-linear code $C\subseteq\F^n$ is $\operatorname{MRD}_q(\ell)$, or simply $\operatorname{MRD}(\ell)$, if for any $\F_q$-linear subspaces $V_1,\dots,V_\ell$ of $\F_q^n$, each of dimension at most $k$, it holds that
\begin{equation}\label{eq:mrd-ell-condition}
\dim_{\F}\left(\bigcap_{i=1}^\ell G_{V_i}\right)=\dim_{\mathbb{K}}\left(\bigcap_{i=1}^\ell W_{V_i}\right),
\end{equation}
where $W$ is the symbolic matrix $\left(Z_{i,j}\right)_{i\in [k], j\in [n]}$ over the function field $\K:=\F_q(Z_{1,1},\dots,Z_{k,n})$.
\end{defn}
By definition, $\operatorname{MRD}(\ell)$ codes are also $\operatorname{MRD}(\ell')$ for $\ell'\leq \ell$. Moreover, it can be shown that $\operatorname{MRD}(1)$ codes are just (linear) $\operatorname{MRD}$ codes due to the fact that the symbolic matrix $W$ has the MRD property. See \cref{eqmrd} and \cref{generic_a}.

Finally, we note that the condition \eqref{eq:mrd-ell-condition} implies \eqref{eq:mds-ell-condition} by choosing the subspaces $V_i$ to be the column space of $I_{S_i}$ over $\F_q$, i.e., the subspace of all vectors $\bv\in\F_q^n$ whose coordinates with indices in $[n]\setminus S_i$ are zero.
Consequently, all $\operatorname{MRD}(\ell)$ codes are also $\operatorname{MDS}(\ell)$.

\subparagraph{$\blacktriangleright\operatorname{LD-MRD}(\ell)$.} 


Recall that a code $C\subseteq\Sigma^n$ is $(\rho,\ell)$-list decodable if for any $\by\in\Sigma^n$, it holds that $|\{\bc\in C: d(\bc,\by)\leq \rho n\}|\leq \ell$. Equivalently, for any $\by\in\Sigma^n$ and $\ell+1$ distinct codewords $\bc_0,\dots,\bc_\ell\in C$, 
\begin{equation}\label{eq:max-distance}
\max_{0\leq i\leq \ell} d(\bc_i, \by) > \rho n.
\end{equation}

The stronger notion of \emph{$(\rho,\ell)$-average-radius list decodability} is defined in the same way, except that we replace the maximum of the distances $d(\bc_i, \by)$ in \eqref{eq:max-distance} by the average of these distances. The formal definition is given as follows.

\begin{defn}[Average-radius list decodability]\label{defn:average}
A code $C\subseteq\Sigma^n$ is $(\rho,\ell)$ average-radius list decodable if for any $\by\in\Sigma^n$ and $\ell+1$ distinct codewords  $\bc_0,\bc_1,\dots,\bc_{\ell}\in C$, it holds that 
\[
\frac{1}{\ell+1}\sum_{i=0}^{\ell}{d(\by,\bc_i)}>\rho n.
\]
\end{defn}

Recall that for an $[n,k]_{\F}$ linear code $C$, the list decoding radius $\rho$ of $C$ satisfies the generalized Singleton bound
$\rho\leq \frac{\ell}{\ell+1}\left(1-\frac{k}{n}\right)$, both in the Hamming metric and the rank metric.
Roth \cite{roth2022higher} first studied $\operatorname{LD-MDS}(\ell)$ codes, which he referred to as strongly-$\left(\frac{\ell(n-k)}{\ell+1}, \ell\right)$-list decodable codes. These are linear codes that meet the generalized Singleton bound under the stricter criterion of average-radius list decodability. Inspired by Roth's work, Brakensiek, Gopi, and Makam further explored $\operatorname{LD-MDS}(\ell)$ codes in their paper \cite{brakensiek2023generic}.

We consider the natural rank-metric counterpart of $\operatorname{LD-MDS}(\ell)$, defined as follows.

\begin{defn}[$\operatorname{LD-MRD}(\ell)$] Let $\F$ be an extension field of $\F_q$.
We say an $[n,k]_{\F}$-linear code $C\subseteq\F^n$ is $\operatorname{LD-MRD}_q(\ell)$, or simply $\operatorname{LD-MRD}(\ell)$, if $C$ is $\left(\frac{\ell}{\ell+1}\left(1-{k}/{n}\right),\ell\right)$-average-radius list decodable in the rank metric. 
In other words, $C$ is $\operatorname{LD-MRD}(\ell)$ if for any $\by\in\F^n$ and $\ell+1$ distinct codewords  $\bc_0,\bc_1,\dots,\bc_{\ell}\in C$, it holds that 
\[
\sum_{i=0}^{\ell}{d_R(\by,\bc_i)}>\ell(n-k),
\]
where $d_R(\by,\bc_i)=\rank_{\F_q}(\by-\bc_i)$ (see \eqref{eq:rank} for the definition of $\rank_{\F_q}$).
We say $C$ is $\ldmrd{\ell}$ if it is $\operatorname{LD-MRD}(\ell')$ for all $\ell'\in [\ell]$. 
\end{defn}

It is straightforward to deduce from the definitions that all 
$\ldmrd{\ell}$ codes are MRD for all $\ell\geq 1$. For a formal proof, see \cref{lem:strengthen}. In addition, all $\operatorname{LD-MRD}(\ell)$ codes are also $\operatorname{LD-MDS}(\ell)$. 
This follows from the fact that the rank distance $d_R(\bx,\by)$ is always bounded by the Hamming distance $d_H(\bx,\by)$.



\paragraph{Equivalence of higher order MRD codes.} 

Similar to the equivalence among the various notions of higher order MDS codes established in \cite{brakensiek2023generic}, we establish the equivalence among the three notions of higher order MRD codes.

\begin{thm}[$1^{st}$ equivalence theorem]\label{1equiv}
For $\ell\ge 1$, a linear code $C$ over $\F/\F_q$ is $\operatorname{GKP}(\ell)$ if and only if it is $\operatorname{MRD}(\ell)$.
\end{thm}

\begin{thm}[$2^{nd}$ equivalence theorem]\label{2equiv} 
For $\ell\ge 1$, a linear code $C$ over $\F/\F_q$ is $\operatorname{MRD}(\ell+1)$ if and only if its dual code $C^{\perp}$ is $\ldmrd{\ell}$.
\end{thm}

Combining the two theorems yields the following corollary.
\begin{cor}\label{cor:three_equiv}
Let $C$ be a linear code over $\F/\F_q$. 
Then the following are all equivalent for $\ell\geq 1$:
\begin{enumerate}
    \item $C$ is $\operatorname{GKP}(\ell+1)$,
    \item $C$ is $\operatorname{MRD}(\ell+1)$, and
    \item $C^{\perp}$ is $\ldmrd{\ell}$.
\end{enumerate}
\end{cor}

\paragraph{Dimension of generic intersections.} 

In the course of proving the equivalence among the three notions of higher order MRD codes, we will establish the following formula for the dimension of the intersection of a collection of subspaces $W_{V_i}$, where $W$ is the symbolic matrix $(Z_{i,j})_{i\in [k], j\in [n]}$. This formula is used in both the proof of \cref{1equiv} and that of \cref{2equiv}.

\begin{thm}[Formula for generic intersection dimension]\label{gintersec}
Let $W$ be the $k\times n$ symbolic matrix $(Z_{i,j})_{i\in [k], j\in [n]}$ over $\F:=\F_q(Z_{1,1},\dots,Z_{k,n})$.
Let $V_1,\dots,V_\ell$ be subspaces of $\F_q^n$, each of dimension at most $k$. Then
\[
\dim_{\mathbb{F}}\left(\bigcap_{i\in[\ell]}W_{V_i}\right)=\max_{P_1\sqcup P_2\sqcup \cdots\sqcup P_s=[\ell]}\left(\sum_{i\in[s]}\dim_{\F_q}\left(\bigcap_{j\in P_i}V_j\right)-(s-1)k\right).
\]
\end{thm}

The above formula generalizes a similar formula proven in \cite{brakensiek2023generic}. More specifically, the formula in \cite{brakensiek2023generic} has the same form as \eqref{intersec1}, except it only considers subspaces of the form $W_S$ for $S\subseteq [n]$ of size at most $k$, i.e., $W_S$ is the span of the columns of $W$ with indices in $S$. As discussed, every such $W_S$ can be realized as some $W_V$, where we choose $V$ to be the subspace of all vectors $\bv\in\F_q^n$ whose coordinates with indices in $[n]\setminus S_i$ are all zero. On the other hand, there are many subspaces $W_V$ that do not come from any $W_S$. Thus, our formula is more general than that in \cite{brakensiek2023generic} and may be of independent interest.

\subsubsection{The GM-MRD Theorem}\label{gmmrdsub}

Similar to the GM-MDS theorem, which states that generic zero patterns are all attained by Reed--Solomon codes, the GM-MRD theorem states that generic kernel patterns are all attained by Gabidulin codes. We present two versions of the GM-MRD theorem. The first one applies to symbolic Gabidulin codes, whereas the second one applies to Gabidulin codes over finite fields. 

\begin{thm}[GM-MRD theorem, informal version of \cref{thm:GM-MRD-v1-new}]\label{gmmrdthm}
Let $1\leq k\leq n$.
Let $\F$ be the function field $\F_q(Z_1,\dots,Z_n)$ be a function field. 
Then the $[n,k]_{\F}$-linear symbolic Gabidulin code $\mG_{n,k}(Z_1,\dots,Z_n)$ defined by the generator matrix $G=\left(Z_j^{q^{i-1}}\right)_{i\in [k], j\in [n]}\in \F^{k\times n}$ is $\mathrm{GKP}(\ell)$ for all $\ell\geq 1$.
\end{thm}

In the detailed version (\cref{thm:GM-MRD-v1-new}) of the above theorem, we also establish a bound on the degree of $\det(M_\mathcal{V})\in\F[Z_1,\dots,Z_n]$, for each generic kernel pattern $\mathcal{V}$, where $M_\mathcal{V}\in\F^{k\times k}$ is an invertible matrix such that $MG$ exhibits the kernel pattern $\mathcal{V}$. By combining this degree bound with the Schwartz--Zippel lemma and the union bound, we derive the finite field GM-MRD theorem.

\begin{thm}[GM-MRD theorem, finite field version, \cref{thm:GM-MRD-finite}]
Let $1\leq k\leq n\leq m$ and $\ell\geq 1$. Let $(\alpha_1,\dots,\alpha_n)$ be uniformly distributed over the set of all vectors in $\F_{q^m}^n$ whose coordinates are linearly independent over $\F_q$. Then with probability at least $1-3kq^{nk\cdot\min\{\ell,k\}+k-m}$, the $[n,k]_{\F_{q^m}}$-linear Gabidulin code $\mG_{n,k}(\alpha_1,\dots, \alpha_n)$ is $\mathrm{GKP}(\ell)$.
\end{thm}



We note that Yildiz and Hassibi previously proved that Gabidulin codes attain all generic zero patterns, thus establishing the GM-MDS theorem for Gabidulin codes (\cite{yildiz2019gabidulin}, see also \cite{yildiz2020support}). However, as discussed, generic zero patterns form only a subset of generic kernel patterns, making our GM-MRD theorem formally stronger than the GM-MDS theorem for Gabidulin codes.

Similarly, while various other GM-MDS theorems have been established so far \cite{lovett2018mds,yildiz2019optimum,brakensiek2023generalized,ron2024efficient}, including the ``ultimate GM-MDS conjecture'' formulated in \cite{brakensiek2023generalized}, these theorems all focus on the attainability of zero patterns rather than kernel patterns. Our GM-MRD theorem is, to our knowledge, the first result that explores the broader category of kernel patterns, which are more naturally connected with rank-metric codes, including MRD codes. Understanding the inner connection between our GM-MRD theorem and the various GM-MDS theorems, and even finding a unifying theory, would be very interesting.

\paragraph{Proof of the GM-MRD theorem and $s$-admissible tuples.}

Our proof of the GM-MRD theorem (\cref{gmmrdthm}) follows the same structure of the proofs of the GM-MDS theorem as presented in \cite{lovett2018mds, yildiz2019optimum}.
The approach involves a careful induction on a collection of $\F_q$-linear spaces $V_1,\dots,V_m$, along with integers $r_1,\dots,r_m\geq 1$, where $r_i$ represents the ``slackness'' in the dimension of $V_i$. The induction step is divided into several cases, enabling us to simplify each instance gradually. Eventually, we arrive at the boundary instances for which the theorem is straightforward to verify.

The proofs in \cite{lovett2018mds, yildiz2019optimum} are very similar, focusing on subsets of the $n$ evaluation points $Z_1,\dots,Z_n$, whereas we consider ``evaluation subspaces'' spanned by linear forms in $Z_1,\dots,Z_n$ over $\F_q$.
In addition, we need to consider $q$-linearized polynomials and their composition in contrast to general polynomials and their multiplication as used in the aforementioned works.

A serious difficulty that arises in the proof is the non-commutative nature of the composition of $q$-linearized polynomials. More specifically, even though our focus is on subspaces $V_i$ of linear forms in $Z_1,\dots,Z_n$, the induction steps require us to also consider subspaces that are the images of $V_i$ under $q$-linearized polynomials $f\in(\F_q[Z_1,\dots,Z_n])[X]$, which may include nonlinear forms in $Z_1,\dots,Z_n$. To capture the structure of the tuples $(V_1,\dots,V_m)$ that may arise, we introduce the notion of \emph{$s$-admissible tuples} of subspaces (\cref{defn:admissible}), where $s$ is a parameter that bounds the dimension of each subspace $V_i$. The family of $s$-admissible tuples is carefully defined to be closed under the induction steps and to exhibit a useful structure that facilitates the proof. For details, we refer the reader to \cref{gmmrdsec}.

\paragraph{Duality of Gabidulin codes.}

The last ingredient we need is the fact that the dual code of a Gabidulin code is again a Gabidulin code (\cref{duali}).

Using the duality, our main theorem---that random Gabidulin codes have the optimal list decodability (\cref{mm1})---follows readily from the equivalence of higher order MRD codes and the GM-MRD theorem.
More specifically, by \cref{cor:three_equiv}, to prove the claim that a random Gabidulin code $C$ over a sufficiently large finite field $\F_{q^m}$ is $\ldmrd{\ell}$ with high probability, we need only demonstrate that its dual code $C^\perp$ is $\operatorname{GKP}(\ell+1)$ with high probability. However, the dual code $C^\perp$ is again a random Gabidulin code. Therefore, the claim easily follows from the finite field GM-MRD theorem (\cref{thm:GM-MRD-finite}).

\section{Preliminaries}\label{Prelim}
We first introduce the basic notation. 
Define $\N=\{0,1,\dots\}$, $\N^+=\{1,2,\dots\}$, and $[n]=\{1,\ldots,n\}$.
Denote by $\F_q$ the finite field of size $q$. We use $\F$ to denote a field that is an extension of $\F_q$. Let $\overline{\F}_q$ be the algebraic closure of $\F_q$. 

All vectors are column vectors unless stated otherwise.
For $v_1,\ldots,v_n$ in a vector space $V$ over $\F$, the $\F$-subspace $V$ spanned by $v_1,\ldots,v_n$ is denoted by $\spa_{\F}\{v_1,\ldots,v_n\}$, or simply $\spa\{v_1,\ldots,v_n\}$ if $\F$ is clear from the context. The set of $m\times n$ matrices over $\F$ is denoted by $\F^{m\times n}$. The sum of two linear subspaces $V_1,V_2\subseteq\F^n$ is $V_1+V_2:=\{\bv_1+\bv_i:\bv_1\in V_1,\bv_2\in V_2\}$. If $V_1\cap V_2=\{0\}$, we also write $V_1\oplus V_2$ instead of $V_1+V_2$. 
A \emph{complement} of a subspace $V_1\subseteq V$ in $V$ is another subspace $V_2\subseteq V$ such that $V=V_1\oplus V_2$. 
Denote by $V^\perp\subseteq\F^n$ the dual space of $V\subseteq\F^n$ with respect to the inner product, i.e., $V^\perp=\{\bv\in\F^n: \bv^T\bu=0, \forall \bu\in V\}$.
Note that $\dim V^\perp=n-\dim V$ and $(V^\perp)^\perp=V$. 

We need the following two lemmas about vector spaces.
\begin{lemma}\label{lm:dualspace}
Let $V_1,\ldots,V_\ell$ be subspaces of $\F^n$. Then $\left(\bigcap_{i=1}^{\ell}V_i\right)^{\perp}=\sum_{i=1}^{\ell} V_i^{\perp}$.
\end{lemma}
\begin{proof} 
First, note that $V_i^\perp \subseteq \left(\bigcap_{i=1}^{\ell}V_i\right)^{\perp}$ for all $i\in [\ell]$. So $\sum_{i=1}^{\ell} V_i^{\perp}\subseteq \left(\bigcap_{i=1}^{\ell}V_i\right)^{\perp}$.
Conversely, to prove that $\sum_{i=1}^{\ell} V_i^{\perp}\supseteq \left(\bigcap_{i=1}^{\ell}V_i\right)^{\perp}$, we will show the equivalent statement that
$\left(\sum_{i=1}^{\ell} V_i^{\perp}\right)^{\perp}\subseteq \bigcap_{i=1}^{\ell}V_i$. Let $\bu\in \left(\sum_{i=1}^{\ell} V_i^{\perp}\right)^{\perp}$. Then for all $i\in [\ell]$, we have $\bu\in (V_i^\perp)^\perp=V_i$. So $\bu\in \bigcap_{i=1}^{\ell}V_i$. It follows that $\left(\sum_{i=1}^{\ell} V_i^{\perp}\right)^{\perp}\subseteq \bigcap_{i=1}^{\ell}V_i$.
\end{proof}

\begin{lemma}\label{lm:dim}
Let $A_1,\ldots,A_m,B_1,\ldots,B_m$ be subspaces of $\F^n$. Then 
\[
\dim\left(\bigcap_{i=1}^m (A_i+B_i)\right)\leq \dim\left(\bigcap_{i=1}^m A_i\right)+\sum_{i=1}^m \dim(B_i).
\]
\end{lemma}
\begin{proof}
Note that the inequality trivially holds if $B_i=\{0\}$ for all $i\in [m]$. 
Then, we replace $\{0\}$ by $B_i$ for $i=1,\dots,m$ and note that the replacement increases RHS by $\dim B_i$ and LHS by at most $\dim B_i$. The claim follows.
\end{proof}

\paragraph{$q$-linearized polynomials.}

A polynomial of the form $f(X)=\sum_{i=0}^{\ell}a_iX^{q^i}$ with coefficients $a_i$ in a commutative algebra $A$ over $\F_q$ is called a \emph{$q$-linearized} polynomial over $A$. 
It has the property that $f(a+b)=f(a)+f(b)$ and $f(ca)=cf(a)$ for $a,b\in A$ and $c\in\F_q$.
The $q$-degree of $f(X)$, denoted by $\deg_q(f)$, is defined to be the largest integer $\ell\geq 0$ such that $a_{\ell}\neq 0$. Throughout this paper, $f$ is always a $q$-linearized polynomial unless stated otherwise. We often view an extension field $\F$ of $\F_q$ as a $\F_q$-linear space. For any $\F_q$-subspace $V\subseteq \F$, define $f(V):=\{f(\alpha):\alpha\in V\}$, which is a $\F_q$-linear subspace of $\F$ when $f$ is a $q$-linearized polynomial. 

For $\alpha_1,\dots,\alpha_m\in \F$, where $\F$ is an extension field of $\F_q$, the matrix 
\[
M_{\alpha_1,\dots,\alpha_m}=\left(\alpha_i^{q^{j-1}}\right)_{i\in [m], j\in [m]}\in\F^{m\times m}
\]
is called the \emph{Moore matrix} associated with $\alpha_1,\dots,\alpha_m$. The following lemma characterizes its nonsingularity. For a proof, see, e.g., \cite[Lemma~1.3.3]{Goss}.

\begin{lemma}\label{lem:moore}
$\det(M_{\alpha_1,\dots,\alpha_m})\neq 0$ if and only if $\alpha_1,\ldots,\alpha_m$ are linearly independent over $\F_q$.
\end{lemma}

We also need the following lemma regarding polynomials that vanish precisely on an $\F_q$-linear subspace of $\F$.

\begin{lemma}\label{prop:q_linear}
Let $\F$ be an extension field of $\F_q$ and let $V\subseteq \F$ be an $\F_q$-linear subspace. Then $f=\prod_{\alpha\in V}(X-\alpha)\in\F[X]$ is a $q$-linearized polynomial.
\end{lemma}

\begin{proof} 
Let $m=\dim V$ and let $\alpha_1,\dots,\alpha_m$ form a basis of $V$. Define the matrix
\[
M=\begin{pmatrix}
\alpha_1 & \alpha_1^q & \cdots & \alpha_1^{q^{m-1}} & \alpha_1^{q^m}\\
\alpha_2 & \alpha_2^q & \cdots & \alpha_2^{q^{m-1}} & \alpha_2^{q^m}\\
\vdots & \vdots & \ddots & \vdots & \vdots\\
\alpha_m & \alpha_m^q & \cdots & \alpha_m^{q^{m-1}} & \alpha_m^{q^m}\\
X & X^q & \cdots & X^{q^{m-1}} & X^{q^m}
\end{pmatrix}\in  \F(X)^{(m+1)\times (m+1)}.
\]
Note that the coefficient of $X^{q^m}$ in $\det(M)$ is $\det(M_{\alpha_1,\dots,\alpha_m})$, which is nonzero by \cref{lem:moore}.
So $\det(M)$ is a $q$-linearized polynomial of degree $q^m$ with the leading coefficient $\det(M_{\alpha_1,\dots,\alpha_m})$. It vanishes at any $\F_q$-linear combination of $\alpha_1,\dots,\alpha_m$, i.e., it vanishes on the set $V$.
So $g:=\det(M_{\alpha_1,\dots,\alpha_m})^{-1} \cdot \det(M)$ is a monic $q$-linearized polynomial vanishing on $V$. But $f$ is another monic polynomial that vanishes on $V$. So $f-g$ is a polynomial over $\F$ of degree less that $q^m=|V|$ that vanishes on $V$. This implies $f-g=0$, i.e., $f=g$. So $f$ is a $q$-linearized polynomial.
\end{proof}
 
For a $k\times n$ matrix $G$ over an extension field $\F$ of $\F_q$ and an $n\times \ell$ matrix $A$ over $\F_q$, define $G_A:=GA\in\F^{k\times \ell}$. We denote by $\langle A \rangle$ the $\F_q$-subspace of $\F_q^n$ spanned by the columns of $A$. 
For an $\F_q$-subspace $V\subseteq \F_q^n$ and $A\in \F_q^{n\times \dim V}$ such that $V=\langle A\rangle$, define $G_V\subseteq \F^k$ to be the column span of $G_A$ over $\F$.

\paragraph{Linear codes.} 
Let $\F$ be a field.
An \emph{$[n,k]_\F$ linear code} (or \emph{$[n,k]_\F$ code} for short) is simply a subspace of $\F^n$ of dimension $k$. 
The \emph{dual code} of an $[n,k]_\F$ code $C$ is the $[n,n-k]_\F$ code $C^\perp$.

For an $[n,k]_\F$ code $C$, a matrix $G\in \F^{k\times n}$ is said to be a \emph{generator matrix} of $C$ if $C=\{G^T \bu: u\in \F^k\}$, and a matrix $H\in \F^{(n-k)\times n}$ is said to be a \emph{parity check matrix} of $C$ if $C=\{\bv\in \F^n: H\bv=0\}$. A generator matrix of $C$ is also a parity check matrix of the dual code $C^\perp$. Similarly, a parity check matrix of $C$ is also a generator matrix of $C^\perp$.


\subsection{Rank-Metric Codes and Gabidulin Codes}
We first review some basic facts and results about rank-metric codes. The \emph{rank distance} $d(A,B)$ between two matrices $A, B\in \F_q^{m\times n}$ is defined to be the rank of $A-B$, i.e., $d(A,B):=\rank(A-B)$. Indeed, this defines a distance \cite{G85}. A rank-metric code $C$ is a subset of $\F_q^{m\times n}$
whose rate and minimum distance are given by
\[R(C) := \frac{\log_q\abs{C}}{nm}\quad \text{ and }\quad d(C) := \min_{\substack{A, B\in C\\ A\neq B}}d(A,B).\]
Without loss of generality, we always assume that $m\geq n$, since otherwise we can exchange $n$ and $m$. 
It is convenient to treat an $m\times n$ matrix $A$ over $\F_q$ as a vector $\bv=(v_1,\ldots,v_n)\in \F_{q^m}^n$ by identifying $\F_q^m$ with $\F_{q^m}$ (by fixing a basis of $\F_{q^m}$) and viewing each column of $A$ as an element in $\F_{q^m}$. Then, we have $\rank(A)=\dim_{\F_q}(\spa_{\F_q}\{v_1,\ldots,v_n\})$. In this way, a rank-metric code $C$ may be viewed as a subset of $\F_{q^m}^n$, and we can study linear rank-metric codes, i.e, codes that are $\F_{q^m}$-subspaces.

\paragraph{Linear rank-metric codes over a general field $\F/\F_q$.}

It is convenient for us to consider a general notion of linear rank-metric codes $C\subseteq \F^n$ over a field $\F/\F_q$ that can even be infinite. To properly define this notion, we first define the $\F_q$-rank and the kernel subspace of a vector $\bv\in \F^n$.

\begin{defn}[$\F_q$-rank]
Let $\F$ be an extension field of $\F_q$.
For $\bv=(v_1,\dots,v_n)\in \F^n$, define 
\[
\rank_{\F_q}(\bv):=\dim_{\F_q}(\spa_{\F_q}\{v_1,\ldots,v_n\}),
\]
called the \emph{$\F_q$-rank} of $\bv$.
\end{defn}

\begin{defn}[Kernel subspace]
For $\bv=(v_1,\dots,v_n)\in \F^n$, define the \emph{$\F_q$-kernel subspace} (or simply the \emph{kernel subspace}) of $\bv$ to be
\[
\ker_{\F_q}(\bv):=\left\{\bu\in \F_q^n: \bu^T \bv=0\right\}=\left\{(u_1,\dots,u_n)\in \F_q^n: \sum_{i=1}^n u_i v_i=0\right\}.
\]
\end{defn}

The following lemma can be seen as an alternative definition of the $\F_q$-rank.

\begin{lemma} \label{lem:rank-kernel}
$\rank_{\F_q}(\bv)=n-\dim_{\F_q}(\ker_{\F_q}(\bv))$.
\end{lemma}

\begin{proof}
Consider the $\F_q$-linear map $\F_q^n\to \F$ sending $\bu\in\F_q^n$ to $\bu^T \bv$. The image of this map is $\spa_{\F_q}\{v_1,\ldots,v_n\}$, whose dimension is $\rank_{\F_q}(\bv)$ by definition. The kernel of this map is $\ker_{\F_q}(\bv)$. So $\rank_{\F_q}(\bv)=n-\dim_{\F_q}(\ker_{\F_q}(\bv))$. 
\end{proof}



We can now define the notion of a linear rank-metric code over a field $\F/\F_q$.

\begin{defn}[Linear rank-metric code]
Let $\F$ be an extension field of $\F_q$.
An \emph{$[n,k]_{\F}$ $($linear$)$ rank-metric code} is simply an $[n,k]_{\F}$ code $C\subseteq \F^n$ equipped with the distance function $d_R: \F^n\times \F^n\to\N$ defined by $d_R(\bv,\bv'):=\rank_{\F_q}(\bv-\bv')$.
The \emph{minimum distance} of $C$ is 
\[
d(C):=\min_{\substack{\bv, \bv'\in C\\ \bv\neq \bv'}}d_R(\bv, \bv')=\min_{\mathbf{0}\neq \bv\in C} \rank_{\F_q}(\bv).
\]
\end{defn}

In analogy with the classical scenario, one can prove the Singleton bound for linear rank-metric codes.

\begin{thm}[Singleton bound]
Let $C$ be an $[n,k]_{\F}$ rank-metric code. Then $d(C)\leq n-k+1$.\footnote{We remark that when $\F=\F_{q^m}$, there exists a Singleton bound, $|C|\leq q^{m(n-d+1)}$, that also applies to nonlinear rank-metric codes $C\subseteq \F^n$ \cite{G85}. However, this bound is given in terms of the size of the code, not the dimension, making it inapplicable when $\F$ is infinite.}
\end{thm}
\begin{proof}
There exists a nonzero codeword $\bv=(v_1,\dots,v_n)\in C$ whose first $k-1$ coordinates are zero. So $d(C)\leq \rank_{\F_q}(\bv)=\dim_{\F_q}(\spa_{\F_q}\{v_1,\ldots,v_n\})=\dim_{\F_q}(\spa_{\F_q}\{v_k,\ldots,v_n\})\leq n-k+1$.
\end{proof}

$[n,k]_{\F}$ rank-metric codes $C$ that attain the Singleton bound (i.e., $d(C)=n-k+1$) are called MRD codes. The next lemma gives an alternative characterization of MRD codes, which is crucial in this paper. 


\begin{lemma}\label{eqmrd}
Let $G\in \F^{k\times n}$ be a generator matrix of an $[n,k]_{\F}$ code $C$. Then the following are all equivalent:
\begin{enumerate}
\item\label{item:mrd-1} $C$ is $\operatorname{MRD}$.
\item\label{item:mrd-2} For any  $A\in\F_q^{n\times k}$ of full rank, the matrix $GA\in \F^{k\times k}$ also has full rank.
\item\label{item:mrd-3} For any $k'\leq k$ and any $A'\in\F_q^{n\times k'}$ of full rank, the matrix $GA'\in \F^{k\times k'}$ also has full rank.
\end{enumerate}
\end{lemma}

\begin{proof}
\cref{item:mrd-3} obviously implies \cref{item:mrd-2}. To see the converse, for any $A'\in\F_q^{n\times k'}$ of full rank, where $k'\leq k$, we may extend $A'$ to a matrix $A\in \F_q^{n\times k}$ of full rank. As the columns of $GA'$ form a subset of the columns of $GA$, if $GA$ has full rank (i.e., its columns are linearly independent over $\F$), then so does $GA'$.
So \cref{item:mrd-2} implies \cref{item:mrd-3}.

To see that \cref{item:mrd-2} implies \cref{item:mrd-1},
assume to the contrary that $C$ is not MRD, i.e., there exists a nonzero codeword $\bv\in C\subseteq\F^n$ such that $\rank_{\F_q}(\bv)\leq n-k$. Then, there exists $A\in\F_q^{n\times k}$ of rank $k$ such that $\bv^T A=0$. 
The row vector $\bv^T$ can be written as $\bx^T G$ for some nonzero $\bx\in \F^k$. Then $\bx^T G A=0$, implying that $GA$ does not have full rank.

Conversely, to see that \cref{item:mrd-1} implies \cref{item:mrd-2}, assume that $GA$ does not have full rank for some $A\in\F_q^{n\times k}$ of full rank.
Then there exists nonzero $\bx\in \F^k$ such that $\bx^T GA=0$. Let $\bv=G^T\bx=(\bx^T G)^T\in\F^n$, which is a nonzero codeword. Then $\bv^T A=0$. As $A$ has rank $k$, we have $\rank_{\F_q}(\bv)\leq n-k$ and hence $d(C)\leq n-k$. So $C$ is not MRD.
\end{proof}

%
%


\begin{lemma}\label{prop:dualmrd}
Let $C$ be an $[n,k]_{\F}$ code. If $C$ is $\operatorname{MRD}$, then $C^\perp$ is also $\operatorname{MRD}$.
\end{lemma}

\begin{proof}
By definition, $C^\perp$ is an $[n,n-k]_{\F}$ code. Let $G\in\F^{k\times n}$ be a generator matrix of $C$, which is also a parity check matrix of $C^{\perp}$. Assume to the contrary that $C^{\perp}$ is not MRD. Then there exists nonzero $\bv=(v_1,\dots,v_n)\in C^{\perp}$ such that $\rank_{\F_q}(\bv)=r\leq k$. 
By definition, $\dim_{\F_q}(\spa_{\F_q}\{v_1,\ldots,v_n\})=r$.
Pick $z_1,\dots,z_r\in \F$ that form a basis of $\spa_{\F_q}\{v_1,\ldots,v_n\}$ over $\F_q$, and let $\bz=(z_1,\dots,z_r)\in \F^r$.
Then $\bv=A\bz$ for some $A\in\F_q^{n\times r}$ of rank $r$.
As $\bv\neq 0$, we have $\bz\neq 0$.
Finally, as $\bv\in C^\perp$, we have $0=G\bv=GA \bz$. So $GA\in \F^{k\times r}$ does not have full rank. By \cref{eqmrd}, $C$ is not MRD.
\end{proof}

We also record the basic fact that the properties of $\operatorname{GKP}(\ell)$, $\operatorname{MRD}(\ell)$, and $\ldmrd{\ell}$ all imply the $\operatorname{MRD}$ property.
\begin{lemma}\label{lem:strengthen}
Let $\ell\geq 1$. A code that is $\operatorname{GKP}(\ell)$, $\operatorname{MRD}(\ell)$, or $\ldmrd{\ell}$ is also $\operatorname{MRD}$.
\end{lemma}

\begin{proof}
By definition, we only need to verify that an $\operatorname{LD-MRD}(1)$ code is $\operatorname{MRD}$. Suppose $C$ is an $[n,k]_{\F}$ code that is $\operatorname{LD-MRD}(1)$. Then $C$ is $\left(\rho,1\right)$-average radius list decodable in the rank metric, where $\rho=\frac{1}{2}(1-k/n)$. Let $\bv, \bv'\in C$ such that $d_R(\bv,\bv')=d(C)$, and let $\by=\bv$. Then $d(C)=d_R(\bv,\bv')=d_R(\by, \bv)+d_R(\by,\bv')>2n\rho=n-k$. So $d(C)\geq n-k+1$, i.e., $C$ is MRD.
\end{proof}
%
%


\paragraph{Gabidulin codes.}

The most famous MRD codes are Gabidulin codes, which are defined by using the evaluation of linearized polynomials. We briefly review the construction of Gabidulin codes \cite{G85} and extend it to a general field $\F/\F_q$.
\begin{defn}[Gabidulin code over $\F$]
Let $0<k\leq n$ be integers.
Let $\F$ be an extension field of $\F_q$ such that $[\F:\F_q]\geq n$.
Let $\alpha_1,\dotsc, \alpha_n\in\F$ be linearly independent over $\F_q$. Define
\[
\mG_{n,k}(\alpha_1,\dotsc, \alpha_n) := \left\{\bx_f=(f(\alpha_1),\ldots,f(\alpha_n)):\; f\in\F[X]~\text{is $q$-linearized and}~\deg_q(f)\leq k-1\right\}.
\]
which is an $[n,k]_\F$ rank-metric code of minimum distance $\min_{f\neq 0,\deg_q(f)<k}\rank_{\F_q}(\bx_f)$.
\end{defn}

For a nonzero codeword $\bx_f=(f(\alpha_1),\ldots,f(\alpha_n))\in \mG_{n,k}(\alpha_1,\dotsc, \alpha_n)$, using the fact that $f$ is $q$-linearized, we have 
\[
\ker_{\F_q}(\bx_f)
=\left\{(u_1,\dots,u_n)\in\F_q^n: \left(\sum_{i=1}^n u_i f(\alpha_i)\right)=0\right\}
=\left\{(u_1,\dots,u_n)\in\F_q^n: f\left(\sum_{i=1}^n u_i \alpha_i\right)=0\right\}
\]
whose dimension over $\F_q$ is bounded by $k-1$ since $\alpha_1,\dots,\alpha_n$ are linearly independent over $\F_q$ and $f$ has at most $\deg(f)\leq q^{k-1}$ roots. So $\rank_{\F_q}(\bx_f)\geq n-k+1$ by \cref{lem:rank-kernel}. This shows that Gabidulin codes are MRD codes.

\subsection{The MRD Property of the Generic Linear Code}

We now show that the symbolic matrix $W=(Z_{i,j})_{i\in [k], j\in [n]}$ with variables $Z_{i,j}$, which can be seen as a generator matrix of the $[n,k]_\F$ ``generic'' linear code over $\F=\F_q(Z_{1,1},\dots,Z_{k,n})$, has the MRD property. To prove this statement, we need a tool known as the \emph{Cauchy--Binet formula}.


\begin{fa}[Cauchy--Binet formula (see, e.g., \cite{tao2023topics})]
Let $\F$ be a field and $n\ge r$. Let $A$ be an $r\times n$ matrix and $B$ be an $n\times r$ matrix over $\F$. For a subset $S\subseteq[m]$ of size $r$, denote by $A_S$ the $r\times r$ submatrix of $A$ whose columns are selected by $S$, and similarly, denote by $B^S$ the $r\times r$ submatrix of $B$ whose rows are selected by $S$. Then,
\[
\det \left(AB\right)=\sum_{\substack{S\subseteq[n]\\|S|=r}}\det(A_S)\det(B^S).
\]
\end{fa}
\begin{lemma}\label{generic_a}
Let $r\leq k\leq n$ be positive integers.
Let $W$ be the $k\times n$ matrix $(Z_{i,j})_{i\in [k], j\in [n]}$ over the function field $\F_q(Z_{1,1},\dots, Z_{k,n})$. 
Let $A$ be an $n\times r$ matrix of full rank over $\F_q$. Then $WA\in \F_q(Z_{1,1},\dots, Z_{k,n})^{k\times r}$ has the full rank of $r$. 
\end{lemma}
\begin{proof}
By replacing $W$ with its top $r\times n$ submatrix, we may assume $k=r$.
By the Cauchy--Binet formula, 
\[\det \left(W A\right)=\sum_{\substack{S\subseteq [n]\\|S|=r}}{\det(W_S)\det(A^S)}.\]
Since $A$ is a full rank matrix over $\F_q$, there exists $S\subseteq [n]$ of size $r$ such that $\det(A^S)\neq 0$. 
Note that $\det(W_S)\neq 0$, which follows by expanding $\det(W_S)$ into a linear combination of monomials and using the fact that the entries of $W$ are distinct variables $Z_{i,j}$.
So $\det(W_S)\det(A^S)\neq 0$.
For any $S'\subseteq [n]$ of size $r$, every monomial of $\det(W_{S'})\det(A^{S'})$ has the form $Z_{1,i_1}\cdots Z_{r, i_r}$ where $\{i_1,\dots,i_r\}=S'$.
So the monomials of $\det(W_S)\det(A^S)$ are different from those of $\det(W_{S'})\det(A^{S'})$ when $S\neq S'$.
It follows that $\det(WA)\neq 0$, i.e., $W A$ has full rank. 
\end{proof}

\section{Generic Kernel Patterns}
In this section, we develop a structural theory regarding the concept of a \emph{generic kernel pattern} of order $\ell$ over a finite field $\mathbb{F}_q$. The notion of generic kernel patterns that we introduce can be viewed as a natural generalization of \emph{generic zero patterns} studied in \cite{dau2014existence, lovett2018mds, yildiz2019optimum, brakensiek2023generic}.

Recall that a tuple $\mathcal{V}=(V_1,V_2,\ldots,V_k)$, where each $V_i$ is a linear subspace of $\F_q^n$, is called a generic kernel pattern over $\F_q$ if for every nonempty set $\Omega\subseteq [k]$, we have
\[
\dim\left(\bigcap_{i\in \Omega} V_i\right)\leq k-|\Omega|.
\]
Furthermore, $\mathcal{V}$ is said to be \emph{of order $\ell$} if there are exactly $\ell$ distinct nonzero subspaces among $V_1,V_2,\dots,V_k$. 

We will now establish a general framework to precisely characterize generic kernel patterns of order at most $\ell$, which play a crucial role in understanding the list decodability of Gabidulin codes. To achieve this, we will develop a linear-algebraic analog of the \emph{generalized Hall's theorem} as proved by Brakensiek, Gopi, and Makam in \cite{brakensiek2023generic}.

We remark that, although the statements and proofs in this section are presented over $\F_q$, sufficient for our applications, they remain valid over a general field $\F$.

\subsection{Generalized Hall's Theorem for Vector Spaces}

We start by proving the following theorem, which can be seen as a linear-algebraic analog of \cite[Theorem~2]{dau2014existence}.


\begin{thm}[Hall’s theorem for vector spaces]\label{hall1}
Let $n$ and $k$ be integers with $n\ge k\geq 0$. Let $V_1,\dots,V_k$ be subspaces of $\F_q^n$. Suppose 
\begin{equation}\label{1}
\dim \left(\sum_{i\in \Omega}V_i\right)\ge n-k+|\Omega|
\end{equation}
holds for all nonempty $\Omega\subseteq [k]$.  
Then there exist subspaces $V'_i\subseteq V_i$, $i=1,\dots,k$, such that
\begin{enumerate}
\item \label{item:hall-1} $\dim \left(\sum_{i\in \Omega}V'_i\right)\ge n-k+|\Omega|$ for all nonempty $\Omega\subseteq [k]$, and
\item \label{item:hall-2} $\dim V'_i=n-k+1$ for all $i\in[k]$.
\end{enumerate}
\end{thm}

\begin{proof}
Our proof closely follows that of \cite[Theorem~2]{dau2014existence}.
First, if there exist subspaces $\widetilde{V}_1\subseteq V_1,\dots, \widetilde{V}_k\subseteq V_k$ such that \eqref{1} holds for $\widetilde{V}_1,\dots, \widetilde{V}_k$ and at least one subspace $\widetilde{V}_i$ is a proper subspace of $V_i$, then we may replace $V_i$ by $\widetilde{V}_i$ and prove the claim for the new subspaces $\widetilde{V}_1,\dots, \widetilde{V}_k$, which would imply the claim for the original subspaces $V_1,\dots,V_k$. 
By repeatedly doing this, we may assume that $(V_1,\dots,V_k)$ is minimal (with respect to component-wise inclusion) subject to the condition \eqref{1}.

We will show that for the minimal $(V_1,\dots,V_k)$ satisfying \eqref{1}, the claim holds by choosing $V_i'=V_i$ for $i\in [k]$.
Note that \cref{item:hall-1} is just \eqref{1} as $V_i'=V_i$.
Assume to the contrary that \cref{item:hall-2} does not hold. Then there exists $r\in [k]$ such that
\begin{equation}\label{eq:Vr}
\dim V_r\ge n-k+2\geq 2.
\end{equation}
This is because we have $\dim V_i \geq n-k+1$ for all $i\in [k]$ by \cref{item:hall-1} and, as \cref{item:hall-2} does not hold, the equality is not attained for some $i\in [k]$.

Pick linearly independent $a,b\in V_r$ and let $V_{a,b}$ be a complement of $\mathrm{span}\{a,b\}$ in $V_r$, i.e., $V_r=V_{a,b}\oplus\mathrm{span}\{a,b\}$.
For $i\in [k]$, define
\[
V_i^a= \begin{cases} V_{a,b}\oplus \mathrm{span}\{a\}, & \text { if } i=r, \\ V_i, & \text { otherwise, }\end{cases}
\quad\text{and}\quad
V_i^b= \begin{cases}V_{a,b}\oplus \mathrm{span}\{b\}, & \text { if } i=r, \\ V_i, & \text { otherwise. }\end{cases}
\]
By the minimality of $(V_1,\dots,V_k)$, both $(V_1^a,\dots,V_k^a)$ and $(V_1^b,\dots,V_k^b)$ violate \eqref{1}. 
Moreover, \eqref{1} is violated only for nonempty sets $\Omega\subseteq [k]$ that contain $r$ since $V_i^a$ and $V_i^b$ agree with $V_i$ if $i\neq r$. 
Therefore, there exist $A, B\subseteq[k]$ not containing $r$ such that
\[
\dim\left(\sum_{i\in A\cup\{r\}} V_i^a \right)<n-k+|A|+1\quad\text{and}\quad\dim\left(\sum_{i\in B\cup\{r\}} V_i^b\right)<n-k+|B|+1.
\]
On the other hand, note that $\sum_{i\in A\cup\{r\}}V_i^a\supseteq \sum_{i\in A }V_i$ and $\dim(\sum_{i\in A }V_i)\geq n-k+|A|$.
So we must have 
\begin{equation}\label{eq:via}
\dim\left(\sum_{i\in A\cup\{r\}}V_i^a\right)=n-k+|A| \quad\text{and}\quad \sum_{i\in A\cup\{r\}}V_i^a=\sum_{i\in A}V_i.
\end{equation}
Similarly,
\begin{equation}\label{eq:vib}
\dim\left(\sum_{i\in B\cup\{r\}}V_i^b\right)=n-k+|B| \quad\text{and}\quad \sum_{i\in B\cup\{r\}}V_i^b=\sum_{i\in B}V_i.
\end{equation}
It follows that
\begin{equation}\label{ff1}
    \left(\sum_{i\in A\cup\{r\}}V_i^a\right)\cap\left(\sum_{i\in B\cup\{r\}}V_i^b\right)=\left(\sum_{i\in A }V_i\right) \cap \left(\sum_{i\in B }V_i\right).
\end{equation}
As $V_r=V_r^a+V_r^b$, we also have
\begin{equation}\label{ff2}\left(\sum_{i\in A\cup\{r\}}V_i^a\right) + \left(\sum_{i\in B\cup\{r\}}V_i^b\right)=\sum_{i\in A\cup B\cup \{r\} }V_i.
\end{equation}
Then
\begin{equation}\label{ff3}
\begin{aligned}
&2(n-k)+|A|+|B|\\
&\stackrel{\eqref{eq:via},\eqref{eq:vib}}{=}\dim\left(\sum_{i\in A\cup\{r\}}V_i^a\right)+\dim\left(\sum_{i\in B\cup\{r\}}V_i^b\right)\\
&=\dim\left(\left(\sum_{i\in A\cup\{r\}}V_i^a\right)+\left(\sum_{i\in B\cup\{r\}}V_i^b\right)\right)+\dim\left(\left(\sum_{i\in A\cup\{r\}}V_i^a\right)\cap\left(\sum_{i\in B\cup\{r\}}V_i^b\right)\right)\\
&\stackrel{\eqref{ff1},\eqref{ff2}}{=}\dim\left(\sum_{i\in A\cup B\cup \{r\} }V_i\right)+\dim\left(\left(\sum_{i\in A }V_i\right) \cap \left(\sum_{i\in B }V_i\right)\right).
\end{aligned}
\end{equation}
By \eqref{1}, we have 
\begin{equation}\label{eq:sum-of-vi}
\dim\left(\sum_{i\in A\cup B\cup \{r\} }V_i\right)\geq n-k+|A\cup B|+1.
\end{equation} 
We claim that
\begin{equation}\label{ff4}
\dim\left(\left(\sum_{i\in A}V_i\right) \cap \left(\sum_{i\in B}V_i\right)\right)\geq n-k+|A\cap B|
\end{equation}
which will be proved shortly.
Then
\begin{align*}
2(n-k)+|A|+|B|&\stackrel{\eqref{ff3}}{=}
\dim\left(\sum_{i\in A\cup B\cup \{r\} }V_i\right)+\dim\left(\left(\sum_{i\in A}V_i\right) \cap \left(\sum_{i\in B}V_i\right)\right)\\
&\stackrel{\eqref{eq:sum-of-vi},\eqref{ff4}}{\geq} 2(n-k)+|A|+|B|+1,
\end{align*}
which is impossible. So \cref{item:hall-2} holds.

It remains to prove $\eqref{ff4}$. In the case where $A\cap B\neq\emptyset$, this follows immediately from $\eqref{1}$ and the fact that $\dim\left(\left(\sum_{i\in A}V_i\right) \cap \left(\sum_{i\in B}V_i\right)\right)\geq\dim\left(\sum_{i\in A\cap B }V_i\right)$. Now assume $A \cap B=\emptyset$.
By \eqref{eq:via} and \eqref{eq:vib}, we have
$V_r^a\subseteq \sum_{i\in A}V_i$ and
$V_r^b\subseteq \sum_{i\in B}V_i$. This implies 
\[
\dim\left(\left(\sum_{i\in A}V_i\right) \cap \left(\sum_{i\in B}V_i\right)\right)\geq \dim\left(V_r^a\cap V_r^b\right)=\dim V_{a,b}=\dim V_r-2\stackrel{\eqref{eq:Vr}}{\geq} n-k=n-k+|A\cap B|,
\]
which proves \eqref{ff4}.
\end{proof}

Instead of directly using \cref{hall1}, we will in fact use its dual version, which is stated as follows.

\begin{cor}[Hall’s theorem for vector spaces, the dual version]\label{halldual1}
Let $n$ and $k$ be integers with $n\ge k\geq 0$. Let $V_1,\dots,V_k$ be subspaces of $\F_q^n$. 
Suppose
\begin{equation}\label{111}
\dim \left(\bigcap_{i\in \Omega}V_i\right)\leq k-|\Omega|
\end{equation}
holds for all nonempty $\Omega\subseteq [k]$. Then there exist subspaces $V'_i$ of $\F_q^n$ containing $V_i$, $i=1,\dots,k$, such that
\begin{enumerate}
\item \label{item:dual-1} $\dim \left(\bigcap_{i\in\Omega}V'_i\right)\leq k-|\Omega|$ for all nonempty $\Omega\subseteq [k]$, and
\item \label{item:dual-2} $\dim V'_i=k-1$ for all $i\in[k]$.
\end{enumerate}
\end{cor}
\begin{proof}
Let $H_i=V_i^{\perp}\subseteq\F_q^n$ for $i\in [k]$.
For any nonempty $\Omega\subseteq [k]$, we have
$\left(\bigcap_{i\in \Omega}V_i\right)^{\perp}=\sum_{i\in \Omega}H_i$ by \cref{lm:dualspace}.
So \eqref{111} is equivalent to $\dim\left(\sum_{i\in \Omega}H_i\right)\geq n-k+|\Omega|$. 
Applying \cref{hall1} to $H_1,\dots,H_k$, we obtain subspaces $H_i'\subseteq H_i$ such that 
$\dim\left(\sum_{i\in\Omega}H'_i\right)\geq n-k+|\Omega|$ for all nonempty $\Omega\subseteq [k]$ and $\dim H'_i=n-k+1$ for all $i\in[k]$.

Let $V_i'=(H_i')^{\perp}$ for $i\in [k]$, so that $V_i'\supseteq V_i$.
For every nonempty $\Omega\subseteq [k]$, we have 
$\left(\bigcap_{i\in \Omega}V_i'\right)^{\perp}=\sum_{i\in \Omega}H_i'$ by \cref{lm:dualspace}. So $\dim \left(\bigcap_{i\in\Omega}V'_i\right)=n-\dim \left(\sum_{i\in\Omega}H'_i\right) \leq k-|\Omega|$, proving \cref{item:dual-1}.
And for $i\in [k]$, we have $\dim V'_i=\dim (H'_i)^{\perp}=n-\dim H'_i=k-1$, proving \cref{item:dual-2}.
\end{proof}

We also need the following equivalence.

\begin{prop}\label{prop:equivalent}
Let $n$ and $k$ be integers with $n\geq k\geq 0$. Let $V_1,\ldots,V_\ell$ be subspaces of $\F_q^n$, each of dimension at most $k$. Then the following are equivalent. 
\begin{enumerate}
\item\label{it1} There exist integers $\delta_1,\ldots,\delta_\ell\geq 0$ such that for every nonempty $\Omega \subseteq [\ell]$, 
\begin{equation}\label{eqn:copypattern}
\dim\left(\bigcap_{i\in \Omega} V_i\right)\leq k-\sum_{i\in \Omega}\delta_i.
\end{equation}
\item\label{it2} The pattern $(T_1,\ldots,T_k)$ that consists of $\delta_i$ copies of $V_i$ for $i\in [\ell]$ and additional $k-\sum_{i=1}^\ell \delta_i$ copies of $\{0\}$ is a generic kernel pattern.\footnote{We assume that the set of copies of each $V_i$ is disjoint from both the set of copies of any other $V_{i'}$ (even if $V_i = V_{i'}$) and the set of the additional $k - \sum_{j=1}^{\ell} \delta_{j}$ copies of $\{0\}$ (even if $V_i = \{0\}$).} That is, for every nonempty $\Omega'\subseteq [k]$, 
\begin{equation}\label{eqn:pattern}
\dim\left(\bigcap_{i\in \Omega'} T_i\right)\leq k-|\Omega'|.
\end{equation}
\end{enumerate}
\end{prop}
\begin{proof}
We first prove that \cref{it1} implies \cref{it2}. Let $d=k-\sum_{i=1}^\ell \delta_i$. Let $\Omega' \subseteq [k]$ be a nonempty set. We want to show that \eqref{eqn:pattern} holds. If $\Omega'$ contains at least one $i\in [k]$ with $T_i=\{0\}$, then \eqref{eqn:pattern} holds trivially. So assume $T_i\neq \{0\}$ for all $i\in \Omega'$.
Moreover, we may assume that for each $V_i$ that appears in $(T_i)_{i\in \Omega'}$, all the $\delta_i$ copies of $V_i$ also appear, since including these copies does not change LHS of \eqref{eqn:pattern} and can only decrease RHS. With this assumption, \eqref{eqn:pattern} just becomes \eqref{eqn:copypattern}. 

To prove the other direction, consider any nonempty $\Omega\subseteq [\ell]$. Let $\Omega'$ be the subset of $[k]$ that consists of the indices of all the $\delta_i$ copies of $V_i$ in $[k]$ for $i\in \Omega$ and no other indices. Applying \eqref{eqn:pattern} with the set $\Omega'$ then yields \eqref{eqn:copypattern}.
\end{proof}

The following theorem is a linear-algebraic analog of the generalized Hall's theorem as proved in \cite{brakensiek2023generic}. 

\begin{thm}[Generalized Hall's theorem for vector spaces]\label{thm:extend}
Let $n$ and $k$ be integers with $n\ge k\geq 0$. Let $V_1,\ldots,V_\ell$ be subspaces of $\F_q^n$, each of dimension at most $k$. Suppose there exist integers $\delta_1,\ldots,\delta_\ell\geq 0$ such that for every nonempty $\Omega\subseteq [\ell]$, \eqref{eqn:copypattern} holds. Then, there exist subspaces $V'_i\supseteq V_i$ of dimension $k-\delta_i$, $i=1,\dots,\ell$, such that for every nonempty $\Omega\subseteq [\ell]$, \eqref{eqn:copypattern} also holds for $V'_1,\ldots,V'_\ell$, i.e., $\dim\left(\bigcap_{i\in \Omega} V'_i\right)\leq k-\sum_{i\in \Omega}\delta_i$.
\end{thm}

Note that \cref{thm:extend} can also be viewed as a generalization of \cref{halldual1}, taking into account the ``multiplicities'' $\delta_i$. Indeed, \cref{halldual1} can be derived from \cref{thm:extend} by setting $\ell=k$ and choosing $\delta_1 = \dots = \delta_\ell = 1$.

\begin{proof}[Proof of \cref{thm:extend}]
By \eqref{eqn:copypattern}, we have $\dim V_i\leq k-\delta_i$ for $i\in [\ell]$.
If $\dim V_i=k-\delta_i$ holds for all $i\in [\ell]$, then we are done by choosing $V_i'=V_i$.
So assume this is not the case. Without loss of generality, we may assume $\dim V_1<k-\delta_1$. 

We want to extend $V_1$ to a larger subspace $V'_1$ of dimension $k-\delta_1$ while still satisfying \eqref{eqn:copypattern} for all nonempty $\Omega\subseteq [\ell]$. If $\delta_1=0$, then we can choose any subspace $V'_1$ of dimension $k$ containing $V_1$. This is because \eqref{eqn:copypattern} holds for any nonempty $\Omega\subseteq [\ell]$ that excludes $1$. Adding $1$ to $\Omega$ can only decrease LHS of \eqref{eqn:copypattern} while RHS remains unchanged. 

Now assume $\delta_1>0$. By \cref{prop:equivalent}, there exists a generic kernel pattern $(T_1,\ldots,T_k)$ consisting of $\delta_i$ copies of $V_i$ for $i\in [\ell]$ and $k-\sum_{i=1}^\ell \delta_i$ copies of $\{0\}$. 
For $i\in [\ell]$, let $J_i$ be the set of indices $j\in [k]$ such that $T_j$ is among the $\delta_i$ copies of $V_i$. 
So $|J_i|=\delta_i$.
Fix an arbitrary subspace $V\subseteq \F_q^n$ of dimension $k$ such that $V_1\subseteq V$.
		
Let $S_i=T_i\cap V$ for $i\in [k]$. The fact that $(T_1,\ldots,T_k)$ is a generic kernel pattern implies that $(S_1,\ldots,S_k)$ is also a generic kernel pattern as replacing $T_i$ with $S_i$ can only decrease LHS of \eqref{eqn:pattern}. 
Applying \cref{halldual1} to $S_1,\dots,S_k\subseteq V$ and the ambient space $V$,  
we see that there exist subspaces $S'_i$ of $V$ containing $S_i$, $i=1,\dots,k$, such that $\dim S_i'=k-1$ for $i\in [k]$ and
\begin{equation}\label{eq:Si-intersection}
\dim \left(\bigcap_{i\in\Omega}S'_i\right)\leq k-|\Omega| \quad \text{for nonempty } \Omega\subseteq [k].
\end{equation}
For $i\in [\ell]$, we define 
\begin{equation}\label{eq:def-Ui}
U_i=\begin{cases}\bigcap_{j\in J_i}S'_j,  & \text{ if }J_i\neq\emptyset,\\
V, & \text{ otherwise}.
\end{cases}
\end{equation}
For $i\in [\ell]$, if $J_i\neq\emptyset$, then $U_i\supseteq \bigcap_{j\in J_i} S_j=V_i\cap V$. And if $J_i=\emptyset$, then $U_i=V\supseteq V_i\cap V$. So $V_i\cap V\subseteq U_i$ in either case.
In particular, as $V_1\subseteq V$, we have $V_1=V_1\cap V\subseteq U_1$.

We claim that $\dim U_1=k-\delta_1$. 
If $J_1=\emptyset$, then $\delta_1=|J_1|=0$ and $\dim U_1=\dim V=k=k-\delta_1$. So the claim holds in the case.
Now consider the case where $J_1\neq\emptyset$.
By \eqref{eq:Si-intersection} and \eqref{eq:def-Ui}, we have 
\[
\dim U_1=\dim \left(\bigcap_{i\in J_1} S'_i\right)\leq k-|J_1| = k-\delta_1.
\]
Moreover, since each $S'_i$ is a subspace of $V$ of codimension one, we have 
\[
\dim U_1 =\dim\left(\bigcap_{j\in J_1}S'_i\right)\geq k-|J_1|=k-\delta_1.
\]
This proves the claim that $\dim U_1=k-\delta_1$. 

Next, we show that \eqref{eqn:copypattern} still holds for all nonempty $\Omega\subseteq [\ell]$ after replacing $V_1$ by $U_1$. We only need to verify this for nonempty $\Omega\subseteq [\ell]$ that contains $1$. 
Fix such $\Omega$ and let $\Omega'=\bigcup_{i\in \Omega} J_i\subseteq [k]$.
Note $|\Omega'|=\sum_{i\in \Omega} |J_i|$ since the sets $J_i$ are disjoint.
Then
\begin{align*}
&k-\sum_{i\in \Omega}\delta_i
=k-\sum_{i\in \Omega}|J_i|
=k-|\Omega'|
\stackrel{\eqref{eq:Si-intersection}}{\geq} 
\dim\left(\bigcap_{i\in \Omega'}S'_i\right)\\
&\stackrel{\eqref{eq:def-Ui}}{=}\dim\left(\bigcap_{i\in \Omega} U_i\right)\geq 
\dim\left(U_1\cap \bigcap_{i\in \Omega\setminus \{1\}}(V_i\cap V)\right)
=\dim\left(U_1\cap \bigcap_{i\in \Omega\setminus \{1\}} V_i \right),
\end{align*}
where the last step uses the fact that $U_1\subseteq V$ and the second last step uses the fact that $V_i\cap V\subseteq U_i$ for $i\in [\ell]$.
So \eqref{eqn:copypattern} holds with $U_1$ in place of $V_1$.

By repeating the above argument, we obtain subspaces $V'_i\supseteq V_i$ of dimension $k-\delta_i$, where $i=1,\dots,\ell$, that satisfy \eqref{eqn:copypattern} for all nonempty $\Omega\subseteq [\ell]$. This completes the proof.
\end{proof}
 
\begin{cor}\label{l-gkp-max}
Let $n$ and $k$ be integers with $n\ge k\geq 0$. 
Let $\mathcal{T}=(T_1,\ldots,T_k)$ be a generic kernel pattern consisting of $\delta_i$ copies of $V_i\subseteq \F_q^n$ for $i\in [\ell]$ and $k-\sum_{i=1}^\ell \delta_i$ copies of $\{0\}$. Then, 
there exist subspaces $V'_i\supseteq V_i$ of dimension $k-\delta_i$, $i=1,\dots,\ell$,
such that $\mathcal{T'}=(T'_1,\ldots,T'_k)$ consisting of $\delta_i$ copies of $V'_i$ for $i\in[\ell]$ and $k-\sum_{i=1}^\ell \delta_i$ copies of $\{0\}$ is a generic kernel pattern.
\end{cor}
\begin{proof}
Combining \cref{prop:equivalent} and \cref{thm:extend} proves the corollary.
\end{proof}

\subsection{Generic Kernel Patterns of Order at Most $\ell$}

The main result of this subsection is the following statement, which gives a characterization of general kernel patterns of order at most $\ell$.

\begin{lemma}\label{order-l-cha}
Let $n$, $k$, and $d$ be integers with $n\ge k\geq d\geq 0$.
Let $V_1,\ldots,V_\ell$ be subspaces of $\F_q^n$, each of dimension at most $k$. Then the following are equivalent. 
\begin{enumerate}
\item\label{eqiv:1} There exists a generic kernel pattern $\mathcal{T}:=(T_1,T_2,\dots,T_k)$
consisting solely of copies of $V_1,\ldots,V_\ell$ and an additional $d$ copies of $\{0\}$.
\item\label{eqiv:2} There exist integers $\delta_1,\ldots,\delta_\ell\geq 0$ such that $\sum_{i=1}^{\ell}\delta_i=k-d$ and for every nonempty $\Omega\subseteq [\ell]$, 
\begin{equation}\label{eq:cond2}
\dim\left(\bigcap_{i\in \Omega}V_i\right)\leq k-\sum_{i\in \Omega}\delta_i.
\end{equation}
\item \label{eqiv:3} For all partitions $P_1\sqcup P_2\sqcup\cdots\sqcup P_s=[\ell]$, we have 
\begin{equation}\label{eq:cond3}
\sum_{i=1}^s\dim\left(\bigcap_{j\in P_i} V_j\right)\leq (s-1)k+d.
\end{equation}
\end{enumerate}
\end{lemma}

\begin{proof}
By \cref{prop:equivalent}, \cref{eqiv:1} is equivalent to \cref{eqiv:2}. 
And \cref{eqiv:2} implies \cref{eqiv:3} because for any partition $P_1\sqcup P_2\sqcup\cdots\sqcup P_s=[\ell]$, we have 
\[
\sum_{i=1}^s \dim\left(\bigcap_{j\in P_i} V_j\right)\stackrel{\eqref{eq:cond2}}{\leq} \sum_{i=1}^s \left(k-\sum_{j\in P_i}\delta_j\right)=sk-(k-d)=(s-1)k+d.
\]
Finally, we prove that \cref{eqiv:3} implies \cref{eqiv:2} via an induction on $\ell$. 
When $\ell=1$, \cref{eqiv:2} and \cref{eqiv:3} are the same.
Now consider $\ell\geq 2$ and assume that \cref{eqiv:3} implies \cref{eqiv:2} for $\ell'<\ell$.
We say a partition $P_1\sqcup P_2\sqcup\cdots\sqcup P_s$ of $[\ell]$ is \emph{tight} if \eqref{eq:cond3} holds with equality.

Suppose \cref{eqiv:3} holds. We claim that there exist subspaces $V'_i\supseteq V_i$ of dimension at most $k$,
$i=1,\dots,\ell$, such that for these new subspaces, which we call \emph{padded subspaces}, \cref{eqiv:3} still holds and there exists a partition $P_1\sqcup P_2\sqcup\cdots\sqcup P_s$ of $[\ell]$ with $s\geq 2$ that is tight.

The remaining proof consists of two parts: We will first prove the above claim. Then we will show that \cref{eqiv:2} holds for the padded subspaces $V_1',\dots,V_\ell'$.
This suffices since \cref{eqiv:2} then holds for the original subspaces $V_1,\dots,V_\ell$ as well by the fact that $\bigcap_{i\in \Omega}V_i \subseteq  \bigcap_{i\in \Omega}V_i'$ for all nonempty $\Omega\subseteq [\ell]$.

For convenience, define $V_\Omega=\bigcap_{i\in\Omega} V_i$ for nonempty $\Omega\subseteq [\ell]$.
Applying \cref{eqiv:3} to the coarsest partition of $[\ell]$ shows that $\dim V_{[\ell]}\leq d$.

To prove the claim, we repeatedly select pairs $(V_i, \bv)$, where $\dim V_i < k$ and $\bv \in \mathbb{F}_q^n \setminus V_i$, and then replace $V_i$ with $V_i \oplus \mathrm{span}(\bv)$, until one of the partitions of $[\ell]$ becomes tight.\footnote{Note that at least one partition becomes tight when (or before) $\dim V_i=k$ for all $i\in [k]$. This is because (1) if all $V_i$ have dimension $k$, then we would have $\sum_{i=1}^\ell \dim V_i=\ell k\geq (\ell-1)k+d$, and (2) replacing $V_i$ with $V_i \oplus \mathrm{span}(\bv)$ increases LHS of \eqref{eq:cond3} by at most one.}
If the tight partition $P_1\sqcup \cdots\sqcup P_s$ arising this way satisfies $s\geq 2$, then we are done. So assume $s=1$, i.e., $\dim V_{[\ell]}=d$. In this case, we show that we can continue padding the subspaces $V_i$ until another partition with $s\ge2$ becomes tight, while maintaining $\dim V_{[\ell]} = d$.

For $i\in [\ell]$, applying \cref{eqiv:3} to the partition $\{i\}\sqcup ([\ell]\setminus \{i\})$ of $[\ell]$ shows
\begin{equation}\label{eq:special-partition}
\dim V_i+\dim\left(\bigcap_{j\in[\ell]\setminus\{i\}}V_j\right)\leq k+d.
\end{equation}
By inclusion-exclusion and the fact that $\dim V_{[\ell]}=d$, we have that for $i\in [\ell]$,
\[\dim\left( V_i+\bigcap_{j\in[\ell]\setminus\{i\}}V_j\right)=\dim V_i+\dim\left(\bigcap_{j\in[\ell]\setminus\{i\}}V_j\right)-\dim V_{[\ell]}\stackrel{\eqref{eq:special-partition}}{\leq} k.
\]
If there exists $i\in[\ell]$ such that $\dim\left( V_i+\bigcap_{j\in[\ell]\setminus\{i\}}V_j\right)=k$, then
\eqref{eq:special-partition} must hold with equality. In this case, the partition $\{i\}\sqcup[\ell]\setminus\{i\}$ is tight and we are done.
So assume
\begin{equation}\label{eq:less-than-k}
\dim\left( V_i+\bigcap_{j\in[\ell]\setminus\{i\}}V_j\right)<k\leq n \quad \text{ for all } i\in[\ell].
\end{equation}
Fix arbitrary $i\in[\ell]$. By \eqref{eq:less-than-k}, there exists a vector $\bv\in\F_q^n\setminus (V_{i}+\bigcap_{j\in[\ell]\setminus\{{i}\}}V_j)$. Let $V_i'=V_i \oplus \mathrm{span}(\bv)$.
Then
\begin{align*} &\dim\left(V_i'\cap\left(\bigcap_{j\in[\ell]\setminus\{i\}}V_j\right)\right)\\
&=\dim V_i'+\dim\left(\bigcap_{j\in[\ell]\setminus\{i\}}V_j\right)-\dim\left(V_i'+\left(\bigcap_{j\in[\ell]\setminus\{i\}}V_j\right)\right)\\
&=\dim V_i+1+\dim\left(\bigcap_{j\in[\ell]\setminus\{i\}}V_j\right)-\dim\left(V_i+\left(\bigcap_{j\in[\ell]\setminus\{i\}}V_j\right)\right)-1 & \text{(since $\bv\notin  V_i+\bigcap_{j\in[\ell]\setminus\{i\}}V_j$)}\\
&=\dim V_{[\ell]}=d.
\end{align*}
Thus, replacing $V_i$ with $V_i'$ preserves the fact that $\dim V_{[\ell]} = d$.
We continue this process until a partition $P_1\sqcup\cdots\sqcup P_s$ with $s\geq 2$ becomes tight. This proves the claim. 

Next, we verify that \cref{eqiv:2} holds for the padded subspaces. For ease of notation, we still denote the padded subspaces by $V_1,\dots,V_\ell$. 
Fix a tight partition  $P_1\sqcup \cdots\sqcup P_s$ of $[\ell]$ with $s\geq 2$, which exists after padding.

Consider arbitrary $i\in [s]$. 
Recall that $V_{P_i}=\bigcap_{j\in P_i}V_j$.
Let $k_i=k-\dim V_{P_i}$, so that $\dim V_{P_i}=k-k_i$.
Let $W_i$ be a complement of $V_{P_i}$ in $\F_q^n$. 
For each $j\in P_i$, 
let $T_j:=V_j\cap W_i$.
For $j\in P_i$,
as $\F_q^n=V_{P_i}\oplus W_i$ and $V_j\supseteq V_{P_i}$, we see that $V_j=V_{P_i}\oplus T_j$ and $\dim T_j = \dim V_j - \dim V_{P_i}\leq k-(k-k_i)=k_i$.

Consider an arbitrary partition $Q_1\sqcup \cdots\sqcup Q_t$ of $P_i$. As $P_1\sqcup\cdots \sqcup P_{i-1}\sqcup Q_1\sqcup \cdots \sqcup Q_t\sqcup P_{i+1}\sqcup\cdots\sqcup P_s$ is a partition of $[\ell]$, by \cref{eqiv:3}, we have
\begin{equation}\label{eq:partition-PQ}
\sum_{j=1}^t \dim V_{Q_j}+\sum_{j\in [s]\setminus \{i\}}\dim V_{P_j} \leq (s+t-2)k+d.
\end{equation}
As the partition $P_1\sqcup\cdots\sqcup P_s$ of $[\ell]$ is tight, we also have
\begin{equation}\label{eq:partition-P}
\sum_{j=1}^s\dim V_{P_j} = (s-1)k+d.
\end{equation}
Define $T_\Omega=\bigcap_{j\in\Omega} T_j$ for nonempty $\Omega\subseteq P_i$.
Then
\begin{align*}
\sum_{j=1}^t \dim T_{Q_j} &=\left(\sum_{j=1}^t \dim V_{Q_j}\right)-t\cdot \dim V_{P_i}\\
&=\left(\sum_{j=1}^t \dim V_{Q_j}+\sum_{j\in [s]\setminus \{i\}}\dim V_{P_j}\right)- \sum_{j=1}^s\dim V_{P_j}-(t-1)\dim V_{P_i}\\
&\stackrel{\eqref{eq:partition-PQ},\eqref{eq:partition-P}}{\leq} (s+t-2)k+d-((s-1)k+d)-(t-1)(k-k_i)=(t-1)k_i.
\end{align*}
This shows that \cref{eqiv:3} holds for the subspaces $(T_j)_{j\in P_i}$ indexed by $P_i$ (instead of by $[\ell]$), where the parameter $d$ is set to zero.

Let $\ell_i=|P_i|$. Note that $\ell_i<\ell$ since $s\geq 2$.
Identifying $P_i$ with $[\ell_i]$ and applying the induction hypothesis, we see that there exist integers $\delta_j\geq 0$ for all $j\in P_i$ such that
\begin{equation}\label{eq:sum-of-delta}
\sum_{j\in P_i}\delta_j=k_i
\end{equation}
and for all nonempty $\Omega\subseteq P_i$,
\begin{equation}\label{eq:T-omega}
\dim T_{\Omega}\leq k_i-\sum_{j\in \Omega} \delta_j.
\end{equation}
As $i\in [s]$ is arbitrary, we can perform the above procedure for $i=1,\dots,s$, which yields $\delta_j$ for all $j\in [\ell]$.
We now verify that $\delta_1,\dots,\delta_{\ell}$ satisfy \cref{eqiv:2}. First, observe that
\[
\sum_{i=1}^\ell \delta_i
=\sum_{i=1}^s\sum_{j\in P_i}\delta_j
\stackrel{\eqref{eq:sum-of-delta}}{=}\sum_{i=1}^s k_i
=\sum_{i=1}^s \bigg(k-\dim V_{P_i}\bigg)\stackrel{\eqref{eq:partition-P}}{=}k-d.
\]
So it remains to verify \eqref{eq:cond2} for all nonempty $\Omega\subseteq [\ell]$. 
Fix nonempty $\Omega\subseteq [\ell]$. Let $I=\{i\in [s]: P_i\cap \Omega \neq \emptyset\}$.
Consider $i\in I$. Recall that for $j\in P_i$, the complement $T_j$ of $V_{P_i}$ in $V_j$ is chosen to be $V_j\cap W_i$, where $W_i$ is a complement of $V_{P_i}$ in $\F_q^n$. It follows that for any $J\subseteq P_i$,
\begin{equation}\label{eq:intersection-of-Vi}
\bigcap_{j\in J} V_j=V_{P_i}\oplus \left(\bigcap_{j\in J} T_j\right)=V_{P_i}\oplus T_{J}.
\end{equation}
Then, we have
\begin{align*}
&\dim\left(\bigcap_{i\in \Omega}V_i\right)
=\dim\left(\bigcap_{i\in I}\bigcap_{j\in P_i\cap \Omega} V_j\right)
\stackrel{\eqref{eq:intersection-of-Vi}}{=}\dim\left(\bigcap_{i\in I}\left(V_{P_i}\oplus T_{P_i\cap\Omega}\right)\right)\\
&\leq\dim\left(\bigcap_{i\in I}V_{P_i}\right)+\sum_{i\in I}\dim\left(T_{P_i\cap\Omega}\right) & \text{(by \cref{lm:dim})}\\
&\stackrel{\eqref{eq:T-omega}}{\leq} \dim\left(\bigcap_{i\in I}V_{P_i}\right)+\sum_{i\in I}\left(k_i-\sum_{j\in P_i\cap\Omega} \delta_j\right)\\
&=\dim\left(\bigcap_{i\in I}V_{P_i}\right)-\sum_{i\in I} \dim V_{P_i}+k|I|-\sum_{j\in \Omega}\delta_j & \text{(as  $\dim V_{P_i}=k-k_i$)}\\
&=\dim\left(\bigcap_{i\in I}V_{P_i}\right)+\sum_{i\in [s]\setminus I}\dim V_{P_i}-\sum_{i\in [s]}\dim V_{P_i}+k|I|-\sum_{j\in \Omega}\delta_j\\
&\stackrel{\eqref{eq:partition-P}}{=}\dim\left(\bigcap_{i\in I}V_{P_i}\right)+\sum_{i\in [s]\setminus I}\dim V_{P_i}-((s-1)k+d)+k|I|-\sum_{j\in \Omega}\delta_j\\
& \leq (s-|I|)k+d-((s-1)k+d)+k|I|-\sum_{j\in \Omega}\delta_j\\ &=k-\sum_{j\in \Omega}\delta_j,
\end{align*}
where the second last step follows by applying \eqref{eq:cond3} to the partition of $[\ell]$ consisting of the set $\bigcup_{i\in I}P_i$ and the sets $P_j$ for $j\in [s]\setminus I$. So \cref{eqiv:2} holds, which concludes the proof.
\end{proof}

\begin{rmka}
The proof of the above lemma closely follows that of \cite[Lemma~2.8]{brakensiek2023generic}, which gives a characterization of general zero patterns of order at most $\ell$. 
However, some details differ due to our focus on vector spaces instead of plain sets.
For example, the proof of \cite[Lemma~2.8]{brakensiek2023generic} uses the pigeonhole principle to argue that one can continue padding a collection of sets until some partition $P_1\sqcup P_2\sqcup\dots\sqcup P_s=[\ell]$ with $s\geq 2$ becomes tight.
For us, such a simple argument does not work, as extending a subspace $V_i$ to $V_i\oplus\mathrm{span}\{\bv\}$ for 
a vector $\bv\not\in V_i$ would introduce elements that are not in $V_i\cup\{\bv\}$. Instead, our proof for the feasibility of padding is based on the inequality \eqref{eq:special-partition}, derived in turn from \cref{eqiv:3} of \cref{order-l-cha}. We find our argument to be more general and believe it may be of independent interest.
\end{rmka}

\section{The GM-MRD Theorem}\label{gmmrdsec}

In this section, we prove the GM-MRD theorem, which states that symbolic Gabidulin code over $\F=\F_q(Z_1,\dots,Z_n)$ is $\mathrm{GKP}(\ell)$ for all $\ell$. For the convenience of applications, we have formulated the theorem as follows.

\begin{thm}[GM-MRD theorem]\label{thm:GM-MRD-v1-new}
Let $1\leq k\leq n$ be integers.
Let $\F=\F_q(Z_1,\dots,Z_n)$ and $G=\left(Z_j^{q^{i-1}}\right)_{i\in [k], j\in [n]}\in\F^{k\times n}$.
For every generic kernel pattern $\mathcal{V}=(V_1,\dots,V_k)$ and matrices $A_1,\dots,A_k$, where $A_i\in\F_q^{n\times \dim V_i}$ and $\langle A_i\rangle=V_i$, there exists a matrix $M_{\mathcal{V}}\in \F^{k\times k}$ such that
\begin{enumerate}
\item\label{item:gmmrd1} $M_{\mathcal{V}}$ is invertible,
\item\label{item:gmmrd2} $\vm_i G A_i=0$ for all $i\in [k]$, where $\vm_i$ denotes the $i$-th row of $M_{\mathcal{V}}$, and
\item\label{item:gmmrd3} The entries of $M_{\mathcal{V}}$ are polynomials in $Z_1,\dots,Z_n$ over $\F_q$ of degree at most $q^{k-1}$.
\end{enumerate}
In particular, the symbolic Gabidulin code $\mG_{n,k}(Z_1,\dots, Z_n)$ over $\F$ is $\mathrm{GKP}(\ell)$ for all $\ell\geq 1$.
\end{thm}

\cref{thm:GM-MRD-v1-new} implies the following result, which states that a random Gabidulin code over a large enough finite field is, with high probability, $\mathrm{GKP}(\ell)$ for all $\ell$.

\begin{thm}[GM-MRD theorem, finite field version]\label{thm:GM-MRD-finite}
Let $1\leq k\leq n\leq m$ be integers. Let $\ell\geq 1$. Let $(\alpha_1,\dots,\alpha_n)$ be uniformly distributed over the set of all vectors in $\F_{q^m}^n$ whose coordinates are linearly independent over $\F_q$. Then with probability at least $1-3kq^{nk\cdot\min\{\ell,k\}+k-m}$, the Gabidulin code $\mG_{n,k}(\alpha_1,\dots, \alpha_n)$ over $\F_{q^m}$ is $\mathrm{GKP}(\ell)$.
\end{thm}

\begin{proof}
Note that we may permute the subspaces $V_i$ in a generic kernel pattern $\mathcal{V}=(V_1,\dots,V_k)$ without affecting whether $\mG_{n,k}(\alpha_1,\dots, \alpha_n)$ attains $\mathcal{V}$.
Up to permutation, each generic kernel pattern $\mathcal{V}$ of order at most $\ell$ can be represented by a list of subspaces $V_1,\dots,V_{\ell'}$ of dimension at most $k-1$, where $\ell'\leq \min\{\ell,k\}$, together with their multiplicities $\delta_1,\dots,\delta_{\ell'}\in [k]$.
Each $V_i$ can be represented by a matrix $A_i\in\F_q^{n\times (k-1)}$ satisfying $V_i=\langle A_i\rangle$, for which there are at most $q^{n(k-1)}$ choices.
So the number of generic kernel patterns we need to consider is bounded by 
\[
\mathcal{N}:=\sum_{\ell'=1}^{\min\{\ell,k\}}(kq^{n(k-1)})^{\ell'}\leq \sum_{\ell'=1}^{\min\{\ell,k\}} q^{nk\ell'}\leq 2q^{nk\cdot\min\{\ell,k\}}.
\]

First assume that $(\alpha_1,\dots,\alpha_n)$ is uniformly distributed over $\F_{q^m}^n$ at random.
Consider a fixed generic kernel pattern $\mathcal{V}$. Let $M_{\mathcal{V}}$ be as in \cref{thm:GM-MRD-v1-new}, whose entries are polynomials in $Z_1,\dots,Z_n$ over $\F_q$.
For $(\alpha_1,\dots,\alpha_n)\in \F_{q^m}^n$ and $i\in [k]$, we have
$\overline{\vm}_i \overline{G} A_i=0$ by \cref{item:gmmrd2} of \cref{thm:GM-MRD-v1-new}, where $\overline{\vm}_i$ is the $i$-th row of $M_{\mathcal{V}}|_{Z_1=\alpha_1,\dots,Z_n=\alpha_n}$ and $\overline{G}$ is the generator matrix $\left(\alpha_j^{q^{i-1}}\right)_{i\in [k], j\in [n]}$ of $\mG_{n,k}(\alpha_1,\dots, \alpha_n)$. By \cref{item:gmmrd3} of \cref{thm:GM-MRD-v1-new}, the degree of $\det\left(M_{\mathcal{V}}\right)$ is at most $kq^{k-1}$.
By the Schwartz--Zippel lemma, $\det\left(M_{\mathcal{V}}|_{Z_1=\alpha_1,\dots,Z_n=\alpha_n}\right)=0$ holds with probability at most $\delta:=kq^{k-1}/q^m=kq^{k-m-1}$.
And when $\det\left(M_{\mathcal{V}}|_{Z_1=\alpha_1,\dots,Z_n=\alpha_n}\right)\neq 0$, we know that $\mG_{n,k}(\alpha_1,\dots, \alpha_n)$ attains $\mathcal{V}$.

By the union bound, the probability that $\mG_{n,k}(\alpha_1,\dots, \alpha_n)$ attains all generic kernel patterns of order at most $\ell$, i.e., it is $\mathrm{GKP}(\ell)$, is at least $1-\mathcal{N}\delta$, assuming that $\alpha:=(\alpha_1,\dots,\alpha_n)$ is drawn from the uniform distribution $U$ over $\F_{q^m}^n$.
But $\alpha$ is actually drawn from the uniform distribution $U_S$ over $S$, where $S$ denotes the set of all vectors in $\F_{q^m}^n$ whose coordinates are linearly independent over $\F_q$. However, note that $\Pr_{\alpha\sim U}[\alpha\not\in S]\leq (1+q+\dots+q^{n-1})/q^m\leq q^{n-m}$. So the statistical distance between $U$ and $U_S$ is bounded by $q^{n-m}$. So for $\alpha\sim U_S$, the probability that $\mG_{n,k}(\alpha_1,\dots, \alpha_n)$ is $\mathrm{GKP}(\ell)$ is at least 
\[
1-\mathcal{N}\delta-q^{n-m}
\geq 1-2q^{nk\cdot\min\{\ell,k\}}\cdot kq^{k-m-1}-q^{n-m}\geq 1-3k q^{nk\cdot\min\{\ell,k\}+k-m},
\]
which complete the proof.
\end{proof}
\subsection{Proof of the GM-MRD Theorem}

In the following, fix $n\in\N^+$ and let $\F=\F_{q}(Z_1,\dots,Z_{n})$. 
For a subset $S\subseteq \F$, denote by $\F_{q}[S]$ the $\F_q$-subalgebra of $\F$ generated by the elements in $S$. In particular, $\F_q[S]\subseteq \F_q[Z_1,\dots,Z_n]$ if $S\subseteq \F_q[Z_1,\dots,Z_n]$.

\begin{defn}\label{defn:admissible}
We say a tuple $(V_i)_{i\in [m]}$ of $\F_q$-subspaces $V_i$ of $\F_{q}[Z_1,\dots,Z_{n}]$ is \emph{$s$-admissible} if there exist $\F_q$-subspaces $V,W$ of $\mathrm{span}_{\F_q}\{Z_1,\dots,Z_{n}\}$ and a nonzero $q$-linearized polynomial $f\in (\F_q[W])[X]$ in $X$ such that 
\begin{enumerate}
\item $\mathrm{span}_{\F_q}\{Z_1,\dots,Z_{n}\}=V\oplus W$,
\item $\dim V\leq s$, and
\item $V_1,\dots,V_m\subseteq f(V)$.
\end{enumerate}
\end{defn}

Note that in the above definition, as $f$ is $q$-linearized, the map $u\mapsto f(u)$ defines an isomorphism between the $\F_q$-linear spaces $V$ and $f(V)$. 

We can apply an invertible linear transformation to $\mathrm{span}_{\F_q}\{Z_1,\dots,Z_n\}$ over $\F_q$, mapping $V$ to $\mathrm{span}_{\F_q}\{Z_1,\dots,Z_{\dim V}\}$ and $W$ to $\mathrm{span}_{\F_q}\{Z_{\dim V+1},\dots,Z_n\}$.
Thus, \cref{defn:admissible} states that, after applying an invertible linear transformation over $\F_q$, the spaces $V_1,\dots,V_m$ are contained in $f(\mathrm{span}_{\F_q}\{Z_1,\dots,Z_d\})$ for some $d\leq s$, and $f$ is a $q$-linearized polynomial whose coefficients are in $\F_q[Z_{d+1},\dots,Z_n]$.
Importantly, $f$ is evaluated only on the $\F_q$-linear span of $Z_1,\dots,Z_d$ while its coefficients do not depend on $Z_1,\dots,Z_d$. This fact is crucially used in the proof of \cref{thm:generalization-new} (specifically in Case~3 of the proof). 

The following lemma follows straightforwardly from the definition.

\begin{lemma}\label{lem:closed}
Let $(V_i)_{i\in [m]}$ be $s$-admissible and let $V_i'$ be an $\F_q$-subspace of $V_i$ for $i\in [m]$. Then $(V_i')_{i\in [m]}$ is also $s$-admissible.
\end{lemma}

%
%

For $k\geq m\geq 1$ and $s\geq 0$, define 
\[
V_{k,m,s}:=\left\{((V_i,r_i))_{i\in [m]}: (V_i)_{i\in [m]}\text{ is $s$-admissible},r_i\in \mathbb{N}^+,\dim(V_i)+r_i\leq k,\sum_{i=1}^m r_i=k\right\}.
\] 
For each $\cS=((V_i,r_i))_{i\in [m]}\in V_{k,m,s}$, we associate with it a matrix $M_{\cS}\in \F^{k\times k}$, defined as follows.

\begin{defn}[Matrix $M_{\cS}$]\label{defn:ms}
Let $\cS=((V_i,r_i))_{i\in [m]}\in V_{k,m,s}$.
For $i\in [m]$, let 
\begin{equation}\label{eq:fi}
f_i(X)=\prod_{\alpha\in V_i}(X-\alpha)^{q^{k-\dim(V_i)-r_i}}=\sum_{j=0}^{k-r_i}a_{i,j}X^{q^j}  
\end{equation}
which is a $q$-linearized polynomial in $X$ with coefficients in $\F_q[Z_1,\dots,Z_n]\subseteq\F$ by \cref{prop:q_linear}. 
Define the matrix 
\begin{equation}\label{eq:ms}
M_{\cS}=\begin{pmatrix}
	M_1\\
	M_2\\
	\vdots\\
	M_m
 \end{pmatrix}\in\mathbb{F}^{k\times k}
\end{equation}
where 
\begin{equation}\label{eq:mi}
M_i=\begin{pmatrix}
a_{i,0}&a_{i,1}&\cdots& a_{i,k-r_i-1} & a_{i,k-r_i}&  0 & \cdots &0   \\
0 &a_{i,0}^q& \cdots  & a_{i,k-r_i-2}^q & a_{i,k-r_i-1}^q & a_{i,k-r_i}^q & \cdots &0 \\
\vdots & \vdots &\ddots & \vdots &\vdots &\vdots& \ddots &0\\
0 &0  & \cdots &a_{i,0}^{q^{r_i-1}} & a_{i,1}^{q^{r_i-1}} & a_{i,2}^{q^{r_i-1}} &\cdots & a_{i,k-r_i}^{q^{r_i-1}} 
\end{pmatrix}\in\mathbb{F}^{r_i\times k}.
\end{equation}
In other words, for $(j,j')\in [r_i]\times [k]$, the $(j,j')$-th entry of $M_i$ is the coefficient of $X^{q^{j'-1}}$ in $f_i(X)^{q^{j-1}}$.
\end{defn}

We will derive \cref{thm:GM-MRD-v1-new} from the following theorem.

\begin{thm}\label{thm:GM-MRD-v2-new}
Let $\cS=((V_i,1))_{i\in [k]}\in V_{k,k,s}$ such that for all nonempty $\Omega\subseteq [k]$, 
\[
\dim\left(\bigcap_{i\in \Omega} V_i\right)\leq k-|\Omega|.
\]
Then $\det(M_{\cS})\neq 0$.
\end{thm}

\cref{thm:GM-MRD-v2-new} is, in turn, derived from the following more general statement.

\begin{thm}\label{thm:generalization-new}
Let $\cS=((V_i,r_i))_{i\in [m]}\in V_{k,m,s}$. Then $\det(M_{\cS})\neq 0$ iff for all nonempty $\Omega\subseteq [m]$, 
\begin{equation}\label{eqn:key-new}
\dim\left(\bigcap_{i\in \Omega} V_i\right)+\sum_{i\in \Omega}r_i\leq \max_{i\in \Omega}\{\dim(V_i)+r_i\}.
\end{equation}
\end{thm}

Note that \cref{thm:GM-MRD-v2-new} follows from \cref{thm:generalization-new} by choosing $m=k$ and $r_i=1$ for $i\in [m]$. 
We now derive \cref{thm:GM-MRD-v1-new} from \cref{thm:GM-MRD-v2-new}.

\begin{proof}[Proof of \cref{thm:GM-MRD-v1-new} given \cref{thm:GM-MRD-v2-new}]
Consider generic kernel pattern $\mathcal{V}=(V_1,\dots,V_k)$ and matrices $A_1,\dots,A_k$, where $A_i\in\F_q^{n\times \dim V_i}$ and $\langle A_i\rangle=V_i$. Let $\{e_1,\dots,e_n\}$ be the standard basis of $\F_q^n$.
Let $\sigma: \F_q^n\to \mathrm{span}_{\F_q}\{Z_1,\dots,Z_n\}$ be the $\F_q$-linear isomorphism sending $e_j$ to $Z_j$ for $j\in [n]$.
Define $V_i'=\sigma(V_i)$ for $i\in [k]$.
Let $\cS=((V_i',1))_{i\in [k]}$.
As $\mathcal{V}=(V_1,\dots,V_k)$ is a generic kernel pattern and $\sigma$ is an $\F_q$-linear isomorphism,
we see that $\cS\in V_{k,k,n}$ and that $\dim\left(\bigcap_{i\in \Omega} V_i'\right)\leq k-|\Omega|$ for nonempty $\Omega\subseteq [k]$. So $\det(M_{\cS})\neq 0$ by \cref{thm:GM-MRD-v2-new}, where $M_{\cS}\in\F^{k\times k}$ is as defined in \cref{defn:ms}.
By definition, for $i\in [k]$, each entry of the $i$-th row of $M_{\cS}$ is a coefficient of $f_i$, where $f_i(X)=\prod_{\alpha\in V_i'}(X-\alpha)^{q^{k-\dim(V_i')-1}}$. As each $\alpha\in V_i'$ is a linear form in $Z_1,\dots,Z_n$ over $\F_q$, we conclude that the entries of $M_{\cS}$ are polynomials in $Z_1,\dots,Z_n$ over $\F_q$ of degree at most $q^{k-1}$.

Choose $M_{\mathcal{V}}$ to be $M_{\cS}$. Then \cref{item:gmmrd1} and \cref{item:gmmrd3} of \cref{thm:GM-MRD-v1-new} hold by the above discussions. It remains to verify \cref{item:gmmrd2}. 
Let $i\in [k]$.
Suppose the $i$-th row of $M_{\mathcal{V}}=M_{\cS}$ is $\vm_i=(c_1,\dots,c_k)$, or equivalently, $f_i(X)=\sum_{j=1}^k c_j X^{q^{j-1}}$. 
Consider arbitrary $\bv=(v_1,\dots,v_n)\in V_i\subseteq \F_q^n$.
As $G=\left(Z_j^{q^i-1}\right)_{i\in [k], j\in [n]}$, we have
\begin{equation}\label{eq:kernel-original}
\vm_i G \bv=(f_i(Z_1),\dots,f_i(Z_n))\bv=f_i\left(\sum_{j=1}^n v_i Z_i\right)=0,
\end{equation}
where the second equality holds since $f_i$ is $q$-linearized, and the last equality holds since $f_i$ vanishes on $V_i'$ and $\sum_{j=1}^n v_i Z_i=\sigma(\bv)\in V_i'$.
Choosing $\bv$ to be the columns of $A_i$ in \eqref{eq:kernel-original} shows that $\vm_i G A_i=0$. So \cref{item:gmmrd2} holds. 
\end{proof}

\subsection{Proof of \cref{thm:generalization-new}}

It remains to prove \cref{thm:generalization-new}.
To this end, we introduce the following lemma, which characterizes the nonsingularity of a matrix $M_{\cS}$ in terms of the compositions of $q$-linearized polynomials.

\begin{lemma}\label{lm:equivalent-new}
Let $\cS=((V_i,r_i))_{i\in [m]}\in V_{k,m,s}$.
For $i\in [m]$, let $f_i$ be given as in \eqref{eq:fi}.
Then, $\det(M_{\cS})=0$ if and only if there exist $q$-linearized polynomials $g_1,\ldots,g_m\in \F[X]$, not all zero, such that the $q$-degree of each $g_i$ is at most $r_i-1$ and $\sum_{i=1}^m g_i\circ f_i=0$. 
\end{lemma}
\begin{proof}
Suppose $\det(M_{\cS})=0$. Then there exists nonzero $\by\in \mathbb{F}^{k}$ such that $\by\cdot M_{\cS}=\mathbf{0}$.
Write $\by=(\by_1,\ldots,\by_m)\in\F^k$ with $\by_i=(y_{i,1},\ldots,y_{i,r_i})\in \mathbb{F}^{r_i}$. 
For $i\in [m]$, let $f_i(X)\in \F[X]$ and $M_i\in\F^{r_i\times k}$ be as in \eqref{eq:fi} and \eqref{eq:mi} respectively, and let $g_i(X)=\sum_{j=1}^{r_i}y_{i,j}X^{q^{j-1}}\in\F[X]$.
By definition, $\by_i \cdot M_i$ is precisely
the vector of the first $k$ coefficients of the $q$-linearized polynomial
\[
\sum_{j=1}^{r_i} y_{i,j} \cdot f_i(X)^{q^{j-1}}=(g_i\circ f_i)(X).
\]
So $\mathbf{0}=\by\cdot M_{\cS}=\sum_{i=1}^m \by_i\cdot M_i$ is the vector of the first $k$ coefficients of the $q$-linearized polynomial $\sum_{i=1}^m g_i\circ f_i$. By definition, for $i\in [m]$, we have $\deg_q(f_i)=k-r_i$ and $\deg_q(g_i)\leq r_i-1$. It follows that the $q$-degree of $\sum_{i=1}^m g_i\circ f_i$ is at most $k-1$. So $\sum_{i=1}^m g_i\circ f_i=0$. 

Conversely, suppose there exist $q$-linearized polynomials $g_1,\ldots,g_m\in \F[X]$, not all zero, such that the $q$-degree of each $g_i$ is at most $r_i-1$ and $\sum_{i=1}^m g_i\circ f_i=0$. Find the nonzero vector $\by=(\by_1,\ldots,\by_m)\in\F^k$ via $\by_i=(y_{i,1},\ldots,y_{i,r_i})\in \mathbb{F}^{r_i}$ and $g_i(X)=\sum_{j=1}^{r_i}y_{i,j}X^{q^{j-1}}$. Reversing the above proof shows $\by\cdot M_{\cS}=\mathbf{0}$. So $\det(M_{\cS})=0$.
\end{proof}

Now we are ready to prove \cref{thm:generalization-new}. Our proof is based on an induction on the parameter $(k,m,s)$ in the lexicographical order, following the approach in \cite{yildiz2019optimum}. The main difference is that the product of polynomials is replaced by the composition of $q$-linearized polynomials, which is not commutative. Consequently, we need to adapt the proof in \cite{yildiz2019optimum} to circumvent this obstacle. 

\begin{proof}[Proof of \cref{thm:generalization-new}]
We first prove the ``only if'' direction.
Assume that for some nonempty $\Omega\subseteq [m]$, the inequality \eqref{eqn:key-new} does not hold. We will show that $\det(M_{\cS})=0$.
Let $V_0=\bigcap_{i\in \Omega}V_i$, $r_0=\sum_{i\in \Omega}r_i$ and $k'=\max_{i\in \Omega}\{\dim(V_i)+r_i\}$. Then $\dim(V_0)+r_0>k'$ as \eqref{eqn:key-new} does not hold. 

Let $M^*$ be the  $r_0\times k$ submatrix of $M_{\cS}$ obtained by removing all blocks $M_i$ for $i\notin \Omega$. 
By the definitions \eqref{eq:fi} and \eqref{eq:mi}, the first $k-k'$ columns of $M^*$ are zero. We remove these $k-k'$ columns and denote by $M'$ the resulting $r_0\times k'$ matrix. The matrix $M'$ consists of the blocks $M'_i$ for $i\in \Omega$, placed vertically, each given by
\begin{equation}\label{eq:mi-prime}
M_i'=\begin{pmatrix}
b_{i,0}&b_{i,1}&\cdots& b_{i,k'-r_i-1} & b_{i,k'-r_i}&  0 & \cdots &0   \\
0 &b_{i,0}^q& \cdots  & b_{i,k'-r_i-2}^q & b_{i,k'-r_i-1}^q & b_{i,k'-r_i}^q & \cdots &0 \\
\vdots & \vdots &\ddots & \vdots &\vdots &\vdots& \ddots &0\\
0 &0  & \cdots &b_{i,0}^{q^{r_i-1}} & b_{i,1}^{q^{r_i-1}} & b_{i,2}^{q^{r_i-1}} &\cdots & b_{i,k'-r_i}^{q^{r_i-1}} 
\end{pmatrix}\in\mathbb{F}^{r_i\times k'}
\end{equation}
where each $b_{i,j}$ is determined via  $b_{i,j}=a_{i,j+k-k'}$ and $a_{i,j+k-k'}$ denotes the coefficient of $X^{q^{j+k-k'}}$ in $f_i(X)=\prod_{\alpha\in V_i}(X-\alpha)^{q^{k-\dim(V_i)-r_i}}$. Therefore,
\begin{equation}\label{eq:fi-prime}
\prod_{\alpha\in V_i}(X-\alpha)^{q^{k-\dim(V_i)-r_i}}=\sum_{j=0}^{k'-r_i}b_{i,j}\left(X^{q^{k-k'}}\right)^{q^j}.
\end{equation}
Pick a basis $\alpha_1,\dots,\alpha_h$ of $V_0$ over $\F_q$, where $h=\dim V_0$.
Let $\alpha_i'=\alpha_i^{q^{k-k'}}$ for $i\in [h]$.
Then $\alpha'_1,\dots,\alpha'_h$ are linearly independent over $\F_q$. 
Let $H=(\alpha_{\ell}'^{q^{j-1}})_{j\in [k'], \ell\in [h]}\in\F^{k'\times h}$. Then $H$ has full rank by \cref{lem:moore}.
For each $i\in \Omega$,
as $V_0\subseteq V_i$, the polynomial in 
\eqref{eq:fi-prime} vanishes at $\alpha_1,\dots,\alpha_h$, i.e., $\sum_{j=0}^{k'-r_i}b_{i,j}\alpha_{\ell}'^{q^j}=0$ for $\ell\in [h]$. Equivalently, $M_i' H=0$ for $i\in \Omega$. So $M' H=0$.
Therefore,
\[
\rank(M^*)=\rank(M')\leq k'-\min\{k', h\}=\max\{0, k'-\dim V_0\}<r_0.
\]
So $M^*$ does not have full row rank. It follows that $M_{\cS}$ does not have full row rank either. So $\det(M_{\cS})=0$.

Next, we prove the ``if'' direction via an induction on $(k,m,s)$ with respect to the lexicographical order:

If $m=1$, then $r_1=k$ and $V_1=\{0\}$. In this case, the matrix $M_{\cS}$ is the $k\times k$ identity matrix by the definitions \eqref{eq:fi}, \eqref{eq:ms}, and \eqref{eq:mi}. In particular, $\det(M_{\cS})\neq 0$.

If $s=0$, then all $V_i$ are zero since $(V_i)_{i\in [m]}$ is $s$-admissible.
In this case, \eqref{eqn:key-new} implies that $m=1$ and $r_1=k$. So $M_{\cS}$ is again the $k\times k$ identity matrix and $\det(M_{\cS})\neq 0$.

Now suppose $k\geq m\geq 2$ and $s\geq 1$.
Assume the ``if'' direction holds for all $(k',m',s')$ that are smaller than $(k,m,s)$ in lexicographical order. 
Let $\cS=((V_i,r_i))_{i\in [m]}\in V_{k,m,s}$ such that \eqref{eqn:key-new} holds for all nonempty $\Omega\subseteq [m]$.
We will show that $\det(M_{\cS})\neq 0$ by dividing the proof into the following three cases:

\paragraph{Case 1:} There exists $\Omega_1\subseteq [m]$ such that $2\leq |\Omega_1|\leq m-1$ and 
\begin{equation}\label{eq:gkp-tight}
\dim\left(\bigcap_{i\in \Omega_1}V_i\right)+\sum_{i\in \Omega_1}r_i=\max_{i\in \Omega_1}\{\dim(V_i)+r_i\}.
\end{equation}

In this case, let $\Omega_2=\{0\}\cup ([m]\setminus\Omega_1)$. Then $2\leq |\Omega_1|,|\Omega_2|\leq m-1$. 
Define $V_0=\bigcap_{i\in \Omega_1}V_i$ and $r_0=\sum_{i\in \Omega_1}r_i$. 
Then \eqref{eq:gkp-tight} becomes 
\begin{equation}\label{eq:gkp-tight-2}
\dim(V_0)+r_0=\max_{i\in \Omega_1}\{\dim(V_i)+r_i\}.
\end{equation}

As $(V_i)_{i\in [m]}$ is $s$-admissible, there exist $\F_q$-subspaces $V,W\subseteq \mathrm{span}_{\F_q}\{Z_1,\dots,Z_{n}\}$ and a nonzero $q$-linearized polynomial $f\in (\F_q[W])[X]$ in $X$ such that $\mathrm{span}_{\F_q}\{Z_1,\dots,Z_{n}\}=V\oplus W$, $\dim V\leq s$, and
$V_1,\dots,V_m\subseteq f(V)$.
Fix a complement $\overline{V}_0$ of $V_0$ in $f(V)$, so that $f(V)=V_0\oplus \overline{V}_0$.
For $i\in \Omega_1$, let $V_i'=V_i\cap \overline{V}_0$, which is a complement of $V_0$ in $V_i$. 
Let $\cS_1=((V_i', r_i))_{i\in \Omega_1}$ and $\cS_2=((V_i, r_i))_{i\in \Omega_2}$.
As $\dim(V'_i)+r_i=\dim(V_i)-\dim(V_0)+r_i\stackrel{\eqref{eq:gkp-tight-2}}{\leq} r_0$ for $i\in\Omega_1$ and $\sum_{i\in \Omega_1} r_i=r_0$, we have $\cS_1\in V_{r_0,|\Omega_1|,s}$.
As
\[
k=\sum_{i=1}^m r_i=\sum_{i\in \Omega_1}r_i+\sum_{i\in [m]\setminus\Omega_1}r_i=r_0+\sum_{i\in [m]\setminus\Omega_1}r_i=\sum_{i\in\Omega_2}r_i,
\]
$\dim(V_i)+r_i\leq k$ for $i\in [m]\setminus\Omega_1$, and $\dim(V_0)+r_0\stackrel{\eqref{eq:gkp-tight-2}}{=}\max_{i\in \Omega_1}\{\dim(V_i)+r_i\}\leq k$, we also have $\cS_2\in V_{k,|\Omega_2|,s}$.

Next, we check that $\cS_1=((V_i', r_i))_{i\in \Omega_1}$ and $\cS_2=((V_i, r_i))_{i\in \Omega_2}$ both satisfy the condition \eqref{eqn:key-new}. Indeed, for any nonempty $\Omega\subseteq \Omega_1$, 
\begin{align*}
\dim\left(\bigcap_{i\in \Omega}V'_i\right)+\sum_{i\in \Omega}r_i
&=\dim\left(\bigcap_{i\in \Omega}V_i\right)-\dim(V_0)+\sum_{i\in \Omega}r_i\\
&\stackrel{\eqref{eqn:key-new}}{\leq} \max_{i\in \Omega} \{\dim(V_i)-\dim(V_0)+r_i\}=\max_{i\in \Omega} \{\dim(V'_i)+r_i\}.
\end{align*}
And for any nonempty $\Omega\subseteq\Omega_2$, if $0\notin \Omega$, then \eqref{eqn:key-new} holds for $\cS_2$ and $\Omega$ since $\Omega\subseteq [m]$. Otherwise, let $\Omega'=\Omega\setminus \{0\}\subseteq [m]\setminus\Omega_1$, and we have
\begin{align*}
\dim\left(\bigcap_{i\in \Omega}V_i\right)+\sum_{i\in \Omega}r_i&=\dim\left(\bigcap_{i\in \Omega'\cup \Omega_1}V_i\right)+\sum_{i\in \Omega'\cup \Omega_1}r_i\\
&\stackrel{\eqref{eqn:key-new}}{\leq} \max_{i\in \Omega'\cup \Omega_1} \{\dim(V_i)+r_i\}\stackrel{\eqref{eq:gkp-tight-2}}{=}\max_{i\in \Omega} \{\dim(V_i)+r_i\},
\end{align*}
where the first equality holds since $V_0=\bigcap_{i\in \Omega_1}V_i$, $r_0=\sum_{i\in \Omega_1}r_i$, and the union $\Omega'\cup\Omega_1$ is a disjoint union. 

For $i\in \{0\}\cup [m]$, let $f_i$ be the $q$-linearized polynomial defined in \eqref{eq:fi}. That is,
\[
f_i(X)=\prod_{\alpha\in V_i}(X-\alpha)^{q^{k-\dim(V_i)-r_i}}\in (\F_q[V_i])[X].
\]
We now build from $\cS_1=((V_i', r_i))_{i\in \Omega_1}$ a new tuple $\cS_1^*$, defined as
\[
\cS_1^*=((f_0(V_i'), r_i))_{i\in \Omega_1}
\]
Denote by $f|_V$ the isomorphism $a\mapsto f(a)$ from $V$ to $f(V)$. Let $V'=f|_V^{-1}(\overline{V}_0)\subseteq V$ and $W'=W\oplus f|_V^{-1}(V_0)$. Then $\mathrm{span}_{\F_q}\{Z_1,\dots,Z_n\}=V'\oplus W'$, $f(V')=\overline{V}_0$, and $f(W')=f(W)\oplus V_0$.
Here $f$ is a nonzero $q$-linearized polynomial with coefficients in $\F_q[W]$.
And $f_0$ is a nonzero $q$-linearized polynomial with coefficients in $\F_q[V_0]$.
As $V_0\subseteq f(W')$ and the coefficients of $f$ are in $\F_q[W]\subseteq\F_q[W']$, we see that $\F_q[V_0]\subseteq\F_q[W']$.
It follows that $f_0\circ f$ is a nonzero $q$-linearized polynomial with coefficients in $\F_q[W']$.
Note that for $i\in \Omega_1$, we have 
\[
f_0(V_i')\subseteq f_0(\overline{V}_0)=(f_0\circ f)(V'),
\]
where $\dim V'\leq \dim V\leq s$.
By definition, the tuple $(f_0(V_i'))_{i\in\Omega_1}$ is $s$-admissible.

Also, note that $f_0$ defines an isomorphism from $\overline{V}_0$ to $f(\overline{V}_0)$, mapping each $V_i'$ to $f_0(V_i')$.
As $\cS_1=((V_i', r_i))_{i\in \Omega_1}$ is in $V_{r_0,|\Omega_1|,s}$ and satisfies \eqref{eqn:key-new}, the same holds for $\cS_1^*=((f_0(V_i'), r_i))_{i\in \Omega_1}$.
Applying the induction hypothesis to $\cS_1^*$ and $\cS_2$ shows that $\det(M_{\cS_1^*})$ and $\det(M_{\cS_2})$ are nonzero.

For $i\in \Omega_1$, 
let $t_i
=r_0-\dim(V_i')-r_i
=(\dim (V_0)+r_0)-(\dim (V_i)+r_i)
\stackrel{\eqref{eq:gkp-tight-2}}{\geq} 0$, and define the $q$-linearized polynomial
\[
f_i^*(X)=\prod_{\alpha\in f_0(V_i')}(X-\alpha)^{q^{t_i}}=\prod_{\alpha\in f_0(V_i')}(X-\alpha)^{q^{r_0-\dim(V_i')-r_i}}
\]
whose $q$-degree is
\begin{equation}\label{eq:q-degree-fi}
\deg_q(f_i^*)=\dim(V_i')+t_i=\dim(V_i)-\dim(V_0)+t_i=r_0-r_i.
\end{equation}
For $i\in\Omega_1$,
\begin{equation}\label{eq:compose}
\begin{aligned}
f_i(X)&=\prod_{\alpha\in V_i}(X-\alpha)^{q^{k-\dim(V_i)-r_i}}
=\prod_{\beta\in V_i'}\prod_{\alpha\in V_0}(X-\alpha-\beta)^{q^{k-\dim(V_0)-r_0+t_i}}\\
&=\prod_{\beta\in V_i'} f_0(X-\beta)^{q^{t_i}}
=\prod_{\beta\in V_i'} (f_0(X)-f_0(\beta))^{q^{t_i}}\\
&=\prod_{\alpha\in f_0(V_i')} (f_0(X)-\alpha)^{q^{t_i}}=(f_i^* \circ f_0)(X).
\end{aligned}
\end{equation}

We now prove $\det(M_{\cS})\neq 0$ by applying \cref{lm:equivalent-new}. Consider arbitrary $q$-linearized polynomials $g_1,\dots,g_m\in \mathbb{F}[X]$ such that the $q$-degree of $q_i(X)$ is at most $r_i-1$ for $i\in [m]$ and $\sum_{i=1}^m g_i\circ f_i=0$. 
Define $g_0=\sum_{i\in \Omega_1}g_i\circ f_i^*$. 
Note that 
\[
\deg_q (g_0)
\leq \max_{i\in \Omega_1} \left\{\deg_q (f_i^*)+\deg_q (g_i)\right\}
\stackrel{\eqref{eq:q-degree-fi}}{\leq} \max_{i\in \Omega_1} \{(r_0-r_i)+(r_i-1)\}=r_0-1.
\]
Also, it holds that 
\begin{equation}\label{eq:gf-zero}
0=\sum_{i=1}^m g_i\circ f_i\stackrel{\eqref{eq:compose}}{=}\sum_{i\in \Omega_1}g_i\circ f_i^*\circ f_0+\sum_{i\in [m]\setminus \Omega_1} g_i\circ f_i=\sum_{i\in \Omega_2}g_i\circ f_i,
\end{equation}
where the last equality holds since 
$\sum_{i\in \Omega_1}g_i\circ f_i^*=g_0$.
By \cref{lm:equivalent-new}, \eqref{eq:gf-zero}, and the fact that $\det(M_{\cS_2})\neq 0$, we have $g_i=0$ for all $i\in \Omega_2$. In particular, $g_0=0$.
Therefore,
\begin{equation}\label{eq:gf-zero-2}
\sum_{i\in \Omega_1}g_i\circ f_i^*=g_0=0.
\end{equation}
By \cref{lm:equivalent-new}, \eqref{eq:gf-zero-2}, and the fact that $\det(M_{\cS_1^*})\neq 0$, we have $g_i=0$ for all $i\in \Omega_1$.
So $g_i=0$ for all $i\in [m]$. 
By \cref{lm:equivalent-new}, we conclude that $\det(M_{\cS})\neq 0$.

\paragraph{Case 2:} 
There exists an unique integer $i\in [m]$ such that $\dim(V_i)+r_i=k$. 

Without loss of generality, suppose $i=m$. 
Let $M_1,\dots,M_m$ be as in \eqref{eq:mi}.
Since 
\begin{equation}\label{eq:not-tight}
\dim(V_m)+r_m=k>\dim(V_i)+r_i \quad\text{ for } i=1,\dots,m-1,
\end{equation}
by the definitions \eqref{eq:fi}, \eqref{eq:ms}, and \eqref{eq:mi}, the entries in first column of $M_{\cS}$ are all zero except that the $(1,1)$-th entry of the block $M_m$ is the coefficient of $X$ in $f_m=\prod_{\alpha\in V_m} (X-\alpha)$, which we denote by $c\in\F$. Note that $c=\prod_{\alpha\in V_m/\{0\}}(-\alpha)\neq 0$. Let $M'\in\F^{(k-1)\times (k-1)}$ be the submatrix of $M_{\cS}$ obtained by removing the first column of $M_{\cS}$ and the first row of the block $M_m$.
By the above observation about $M_{\cS}$,
\begin{equation}\label{eq:submatrix}
\det(M_{\cS})=\pm c\cdot \det(M').
\end{equation}

Assume $r_m>1$. Let $\cS'=((V_i, r_i'))_{i\in [m]}$, where $r_i'=r_i$ for $i\in [m-1]$ and $r_m'=r_m-1$.
By \eqref{eq:not-tight} and the fact $\sum_{i=1}^m r'_i=k-1$, we have $\cS'\in V_{k-1,m,s}$. Next, we verify that  $\cS'$ satisfies \eqref{eqn:key-new} for nonempty $\Omega\subseteq [m]$. If $m\not\in\Omega$, this holds since $\cS$ satisfies \eqref{eqn:key-new} and $r_i'=r_i$ for $i\in [m-1]$. On the other hand, if $m\in\Omega$, then
\[
\dim\left(\bigcap_{i\in \Omega} V_i\right)+\sum_{i\in \Omega}r'_i=\dim\left(\bigcap_{i\in \Omega} V_i\right)+\left(\sum_{i\in \Omega}r_i\right)-1\stackrel{\eqref{eqn:key-new}}{\leq} \max_{i\in \Omega}\{\dim(V_i)+r_i\}-1
\stackrel{\eqref{eq:not-tight}}{=}\max_{i\in \Omega}\{\dim(V_i)+r'_i\}.
\]
So again $\cS'$ satisfies \eqref{eqn:key-new}.
By the induction hypothesis, $\det(M_{\cS'})\neq 0$.

The matrix $M_{\cS'}$ consists of the blocks $M'_i$ for $i\in [m]$, placed vertically, each given by
\[
M_i'=\begin{pmatrix}
b_{i,0}&b_{i,1}&\cdots& b_{i,(k-1)-r_i'-1} & b_{i,(k-1)-r_i'}&  0 & \cdots &0   \\
0 &b_{i,0}^q& \cdots  & b_{i,(k-1)-r_i'-2}^q & b_{i,(k-1)-r_i'-1}^q & b_{i,(k-1)-r_i'}^q & \cdots &0 \\
\vdots & \vdots &\ddots & \vdots &\vdots &\vdots& \ddots &0\\
0 &0  & \cdots &b_{i,0}^{q^{r_i'-1}} & b_{i,1}^{q^{r_i'-1}} & b_{i,2}^{q^{r_i'-1}} &\cdots & b_{i,(k-1)-r_i'}^{q^{r_i'-1}} 
\end{pmatrix}\in\mathbb{F}^{r_i'\times (k-1)}
\]
where each $b_{i,j}$ is determined via 
\begin{equation}\label{eq:bij}
\prod_{\alpha\in V_i}(X-\alpha)^{q^{(k-1)-\dim(V_i)-r_i'}}=\sum_{j=0}^{(k-1)-r_i'}b_{i,j} X^{q^j},
\end{equation}
or equivalently
\begin{equation}\label{eq:bij2}
\prod_{\alpha\in V_i}(X-\alpha)^{q^{k-\dim(V_i)-r_i'}}=\sum_{j=1}^{k-r_i'}b_{i,j-1}^q X^{q^j}.
\end{equation}
By \eqref{eq:bij} and \eqref{eq:bij2}, we have $b_{i,j}=a_{i,j}$ if $i=m$ and $b_{i,j}=a_{i,j+1}^{1/q}$ if $i\in [m-1]$. By the definition of $M'$, we see that each entry of $M'$ is exactly the $q$-th power of the corresponding entry of $M_{\cS'}$. As the map $x\mapsto x^q$ is a ring endomorphism of $\F$, this implies $\det(M')=\det(M_{\cS'})^q$.
As $\det(M_{\cS'})$ is nonzero, so is $\det(M')$.
Then $\det(M_{\cS})\neq 0$ by \eqref{eq:submatrix}.

Now assume $r_m=1$. Let $\cS'=((V_i, r_i))_{i\in [m-1]}$. In this case, we have $S'\in V_{k-1,m-1,s}$ by \eqref{eq:not-tight} and the fact that $\sum_{i=1}^{m-1} r_i=k-1$. 
The tuple $\cS'$ satisfies \eqref{eqn:key-new} since $\cS$ does. By the induction hypothesis, we have $\det(M_{\cS'})\neq 0$.
Also, similar to the case where $r_m>1$, each entry of $M'$ is again the $q$-th power of the corresponding entry of $M_{\cS'}$, and this implies $\det(M')=\det(M_{\cS'})^q\neq 0$.\footnote{The only difference is that when $r_m>1$, we replace $(V_m, r_m)$ by $(V_m, r_m-1)$; whereas when $r_m=1$, we remove the pair $(V_m, r_m)=(V_m,1)$ and also decrease $m$ by one.}
Again, we conclude that $\det(M_{\cS})\neq 0$ by \eqref{eq:submatrix}.

\paragraph{Case 3:} Neither Case 1 nor Case 2 holds. In other words, for any nonempty $\Omega\subseteq [m]$ such that $2\leq |\Omega|\leq m-1$, 
\begin{equation}\label{eq:non-tightness-2}
\dim\left(\bigcap_{i\in \Omega}V_i\right)+\sum_{i\in \Omega}r_i \leq \max_{i\in \Omega}\{\dim(V_i)+r_i\}-1
\end{equation}
and there exist distinct $i_1,i_2\in [m]$ such that 
\begin{equation}\label{eq:attain-max}
\dim(V_{i_1})+r_{i_1}=\dim(V_{i_2})+r_{i_2}=\max_{i\in [m]}\{\dim (V_i)+r_i\}. 
\end{equation}
Without loss of generality, assume $i_1=m-1$ and $i_2=m$. First, note that if $V_{m-1}=V_m$, then 
\[
\dim(V_{m-1}\cap V_m)+r_{m-1}+r_m=\dim(V_m)+r_{m-1}+r_m>\max\{\dim(V_{m-1})+r_{m-1},\dim(V_m)+r_m\}
\]
contradicting \eqref{eqn:key-new} with $\Omega=\{m-1,m\}$.
So $V_{m-1}\neq V_m$.

As $(V_i)_{i\in [m]}$ is $s$-admissible, there exist $\F_q$-subspaces $V,W\subseteq \mathrm{span}_{\F_q}\{Z_1,\dots,Z_{n}\}$ and a nonzero $q$-linearized polynomial $f\in (\F_q[W])[X]$ in $X$ such that $\mathrm{span}_{\F_q}\{Z_1,\dots,Z_{n}\}=V\oplus W$, $\dim V\leq s$, and
$V_1,\dots,V_m\subseteq f(V)$.
As $V_{m-1}\neq V_m$, either $V_{m-1}\neq f(V)$ or $V_m\neq f(V)$.
Without loss of generality, assume $V_m\neq f(V)$.
Note that applying an invertible linear transformation to the variables $Z_1,\dots,Z_n$ over $\F_q$ induces a ring isomorphism of $\F_q[Z_1,\dots,Z_n]\ni \det(M_{\cS})$, which preserves the (non)zeroness of $\det(M_{\cS})$. By applying such a linear transformation, we may assume $V=\mathrm{span}_{\F_q}\{Z_1,\dots,Z_d\}$, $f|_V^{-1}(V_m)\subseteq \mathrm{span}_{\F_q}\{Z_1,\dots,Z_{d-1}\}$ and $W=\mathrm{span}_{\F_q}\{Z_{d+1},\dots,Z_n\}$, where $d:=\dim V\leq s$.

Let $\pi:V=\mathrm{span}_{\F_q}\{Z_1,\dots,Z_{d}\}\to \mathrm{span}_{\F_q}\{Z_1,\dots,Z_{d-1}\}$ be the projection sending $u=\sum_{i=1}^d c_i Z_i$ to $u|_{Z_d=0}=\sum_{i=1}^{d-1} c_i Z_i$.
As $f\in(\F_q[W])[X]$ does not depend on $Z_d$, we have 
\begin{equation}\label{eq:restriction}
f(\alpha)|_{Z_d=0}=f(\alpha|_{Z_d=0})=f(\pi(\alpha)) \quad\text{ for } \alpha\in V.   
\end{equation}
For $i\in [m]$, let $\widetilde{V}_i=f|_V^{-1}(V_i)$, $\widetilde{V}_i'=\pi(\widetilde{V}_i)\subseteq \mathrm{span}_{\F_q}\{Z_1,\dots,Z_{d-1}\}$ and $V_i'=f(\widetilde{V}_i')=f(\pi(\widetilde{V}_i))$. 
Note that
\begin{equation}\label{eq:unchanged}
V_m'=f(\pi(\widetilde{V}_m))=f(\widetilde{V}_m)=V_m
\end{equation}
where the second equality holds since $\widetilde{V}_m=f|_V^{-1}(V_m)\subseteq \mathrm{span}_{\F_q}\{Z_1,\dots,Z_{d-1}\}$.
More generally, for $i\in [m]$,
\begin{equation}\label{eq:decrease-by-one}
\dim V_i'=\dim f(\pi(\widetilde{V}_i))=\dim \pi(\widetilde{V}_i)\geq \dim(\widetilde{V}_i)-1=\dim (V_i)-1.
\end{equation}
For $i\in [m]$, we have
\begin{equation}\label{eq:assign}
\begin{aligned}
\prod_{\alpha\in V_i} (X-\alpha|_{Z_d=0})^{q^{k-\dim (V_i)-r_i}}&=\prod_{\alpha\in \widetilde{V}_i} (X-f(\alpha)|_{Z_d=0})^{q^{k-\dim (V_i)-r_i}}\\
&\stackrel{\eqref{eq:restriction}}{=}\prod_{\alpha\in \widetilde{V}_i} (X-f(\pi(\alpha)))^{q^{k-\dim (V_i)-r_i}}\\
&=\prod_{\alpha\in V_i'} (X-\alpha)^{q^{k-\dim (V_i)-r_i+\dim(\widetilde{V}_i)-\dim(V_i')}}\\
&=\prod_{\alpha\in V_i'} (X-\alpha)^{q^{k-\dim(V_i')-r_i}}.
\end{aligned}
\end{equation}

Let $V'=\mathrm{span}_{\F_q}\{Z_1,\dots,Z_{d-1}\}\subseteq V$ and $W'=\mathrm{span}_{\F_q}\{Z_d,\dots,Z_n\}\supseteq W$.
Then $V'_i=f(\pi(\widetilde{V}_i))\subseteq f(V')$ for $i\in [m]$. Note that $f$ is a $q$-linearized polynomial with coefficients in $\F_q[W]\subseteq \F_q[W']$.
Also note that $\dim V'=d-1\leq s-1$.
By definition, $(V_i')_{i\in [m]}$ is $(s-1)$-admissible.
Let $\cS'=(V_i', r_i)_{i\in [m]}$.
As $\dim V_i'\leq \dim \widetilde{V}_i=\dim V_i$ for $i\in [m]$, we have by definition that $\cS'\in V_{k,m,s-1}$.

Next, we verify that $\cS'$ satisfies \eqref{eqn:key-new} for all nonempty $\Omega\subseteq [m]$.
If $|\Omega|=1$, then \eqref{eqn:key-new} holds trivially for $\cS'$.
For $\Omega=[m]$, \eqref{eqn:key-new} also holds for $\cS'$ since
\begin{align*}
\dim\left(\bigcap_{i\in [m]} V'_i\right)+\sum_{i\in[m]} r_i
\leq \dim\left(\bigcap_{i\in [m]} V_i\right)+\sum_{i\in[m]} r_i&\leq \max_{i\in [m]}\{\dim (V_i)+r_i\}\\
&\stackrel{\eqref{eq:attain-max}}{=}\dim (V_m)+r_m\\
&\stackrel{\eqref{eq:unchanged}}{=}\dim(V_m')+r_m\\
&\leq \max_{i\in [m]}\{\dim (V_i')+r_i\},
\end{align*}
where the second inequality holds since $\cS$ satisfies \eqref{eqn:key-new} with $\Omega=[m]$.
Finally, for nonempty $\Omega\subseteq [m]$ with $2\leq |\Omega|\leq m-1$, \eqref{eqn:key-new} holds for $\cS'$ since
\begin{align*}
\dim\left(\bigcap_{i\in \Omega} V'_i\right)+\sum_{i\in\Omega} r_i
\leq \dim\left(\bigcap_{i\in \Omega} V_i\right)+\sum_{i\in\Omega} r_i
&\stackrel{\eqref{eq:non-tightness-2}}{\leq} \max_{i\in \Omega}\{\dim (V_i)+r_i\}-1\\
&\stackrel{\eqref{eq:decrease-by-one}}{\leq} \max_{i\in [m]}\{\dim (V_i')+r_i\}.
\end{align*}
By \eqref{eq:assign} and the definitions \eqref{eq:fi}, \eqref{eq:ms}, and \eqref{eq:mi}, it holds that $\det(M_{\cS'})=\det(M_{\cS})|_{Z_d=0}$.
And by the induction hypothesis, we have $\det(M_{\cS'})\neq 0$.
So $\det(M_{\cS})\neq 0$.
\end{proof}


\section{Equivalence Between $\operatorname{GKP}(\ell)$ and $\operatorname{MRD}(\ell)$}\label{ss5}

In this section, we will establish the equivalence between $\operatorname{GKP}(\ell)$ and $\operatorname{MRD}(\ell)$ over a general field $\F/\F_q$. As a key technical ingredient, we will prove our formula for generic intersection dimension (see \cref{gintersec}), 
generalizing a similar formula established in \cite{brakensiek2023generic}. Our proofs closely follow the arguments of Brakensiek--Gopi--Makam \cite{brakensiek2023generic}.

\subsection{Formula for Generic Intersection Dimension: Proof of \cref{gintersec}}
We first show that a generic linear code attains all generic kernel patterns. Our proof of this fact follows closely that of \cite[Proposition 2.1]{brakensiek2023generic}. 

\begin{prop}\label{prop:GZP}
Let $W$ be the $k\times n$ matrix $(Z_{i,j})_{i\in [k], j\in [n]}$ over $\F:=\F_q(Z_{1,1},\dots,Z_{k,n})$.
Let $C\subseteq \F^n$ be the $[n,k]_\F$ code with $W$ as a generator matrix.
Then $C$ attains any generic kernel pattern. In other words, $C$ is $\mathrm{GKP}(\ell)$ for all $\ell\geq 1$. 
\end{prop} 

\begin{proof}
Let $(V_1,\dots,V_k)$ be a generic kernel pattern, i.e., $\dim\left(\bigcap_{i\in \Omega}V_i\right)\leq k-|\Omega|$ for any nonempty $\Omega\subseteq [k]$.  By \cref{halldual1}, we can find $V'_i\subseteq \F_q^n$ of dimension $k-1$ such that $V_i\subseteq V'_i$ and $\dim(\bigcap_{i\in \Omega}V'_i)\leq k-|\Omega|$ for any nonempty $\Omega\subseteq [k]$. By replacing $V_i$ with $V_i'$, we may assume that $\dim V_i=k-1$. 

For $i\in [k]$, choose $A_i\in \F_q^{n\times \dim V_i}$ such that $V_i=\langle A_i\rangle$. 
We want to show that there exists an invertible matrix $M\in \F^{k\times k}$, whose $i$-th row is denoted by $\vm_i$, such that $\vm_i W A_i=0$ for $i\in [k]$. 
Consider arbitrary $i\in [k]$.
Since $\dim V_i=k-1$, we have $\rank(WA_i)=k-1$ by \cref{generic_a}. This implies that the solution space $\vm_i$ is one-dimensional, i.e., $\vm_i$ is uniquely determined up to scaling. 
In fact, we may write each entry of $\vm_i$ as a polynomial in the entries of a nonsingular maximal minor $P$ of $WA_i$ by applying Cramer's rule and clearing the common denominator $\det(P)$. As this holds for all $i\in [k]$, $\det(M)$ may be expressed as a polynomial in $Z_{ij}$. By the GM-MRD theorem (\cref{thm:GM-MRD-v1-new}), $\det(M)$ is nonzero even after assigning $Z_i^{q^{j-1}}$ to $Z_{i,j}$ for $i\in [k]$ and $j\in [n]$. So $\det(M)\neq 0$, i.e., $M$ is invertible.\footnote{One can also prove $\det(M)\neq 0$ directly by adapting the proof of \cref{thm:GM-MRD-v1-new}, avoiding the use of symbolic Gabidulin codes.} 
	\end{proof}

Next, we establish an intersection dimension formula for GKP($\ell$) codes.
Recall that for a $k\times n$ matrix $G$ over a field $\F/\F_q$ and an $n\times \ell$ matrix $A$ over $\F_q$, we denote $G_A:=GA\in \F^{k\times \ell}$ and define $G_V\subseteq \F^k$ to be the column span of $G_A$. 
 
\begin{thm}\label{interction_thm_G}
Let $\F$ be an extension field of $\F_q$. Let $C$ be an $[n, k]_{\F}$ code that is $\operatorname{GKP}(\ell)$ with a generator matrix $G\in\F^{k\times n}$. Let $V_1,\dots,V_\ell$ be subspaces of $\F_q^n$, each of dimension at most $k$. Then
\begin{equation}\label{intersec_fG}
\dim_{\F}\left(\bigcap_{i\in[\ell]}G_{V_i}\right)=\max_{P_1\sqcup P_2\sqcup \cdots\sqcup P_s=[\ell]}\left(\sum_{i\in[s]}\dim_{\F_q}\left(\bigcap_{j\in P_i}V_j\right)-(s-1)k\right).
\end{equation} 
\end{thm}
\begin{proof}
We first prove that LHS $\geq$ RHS. 
For any nonempty $S\subseteq [\ell]$, it holds that $G_{V_S}\subseteq \bigcap_{i\in S}G_{V_i}$, where $V_S:=\bigcap_{i\in S} V_i$. 
Therefore, for any partition $P_1\sqcup P_2\sqcup \cdots\sqcup P_s=[\ell]$, we have
\begin{equation}\label{eq:containment}
\bigcap_{i\in[s]}G_{V_{P_{i}}}\subseteq \bigcap_{i\in[\ell]}G_{V_{i}}.
\end{equation}
Also note that by \cref{lm:dualspace},
\begin{equation}\label{eq:intersection-sum}
\bigcap_{i\in [s]} G_{V_{P_i}}
=\left(\sum_{i\in[s]} G_{V_{P_i}}^{\perp}\right)^{\perp}.
\end{equation}
As $C$ is a $\operatorname{GKP}(\ell)$ code, $C$ is a MRD code. So we have $\dim G_{V_{P_i}}=\dim V_{P_i}$ for $i\in [s]$.
It follows that
\begin{align*}
\dim\left(\bigcap_{i\in[\ell]}G_{V_{i}}\right)
&\stackrel{\eqref{eq:containment}}{\geq}\dim\left(\bigcap_{i\in[s]}G_{V_{P_{i}}}\right)
\stackrel{\eqref{eq:intersection-sum}}{=} k-\dim\left(\sum_{i\in[s]} G_{V_{P_i}}^{\perp}\right)
\geq k-\sum_{i\in[s]}\dim\left(G_{V_{P_i}}^{\perp}\right)\\
&=k-\sum_{i\in[s]}\left(k-\dim V_{P_i}\right)
=\sum_{i\in[s]}\dim\left(\bigcap_{j\in P_i}V_j\right)-(s-1)k.
\end{align*}
So LHS $\geq$ RHS in \eqref{intersec_fG}.

Next, we prove that LHS $\leq$ RHS in \eqref{intersec_fG}. Denote RHS of \eqref{intersec_fG} by $d$, i.e.,
\[
d=\max_{P_1\sqcup P_2\sqcup \cdots\sqcup P_s=[\ell]}\left(\sum_{i\in[s]}\dim_{\F_q}\left(\bigcap_{j\in P_i}V_j\right)-(s-1)k\right).
\]
Then for all partitions $P_1\sqcup P_2\sqcup\cdots\sqcup P_s=[\ell]$,
\[
\sum_{i=1}^s\dim\left(\bigcap_{j\in P_i} V_j\right)\leq (s-1)k+d.
\]
Thus, by \cref{order-l-cha}, there exist integers $\delta_1,\ldots,\delta_\ell\geq 0$ such that $\sum_{i=1}^{\ell}\delta_i=k-d$ and for all nonempty $\Omega\subseteq [\ell]$, 
\[\dim\left(\bigcap_{i\in \Omega}V_i\right)\leq k-\sum_{i\in \Omega}\delta_i.\]
By \cref{prop:equivalent}, we know that the pattern $(T_1,\dots,T_k)$ with $\delta_i$ copies of $V_i$ for $i\in[\ell]$ and $d$ additional copies of $\{0\}$ is a generic kernel pattern. 
Without loss of generality, assume the $d$ copies of $\{0\}$ are $T_{k-d+1},\cdots,T_k$. For $i\in [k]$, choose $A_i\in\F_q^{n\times \dim T_i}$ such that $T_i=\langle A_i\rangle$.

As $C$ is $\operatorname{GKP}(\ell)$, there exists an invertible matrix $M\in \F^{k\times k}$ such that $\vm_i GA_i=0$ for all $i\in[k]$, where $\vm_i$ denotes the $i$-th row of $M$. As $M$ is invertible, we have $\dim_{\F}\left(\bigcap_{i\in[\ell]}(MG)_{V_i}\right)=\dim_{\F}\left(\bigcap_{i\in[\ell]}G_{V_i}\right)$.
Thus, to prove that LHS $\leq$ RHS in \eqref{intersec_fG}, i.e., $\dim_{\F}\left(\bigcap_{i\in[\ell]}G_{V_i}\right)\leq d$, it suffices to show that for any $z\in\bigcap_{i\in[\ell]}(MG)_{V_i}$, only the last $d$ coordinates of $z$ could be nonzero. 

Let $z=(z_1,\dots,z_k)\in\bigcap_{i\in[\ell]}(MG)_{V_i}$. Consider arbitrary $i\in [k-d]$. Then $T_i$ is a copy of $V_j$ for some $j\in [\ell]$. As $z\in (MG)_{V_j}=(MG)_{T_i}$, we have $z_i=\vm_i GA_i\bu$ for some $\bu\in \F^{\dim V_i}$.
But $\vm_i GA_i=0$. So $z_i=0$. This proves the claim that only the last $d$ coordinates of $z$ could be nonzero, thereby completing the proof. 
\end{proof}


Combining \cref{prop:GZP} and \cref{interction_thm_G} yields our formula for generic intersection dimension.

\begin{cor}[\cref{gintersec}, restated]
Let $W$ be the $k\times n$ matrix $(Z_{i,j})_{i\in [k], j\in [n]}$ over $\F:=\F_q(Z_{1,1},\dots,Z_{k,n})$.
Let $V_1,\dots,V_\ell$ be subspaces of $\F_q^n$, each of dimension at most $k$. Then
\begin{equation}\label{intersec1}
\dim_{\mathbb{F}}\left(\bigcap_{i\in[\ell]}W_{V_i}\right)=\max_{P_1\sqcup P_2\sqcup \cdots\sqcup P_s=[\ell]}\left(\sum_{i\in[s]}\dim_{\F_q}\left(\bigcap_{j\in P_i}V_j\right)-(s-1)k\right).
\end{equation} 
\end{cor}


\subsection{Proof of \cref{1equiv}}

We are now ready to prove \cref{1equiv}, which establishes an equivalence between $\mathrm{GKP}(\ell)$ and $\mathrm{MRD}(\ell)$. The direction from $\mathrm{GKP}(\ell)$ to $\mathrm{MRD}(\ell)$ follows easily from our intersection dimension formula for  $\mathrm{GKP}(\ell)$ codes.

\begin{thm}\label{thm:gkp2mrd}
Let $C$ be an $[n,k]_{\F}$ code over a field $\F/\F_q$.
If $C$ is $\mathrm{GKP}(\ell)$, then it is $\mathrm{MRD}(\ell)$.
\end{thm}
\begin{proof}
By \cref{prop:GZP}, the ``generic'' linear code over $\F_q(Z_{1,1},\dots,Z_{k,n})$ defined by the generator matrix $W=(Z_{i,j})_{i\in [k], j\in [n]}$ is also $\mathrm{GKP}(\ell)$.
Let $G\in \F^{k\times n}$ be a generator matrix of $C$.
By \cref{interction_thm_G}, for any subspaces $V_1,\dots,V_\ell\subseteq \F_q^n$ of dimension at most $k$, we have $\dim(\bigcap_{i\in [\ell]} G_{V_i})=\dim(\bigcap_{i\in [\ell]} W_{V_i})$, i.e., $C$ is $\mathrm{MRD}(\ell)$.
\end{proof}

Next, we prove the implication from $\mathrm{MRD}(\ell)$ to $\mathrm{GKP}(\ell)$.

\begin{thm}\label{mrd-gkp}
Let $\F$ be an extension field of $\F_q$. Let $C$ be an $[n, k]_{\F}$ code that is $\operatorname{MRD}(\ell)$.
Let $\mathcal{T}=(T_1,\dots,T_k)$ be a  generic kernel pattern of order at most $\ell$. Then $C$ attains $\mathcal{T}$. 
\end{thm}
\begin{proof}
By \cref{l-gkp-max}, by padding the subspaces if necessary, we may assume that there exist subspaces $V_1,\dots,V_\ell\subseteq\F_q^n$ and integers $\delta_1,\dots,\delta_\ell\geq 0$ 
such that $\dim V_i=k-\delta_i$ for all $i\in[\ell]$ and $\mathcal{T}$ consists of $\delta_i$ copies of $V_i$ and $d:=k-\sum_{i\in[\ell]}\delta_i$ additional copies of $\{0\}$.

Since $\mathcal{T}$ is a generic kernel pattern of order at most $\ell$, by \cref{order-l-cha}, for all partitions $P_1\sqcup P_2\sqcup\dots\sqcup P_s=[\ell]$,
\begin{equation}\label{eq:intersect-bound-1}
\sum_{i\in [s]}\dim\left(\bigcap_{j\in P_i}V_{j}\right)\leq (s-1)k+d.
\end{equation}
For the finest partition $\{1\}\sqcup\{2\}\sqcup\cdots\sqcup\{\ell\},$ the above inequality is indeed an equality because $\dim(V_i)=k-\delta_i$ and thus 
\begin{equation}\label{eq:intersect-bound-2}
\sum_{i=1}^{\ell} \dim(V_i)=sk-\sum_{i=1}^{\ell}\delta_i=(s-1)k+d.
\end{equation}

Let $G\in\F^{k\times n}$ be a generator matrix of $C$.
Then we have
\[
\dim\left(\bigcap_{i=1}^\ell G_{V_{i}}\right)=\dim\left(\bigcap_{i=1}^\ell W_{V_{i}}\right)=\sum_{i=1}^{\ell} \dim(V_i)-(s-1)k=d,
\]
where the first equality holds since $C$ is $\mathrm{MRD}(\ell)$ and the other equalities follow from \cref{gintersec}, \eqref{eq:intersect-bound-1}, and \eqref{eq:intersect-bound-2}.
So by \cref{lm:dualspace},
\begin{equation}\label{eq:dim-of-sum}
\dim \left(\sum_{i=1}^\ell G_{V_i}^{\perp}\right)=k-\dim\left(\bigcap_{i=1}^\ell G_{V_{i}}\right)=k-d.
\end{equation}
On the other hand, as $C$ is MRD, we have $\dim G_{V_i}^{\perp}=k-\dim G_{V_i}=k-\dim V_i=\delta_i$ for $i\in [\ell]$. Therefore
\[
k-d\stackrel{\eqref{eq:dim-of-sum}}{=}\dim \left(\sum_{i=1}^\ell G_{V_i}^{\perp}\right)\leq \sum_{i=1}^\ell \dim\left(G_{V_i}^\perp\right)= \sum_{i=1}^{\ell}\delta_i=k-d
\]
which implies that $\dim \left(\sum_{i=1}^\ell G_{V_i}^{\perp}\right)=\sum_{i=1}^\ell \dim\left(G_{V_i}^\perp\right)$. Therefore,, the sum $\sum_{i=1}^\ell G_{V_i}^{\perp}$ is a direct sum, i.e., $\sum_{i=1}^\ell G_{V_i}^{\perp}=\bigoplus_{i=1}^\ell G_{V_i}^{\perp}$.
For $i\in [\ell]$, let $B_i=\left\{\bv_1^{(i)},\dots,\bv_{\delta_i}^{(i)}\right\}$ be a basis of $G_{V_i}^{\perp}$.
Then $B:=B_1\sqcup \dots\sqcup B_\ell$ is a basis of $\sum_{i=1}^\ell G_{V_i}^{\perp}$. Extending $B$ to a basis $B'=B\sqcup\{\bu_1,\dots,\bu_d\}$ of $\F^k$.

To show that $C$ attains $\mathcal{T}$, by permuting the subspaces $T_i$, we may assume that
\[
\mathcal{T}=\left(V_1^{(1)},\dots,V_1^{(\delta_1)},\dots,V_{\ell}^{(1)},\dots,V_{\ell}^{(\delta_{\ell})}, \{0\},\dots\{0\}\right),
\]
where $V_i^{(j)}$ denotes the $j$-th copy of $V_i$.
Choose $M \in \F^{k\times k}$ such that the first $k-d$ columns of $M^T$ are $\bv^{(i)}_j$ for $i\in [\ell]$ and $j\in [\delta_i]$ (given in lexicographic order of $(i,j)$), and the last $d$ columns of $M^T$ are $\bu_1,\dots,\bu_d$.
Then $M$ is invertible since the columns of $M^T$ form a basis of $\F^k$.
By the construction of $M$, we have $\vm_i GA_i=0$ for all $i\in [k]$, where $\vm_i$ denotes the $i$-th row of $M$ and $A_i\in \F_q^{n\times \dim T_i}$ satisfies $T_i=\langle A_i\rangle$. 
Therefore, $C$ attains the kernel pattern $\mathcal{T}$.
\end{proof}


\begin{cor}\label{cor:mrd2gkp}
Let $C$ be an $[n,k]_{\F}$ code over a field $\F/\F_q$.
If $C$ is $\mathrm{MRD}(\ell)$, then it is $\mathrm{GKP}(\ell)$.
\end{cor}

Combining \cref{thm:gkp2mrd} and \cref{cor:mrd2gkp} proves \cref{1equiv}.
\section{Equivalence Between $\operatorname{MRD}(\ell)$ and $\ldmrd{\ell-1}$ up to Duality}

In this section, we will show that a linear code $C$ is $\operatorname{MRD}(\ell)$ if and only if its dual code is $\ldmrd{\ell-1}$. Our proofs resemble the arguments of Brakensiek--Gopi--Makam \cite{brakensiek2023generic}. 

\subsection{An Alternative Characterization of $\operatorname{MRD}(\ell)$ Over $\F/\F_q$}

We start by providing an alternative characterization of $\operatorname{MRD}(\ell)$ codes. First, we need the following lemma.

\begin{lemma}\label{fgva}
Let $\F$ be an extension field of $\F_q$ and let $G\in\F^{k\times n}$. For $i=1,\dots,\ell$, let $V_i$ be a subspace of $\F_q^n$ and let $A_i\in\F_q^{n\times \dim V_i}$ such that $V_i=\langle A_i\rangle$.
Then
\begin{equation}\label{eq:null-intersection}
\dim\left(\bigcap_{i\in[\ell]}G_{V_i}\right)=\sum_{i\in[\ell]}\dim  G_{V_i}-\rank\left(G_{\{A_i\}_{i\in[\ell]}}\right),    
\end{equation}
where we define the matrix $G_{\{A_i\}_{i\in[\ell]}}:=
\begin{pmatrix}
GA_1 & GA_2 &  &  &  \\GA_1 &  & GA_3 &  & \\
\vdots &  &  & \ddots &  \\GA_1 &  &  &  & GA_\ell
\end{pmatrix}$.
\end{lemma}
\begin{proof}
If $\dim G_{V_i}<\dim V_i$ for some $i\in [\ell]$, then some columns of $GA_i$ are linear combinations of other columns. In this case, we may delete the corresponding columns from $A_i$ and update $V_i$ correspondingly without affecting the two sides of \eqref{eq:null-intersection}.
By repeatedly performing the deletions, we may assume that $\dim G_{V_i}=\dim V_i$ for all $i\in [\ell]$. 
In particular, RHS of \eqref{eq:null-intersection} equals the dimension of  
\[
U:=\left\{\bu\in\F^{\sum_{i\in[\ell]}\dim V_i}\text{ such that } G_{\{A_i\}_{i\in[\ell]}}\cdot \bu=0\right\}
\]
since the number of linearly independent constraints in the linear system $G_{\{A_i\}_{i\in[\ell]}}\cdot \bu=0$ is exactly $\rank\left(G_{\{A_i\}_{i\in[\ell]}}\right)$.

It remains to find an isomorphism $\sigma: U\to \bigcap_{i\in[\ell]}G_{V_i}$ between the two $\F$-linear spaces.
Let $\sigma$ send $\bu=(\bu_1,\bu_2,\dots,\bu_\ell)\in U$ with $\bu_i\in \F^{\dim V_i}$ to $-GA_1\cdot \bu_1$. 
Note that the definition of $U$ and that of  $G_{\{A_i\}_{i\in[\ell]}}$ imply that for $\bu=(\bu_1,\bu_2,\dots,\bu_\ell)\in U$,
\begin{equation}\label{1z1}
-GA_1\cdot\bu_1=GA_2\cdot\bu_2=\cdots=GA_\ell\cdot\bu_\ell.
\end{equation}
It follows that the image of $U$ under $\sigma$ is indeed contained in $\bigcap_{i\in[\ell]}G_{V_i}$.
Moreover, for any $\by\in \bigcap_{i\in[\ell]}G_{V_i}$, by the fact that $\dim G_{V_i}=\dim V_i$ for $i\in [\ell]$, there exists unique $(\bu_1,\dots,\bu_\ell)$ with $\bu_i\in\F^{\dim V_i}$ such that $\by=-GA_1\cdot\bu_1=GA_2\cdot\bu_2=\cdots=GA_\ell\cdot\bu_\ell$.
The map $\by\mapsto (\bu_1,\dots,\bu_\ell)$ is then the inverse of $\sigma$.
\end{proof} 

\begin{cor}\label{cor:semicontinuity}
Let $\F$ be an extension field of $\F_q$.
Let $C$ be an $[n,k]_{\F}$ code that is $\operatorname{MRD}$, and let $G\in\F^{k\times n}$ be a generator matrix of $C$.
Let $V_1,V_2,\dots,V_\ell$ be subspaces of $\F_q^n$, each of dimension at most $k$.
Then $\dim\left(\bigcap_{i=1}^{\ell}G_{V_i}\right)\geq \dim \left(\bigcap_{i=1}^{\ell}W_{V_i}\right)$, where $W=(Z_{i,j})_{i\in [k], j\in [n]}\in\F(Z_{1,1},\dots,Z_{k,n})^{k\times n}$.
\end{cor}

\begin{proof}
For $i\in [\ell]$, choose $A_i\in\F_q^{n\times \dim V_i}$ such that $V_i=\langle A_i\rangle$.
As $C$ is MRD, we have $\dim G_{V_i}=\dim V_i$ for $i\in [\ell]$. By \cref{generic_a}, we also have $\dim W_{V_i}=\dim V_i$ for $i\in [\ell]$. Then by \cref{fgva}, it suffices to prove that $\operatorname{rank}\left(G_{\{A_i\}_{i\in[\ell]}}\right) < \operatorname{rank}\left(W_{\{A_i\}_{i\in[\ell]}}\right)$. This holds since $G_{\{A_i\}_{i\in[\ell]}}$ can be obtained from $W_{\{A_i\}_{i\in[\ell]}}$ by assigning the $(i,j)$-th entry of $G$ to $Z_{i,j}$ for $(i,j)\in [k]\times [n]$. Such an assignment does not increase the rank since if the determinant of a submatrix is nonzero after the assignment, then it must have been nonzero before the assignment.
\end{proof}


The next lemma presents an alternative characterization of $\operatorname{MRD}(\ell)$ codes, which appears weaker but is, in fact, equivalent.

\begin{lemma}\label{null_inter_prop}
Let $\F$ be an extension field of $\F_q$.
Let $C$ be an $[n,k]_{\F}$ code with a generator matrix
$G\in \F^{k\times n}$. Let $\ell\geq 1$. Then the following are equivalent:
\begin{enumerate}
\item\label{item:mrd-generic} $C$ is $\operatorname{MRD}(\ell)$.
\item\label{item:mrd-ell} $C$ is $\operatorname{MRD}$ and for all subspaces $V_1,V_2,\dots,V_\ell\subseteq \F_q^n$, each of dimension at most $k$, we have $\bigcap_{i=1}^{\ell}G_{V_i}=0$ iff $\bigcap_{i=1}^{\ell}W_{V_i}=0$, where $W=(Z_{i,j})_{i\in [k], j\in [n]}\in\F(Z_{1,1},\dots,Z_{k,n})^{k\times n}$.
\end{enumerate}
\end{lemma}
\begin{proof}
\cref{item:mrd-generic} implies \cref{item:mrd-ell} by the definition of $\operatorname{MRD}(\ell)$.
Conversely, we show that if \cref{item:mrd-generic} does not hold, neither does \cref{item:mrd-ell}. 

Suppose $C$ is not $\operatorname{MRD}(\ell)$. 
Further assume that $C$ is MRD since otherwise \cref{item:mrd-ell} certainly does not hold.
Then there exist subspaces $V_1,V_2,\dots,V_\ell\subseteq\F_q^n$, each of dimension at most $k$, such that $\dim\left(\bigcap_{i=1}^{\ell}G_{V_i}\right)\neq \dim\left(\bigcap_{i=1}^{\ell}W_{V_i}\right)$. 
By \cref{cor:semicontinuity}, we have
\begin{equation}\label{eq:intersection-lower-bound}
\dim\left(\bigcap_{i=1}^\ell G_{V_i}\right)>d:=\dim\left(\bigcap_{i=1}^\ell W_{V_i}\right).
\end{equation}
We know by \cref{fgva} and \cref{generic_a} that
\begin{equation}\label{ffo1}
    \operatorname{rank}\left(W_{\{A_i\}_{i\in[\ell]}}\right)=\left(\sum_{i\in[\ell]}\dim  W_{V_i}\right)-d=\left(\sum_{i\in[\ell]}\dim V_i\right)-d.
\end{equation}
This means we can find $d$ columns $\mathbf{c}_1,\dots,\mathbf{c}_d$ of $W_{\{A_i\}_{i\in[\ell]}}$ that can be written as linear combinations of the remaining columns. 
Remove the corresponding $d$ columns in total from $A_1,\dots,A_{\ell}$ to obtain $A_1',\dots,A_{\ell}'$, where $A_i'$ is a submatrix of $A_i$, so that $W_{\{A_i'\}_{i\in[\ell]}}$ is obtained from $W_{\{A_i\}_{i\in[\ell]}}$ by removing $\mathbf{c}_1,\dots,\mathbf{c}_d$.
Then $\rank\left(W_{\{A_i'\}_{i\in[\ell]}}\right)=\rank\left(W_{\{A_i\}_{i\in[\ell]}}\right)$. Let $V_i'=\langle A_i' \rangle$ for $i\in [\ell]$. Then $\sum_{i\in [\ell]}\dim V_i'=\left(\sum_{i\in [\ell]}\dim V_i\right)-d$.
By \cref{fgva} and \cref{generic_a}, we have
\begin{align*}
\dim \left(\bigcap_{i=1}^{\ell}W_{V_i'}\right)
&=\left(\sum_{i\in[\ell]}\dim  W_{V_i'}\right)-\rank\left(W_{\{A_i'\}_{i\in[\ell]}}\right)
=\left(\sum_{i\in[\ell]}\dim  V_i'\right)-\rank\left(W_{\{A_i'\}_{i\in[\ell]}}\right)\\
&=\left(\sum_{i\in [\ell]}\dim V_i\right)-d-\rank\left(W_{\{A_i\}_{i\in[\ell]}}\right)\stackrel{\eqref{ffo1}}{=} 0.
\end{align*}
On the other hand, removing a column from $A_i$ decreases $\dim V_i=\dim \langle A_i\rangle$ by one and decreases $\rank\left(G_{\{A_i\}_{i\in[\ell]}}\right)$ by zero or one.
It follows from \eqref{eq:null-intersection} that
\[
\dim\left(\bigcap_{i\in[\ell]}G_{V_i'}\right)
\geq \dim\left(\bigcap_{i\in[\ell]}G_{V_i}\right)-d
\stackrel{\eqref{eq:intersection-lower-bound}}{>} d-d=0.
\]
So \cref{item:mrd-ell} does not hold.
\end{proof}

\subsection{Proof of Theorem \ref{2equiv}}

We now prove \cref{2equiv}, which establishes an equivalence between $\ldmrd{\ell}$ and $\mathrm{MRD}(\ell+1)$ up to duality. 
First, we show the implication from $\ldmrd{\ell}$ to $\mathrm{MRD}(\ell+1)$.

\begin{thm}\label{thm:ldmrd2mrd}
Let $C$ be an $[n,k]_{\F}$ code over a field $\F/\F_q$.
If $C$ is $\ldmrd{\ell}$, then $C^{\perp}$ is $\operatorname{MRD}(\ell+1)$.
\end{thm}

\begin{proof}
We will prove the contrapositive: If  $C^{\perp}$ is not MRD$(\ell+1)$, then $C$ is not $\ldmrd{\ell}$. 
If $C^{\perp}$ is not MRD, then $C$ is not either by \cref{prop:dualmrd}. This, in turn, implies that $C$ is not $\ldmrd{\ell}$, since all $\ldmrd{\ell}$ codes are MRD. 
So we may assume that $C^{\perp}$ is MRD.

Let $H\in \F^{(n-k)\times n}$ be a parity check matrix of $C$. Then $H$ is also a generator matrix of $C^{\perp}$. 
By \cref{cor:semicontinuity} and \cref{null_inter_prop}, the fact that $C^{\perp}$ is MRD but not MRD$(\ell+1)$ implies that there exist subspaces $V_0,\ldots,V_\ell\subseteq \F_q^n$, each of dimension at most $n-k$, such that 
\begin{equation}
\bigcap_{0\leq i\leq \ell} H_{V_i}\neq \{0\}\quad \text{and} \quad \bigcap_{0\leq i\leq \ell} W_{V_i}=\{0\}.
\end{equation}
For all partitions $P_1\sqcup P_2\sqcup\cdots\sqcup P_s=\{0,1,\ldots,\ell\}$, by \cref{gintersec} and the fact that $\bigcap_{0\leq i\leq \ell}W_{V_{i}}=\{0\}$, we have 
\begin{equation}\label{eq:sum-vpi}
\sum_{i=1}^s \dim V_{P_i}\leq (s-1)(n-k),
\end{equation}
where $V_{P_i}=\bigcap_{j\in P_i} V_j$. 
Fix nonzero $\by\in \bigcap_{0\leq i\leq \ell}H_{V_{i}}$, which is possible as $\bigcap_{0\leq i\leq \ell}H_{V_{i}}\neq \{0\}$.

For $0\leq i\leq \ell$, choose $A_i\in \F_q^{n\times \dim V_i}$ such that $V_i=\langle A_i\rangle$.
As $\by\in \bigcap_{0\leq i\leq \ell}H_{V_{i}}$,
there exist vectors $\bu_0,\dots,\bu_{\ell}$ with $\bu_i\in  \F^{\dim V_i}$ such that $H A_{0}\bu_0=H A_{1}\bu_1=\cdots=H A_{\ell}\bu_\ell=\by$.
Let $\bv_i=A_{i} \bu_i$ for $0\leq i\leq \ell$, so that $H\bv_i=\by$. 
Define the partition $P_1\sqcup P_2\sqcup\cdots\sqcup P_s=\{0,1,\ldots,\ell\}$ such that $j,j'\in\{0,1,\dots,\ell\}$ are in the same set of the partition iff $\bv_j=\bv_{j'}$. For $i\in [s]$, let $\bv_{P_i}:=\bv_j$ for any $j\in P_i$. Then we have $s$ distinct vectors $\bv_{P_1},\dots,\bv_{P_s}\in \F^n$.

Consider arbitrary $i\in [s]$. For all $j\in P_i$ and $\bx\in V_j^{\perp}\subseteq \F_q^n$, we have $\bx^T \bv_{P_i}=\bx^T \bv_j=\bx^T A_j \bu_j=0$ since $\bx^T A_j=0$.
As this holds for all $j\in P_i$ and $\bx\in V_j^{\perp}$, we have that $\bx^T \bv_{P_i}=0$ for all $\bx\in \sum_{j\in P_i} V_j^{\perp}=\left(\bigcap_{j\in P_i} V_j\right)^{\perp}=V_{P_i}^{\perp}$. (Here $\sum_{j\in P_i} V_j^{\perp}=\left(\bigcap_{j\in P_i} V_j\right)^{\perp }$ holds by \cref{lm:dualspace}.)
So $V_{P_i}^{\perp}\subseteq \ker_{\F_q}(\bv_{P_i})$.
It follows by \cref{lem:rank-kernel} that 
\begin{equation}\label{eq:rank-bound-0}
\rank_{\F_q}(\bv_{P_i})=n-\dim(\ker_{\F_q}(\bv_{P_i}))\leq n-\dim\left(V_{P_i}^{\perp}\right)=\dim V_{P_i}.
\end{equation}
So we have 
\begin{equation}\label{eq:rank-bound}
\sum_{i=1}^s \rank_{\F_q}(\bv_{P_i})\stackrel{\eqref{eq:rank-bound-0}}{\leq} \sum_{i=1}^s \dim V_{P_i}\stackrel{\eqref{eq:sum-vpi}}{\leq} (s-1)(n-k).
\end{equation}
For $i\in [s]$, let $\bz_i=\bv_{P_i}-\bv_{P_1}$. Note that $H \bz_i=H \bv_{P_i}-H \bv_{P_1}=\by-\by=0$. So $\bz_1,\dots,\bz_s\in C$. And
\begin{equation}\label{eq:total-distance}
\sum_{i=1}^s d_R(-\bv_{P_1}, \bz_i)
=\sum_{i=1}^s \rank_{\F_q}(\bz_i-(-\bv_{P_1}))
=\sum_{i=1}^s \rank_{\F_q}(\bv_{P_i})\stackrel{\eqref{eq:rank-bound}}{\leq} (s-1)(n-k).
\end{equation}
If $s=1,$ we get from \eqref{eq:rank-bound} that $\rank_{\F_q}(\bv_{P_1})=0$, implying that $\bv_{P_1}=0$. But this is impossible since $\by=H \bv_{P_1}$ is nonzero. So $s\geq 2$.
By \eqref{eq:total-distance}, $C$ is not $\operatorname{LD-MRD}(s-1)$. As $s\leq \ell+1$, we see that $C$ is not $\ldmrd{\ell}$.
\end{proof}

Next, we show the other direction.

\begin{thm}\label{thm:mrd2ldmrd}
Let $C$ be an $[n,k]_{\F}$ code over a field $\F/\F_q$.
If $C^\perp$ is $\operatorname{MRD}(\ell+1)$, then $C$ is $\ldmrd{\ell}$.
\end{thm}
\begin{proof} 
We will instead prove the following equivalent statement: Suppose $C$ is MRD but not $\ldmrd{\ell}$. Then $C^\perp$ is not $\operatorname{MRD}(\ell+1)$. 

Let $H\in \F^{(n-k)\times n}$ be a parity check matrix of $C$. Then $H$ is also a generator matrix of $C^{\perp}$. Since $C$ is not $\ldmrd{\ell}$, there exist $L\leq \ell$, $\bz\in \F^n$, and distinct $\bz_0,\dots,\bz_L\in\F^n$ such that
\[
\sum_{i=0}^{L}\rank_{\F_q}(\bz_i-\bz)\leq L(n-k) \quad\text{and}\quad H\bz_0=\cdots=H\bz_{L}=0.
\]
Let $\bv_i=\bz_i-\bz\in\F^n$ for $0\leq i\leq L$.
Then $\bv_0,\ldots,\bv_{L}$ are distinct vectors satisfying 
\begin{equation}\label{eq:rank-bound-2}
\sum_{i=0}^{L}\rank_{\F_q}(\bv_i)\leq L(n-k) \quad\text{and}\quad H\bv_0=\cdots=H\bv_{L}.
\end{equation}
Let $\by:=H\bv_0$, which equals $H\bv_{i}$ for all $0\leq i\leq L$. Assume that $\by=0$. Then $H \bv_i=0$ for all $0\leq i\leq L$. So $\bv_0,\dots,\bv_L$ are codewords of $C$. As $C$ is MRD, all of these codewords $\bv_i$ except the zero codeword satisfy $\rank_{\F_q}(\bv_i)\geq n-k+1$. But this contradicts \eqref{eq:rank-bound-2}.
So $\by\neq 0$.
Moreover, we may assume that $\rank_{\F_q}(\bv_i)\leq n-k$ for all $0\leq i\leq L$, because if this were not the case, we could remove some $\bv_i$ and \eqref{eq:rank-bound-2} would still hold (for smaller $L$).


Consider arbitrary $i\in \{0,\dots,L\}$.
Write $\bv_i=(v_{i,1},\dots,v_{i,n})$.
Let $r_i=\rank_{\F_q}(\bv_i)$.
We know $r_i=\spa_{\F_q}\{v_{i,1},\dots,v_{i,n}\}$ by definition.
Pick $u_{i,1},\dots,u_{i,r_i}\in \F$ that form a basis of $\spa_{\F_q}\{v_{i,1},\dots,v_{i,n}\}$ over $\F_q$, and let $\bu_i=(u_{i,1},\dots,u_{i,r_i})\in\F^{r_i}$.
Then there exists a unique matrix $A_i\in\F_q^{n\times r_i}$ such that $\bv_i=A_i \bu_i$.
Let $V_i=\langle A_i\rangle$, i.e., $V_i$ is the column span of $A_i$ over $\F_q$.
Let $V_i'$ be the column span of $A_i$ over $\F$.
By definition, we have
\begin{enumerate}
\item $\dim_{\F_q} V_i=\dim_{\F} V_i'=r_i=\rank_{\F_q}(\bv_i)\leq n-k$.
\item $\bv_i=A_i\bu_i\in V_i'$.
\item $H_{V_i}=\{HA_i \bu: \bu\in \F^{r_i}\}=\{H\bx: \bx\in V_i'\}$. In particular, $H \bv_i\in H_{V_i}$.
\end{enumerate}


As $\by=H\bv_0=\cdots=H\bv_{L}$, we have $\by\in \bigcap_{0\leq i\leq L}H_{V_{i}}$.
Let $W=(Z_{i,j})_{i\in [n-k],j\in [n]}$.
First consider the case where $\bigcap_{0\leq i\leq L}W_{V_{i}}=\{0\}$.
Note that,
\[
\dim \left(\bigcap_{0\leq i\leq L}H_{V_{i}}\right)>0=\dim\left(\bigcap_{0\leq i\leq L}W_{V_{i}}\right)
\]
since we already know $\by\neq 0$ and $\by\in \bigcap_{0\leq i\leq L}H_{V_{i}}$.
So $C^{\perp}$ is not $\operatorname{MRD}(\ell+1)$.

Now consider the case where $\bigcap_{0\leq i\leq L}W_{V_{i}}\neq\{0\}$.
By \cref{gintersec}, we have
\begin{equation}\label{eq:intersect-h-v}
\dim\left(\bigcap_{0\leq i\leq L}W_{V_{i}}\right)=\sum_{i=1}^s \dim(V_{P_i})-(s-1)(n-k)>0
\end{equation}
for some partition  $P_1\sqcup P_2\sqcup\cdots\sqcup P_s=\{0,1,\ldots,L\}$, where $V_{P_i}=\bigcap_{j\in P_i} V_j$. Note that this partition is not the finest partition of $\{0,1\dots,L\}$ since 
\[
\sum_{i=0}^L \dim V_i-L(n-k)=\sum_{i=0}^L \rank_{\F_q}(\bv_i)-L(n-k)\stackrel{\eqref{eq:rank-bound-2}}{\leq} 0.
\]
So $s<L+1$.

Define the $\F$-subspace $V\subseteq \prod_{i=0}^L \F^{r_i}$ to be
\[
V=\left\{(\bc_0,\ldots,\bc_L)\in \prod_{i=0}^L \F^{r_i}: H A_0\bc_0=\cdots=H A_L\bc_{L}\right\}.
\]
Then $\sigma: (\bc_0,\dots,\bc_L)\mapsto HA_0 \bc_0$ maps $V$ to $\bigcap_{0\leq i\leq L}H_{V_{i}}$.
As $C^\perp$ is MRD (since $C$ is) and $\dim V_i\leq n-k$ for $0\leq i\leq L$, for every $\bw\in \bigcap_{0\leq i\leq L}H_{V_{i}}$, there exists unique $\bx=(\bc_0,\dots,\bc_L)\in V$ such that $\sigma(\bx)=\bw$, i.e., $\sigma: V\to \bigcap_{0\leq i\leq L}H_{V_{i}}$ is an isomorphism of $\F$-linear spaces. In particular,
\begin{equation}\label{eqref:intersect-h-v-2}
\dim\left(\bigcap_{0\leq i\leq L}H_{V_{i}}\right) = \dim V.
\end{equation}
Define 
\[
V'=\left\{(\bc_0,\ldots,\bc_L)\in V: A_j\bc_j=A_{j'}\bc_{j'} \text{ for } i\in [s] \text{ and } j,j'\in P_i\right\}\subseteq V.
\]
Note that $(\bu_0,\dots,\bu_L)\in V\setminus V'$ since $\bv_0=A_0\bu_0,\dots,\bv_L=A_L\bu_L$ are distinct and $s<L+1$. So 
\begin{equation}\label{eqref:v-vprime}
\dim V>\dim V'.
\end{equation}
Consider $(\bc_0,\dots,\bc_L)\in V'$ and $i\in [s]$. For $j,j'\in P_i$, we have
$A_j\bc_j=A_{j'}\bc_{j'}\in V'_j\cap V'_{j'}$. 
So for $j\in P_i$, it holds that 
\[
\bc_j\in \bigcap_{j'\in P_i} V_{j'}'=:V'_{P_i}.
\]
As each $V_j'$ is the $\F$-span of the $\F_q$-linear space $V_j$, here $V'_{P_i}$ is the $\F$-span of the $\F_q$-linear space $V_{P_i}=\bigcap_{j\in P_i} V_{j}$.\footnote{Each $V_j'$ may be identified with $V_j\otimes_{\F_q} \F$. We are using the fact that $\bigcap_{j\in P_i} (V_j\otimes_{\F_q} \F)=
\left(\bigcap_{j\in P_i} V_j\right) \otimes_{\F_q} \F$.
See \cite{Ell19}.
}
So $\sigma$ maps $V'$ to $\bigcap_{i\in [s]} H_{V_{P_i}}$.
And for every $\bw\in \bigcap_{i\in [s]} H_{V_{P_i}}$, the inverse $\bc=(\bc_0,\dots,\bc_L)=\sigma^{-1}(\bw)\in V$ satisfies that $A_j\bc_j=A_{j'}\bc_{j'}$ for $i\in [s]$ and $j,j'\in P_i$ by the uniqueness of $\bc$. So $\sigma$ restricts to an isomorphism from $V'$ to $\bigcap_{i\in [s]} H_{V_{P_i}}$.
Therefore,
\begin{equation}\label{eq:intersect-w-v}
\begin{aligned}
\dim V'=\dim\left(\bigcap_{i\in [s]} H_{V_{P_i}}\right)\geq \dim\left(\bigcap_{i\in [s]} W_{V_{P_i}}\right)
&\stackrel{\eqref{intersec1}}{\geq} \sum_{i=1}^s \dim(V_{P_i})-(s-1)(n-k)\\
&\stackrel{\eqref{eq:intersect-h-v}}{=}\dim\left(\bigcap_{0\leq i\leq L}W_{V_{i}}\right)
\end{aligned}
\end{equation}
where the second step uses \cref{cor:semicontinuity}.
Then
$\dim \left(\bigcap_{0\leq i\leq L}H_{V_{i}}\right)>\dim\left(\bigcap_{0\leq i\leq L}W_{V_{i}}\right)$ by \eqref{eqref:intersect-h-v-2}, \eqref{eqref:v-vprime}, and \eqref{eq:intersect-w-v}.
So $C^{\perp}$ is not $\operatorname{MRD}(\ell+1)$. 
\end{proof}

Combining \cref{thm:ldmrd2mrd} and \cref{thm:mrd2ldmrd} proves \cref{2equiv}.

\section{Putting It Together}
The last tool we need is the duality of Gabidulin codes, as stated by the following theorem. For the proof, see \cite[Lemma~2.7.2]{bartz2022rank}.\footnote{\cite[Lemma~2.7.2]{bartz2022rank} assumes $\F$ to be a finite extension over $\F_q$. However, the same proof applies to any extension field $\F$ of $\F_q$ without modification. }

\begin{thm}[Duality of Gabidulin codes]\label{duali}
Let $\F$ be an extension field of $\F_q$, and let $\alpha_1,\dots,\alpha_n\in \F$ be linearly independent over $\F_q$.
Then there exists $(\beta_1,\dots,\beta_n)\in \F^n\setminus\{0\}$ such that
\begin{equation}\label{eq:pairing}
\sum_{i=1}^n \alpha_i^{q^{j-1}}\beta_i^{q^{h-1}}=0 \qquad \text{for $(j,h)\in [k]\times [n-k]$}.
\end{equation}
The choice of $(\beta_1,\dots,\beta_n)$ satisfying \eqref{eq:pairing} is unique up to a scalar in $\F\setminus\{0\}$.
Moreover, $\beta_1,\dots,\beta_n$ are linearly independent over $\F_q$, and $\left(\beta_j^{q^{i-1}}\right)_{i\in [n-k], j\in [n]}$ is a parity check matrix of $\mG_{n,k}(\alpha_1,\dotsc, \alpha_n)$, i.e., 
\[
\mG_{n,k}(\alpha_1,\dotsc, \alpha_n)^\perp=\mG_{n,n-k}(\beta_1,\dotsc, \beta_n).
\]
\end{thm}

We are now ready to prove our main theorems (\cref{mm1} and \cref{mm2}) on the list decodability of random Gabidulin codes:

\begin{thm}\label{thm:list-decodability}
Let $\delta>0$ and $\ell\in [k]$.
Let $1\leq k\leq n\leq m$ be integers such that $q^m\geq 3(n-k)q^{n(n-k)\cdot\min\{\ell+1,n-k\}+n-k}/\delta$.
 Let $(\alpha_1,\dots,\alpha_n)$ be uniformly distributed over the set of all vectors in $\F_{q^m}^n$ whose coordinates are linearly independent over $\F_q$. Then it holds with probability at least $1-\delta$ that the Gabidulin code $\mG_{n,k}(\alpha_1,\dots, \alpha_n)$ over $\F_{q^m}$ is $\ldmrd{\ell}$, i.e., it is $\left(\frac{L}{L+1}\left(1-{k}/{n}\right), L\right)$-average radius list decodable (and hence also $\left(\frac{L}{L+1}\left(1-{k}/{n}\right), L\right)$-list decodable) for all $L\in [\ell]$.
\end{thm}

\begin{proof}
Let $S$ be the set of all $(\alpha_1,\dots,\alpha_n)\in \F^n$ whose coordinates are linearly independent over $\F_q$. Define an equivalence relation $\sim$ on $S$ by letting $\alpha\sim \alpha'$ iff $\alpha=c\alpha'$ for some $c\in \F\setminus\{0\}$. Denote the equivalence class of $\alpha$ by $[\alpha]$ and denote the set of the equivalence classes by $S/\sim$.
By \cref{duali}, for $\alpha=(\alpha_1,\dots,\alpha_n)\in S$, we can find its dual basis $\beta=(\beta_1,\dots,\beta_n)\in S$ by solving \eqref{eq:pairing}, and $[\beta]$ is uniquely determined by $\alpha$. Also note that scaling $\alpha$ does not affect $[\beta]$.
So we obtain a map from $S/\sim$ to itself that sends $[\alpha]$ to $[\beta]$, where $\alpha$ and $\beta$ satisfy \cref{duali}.
Moreover, by applying \cref{duali} again, but with $k$ replaced by $n-k$, we can solve $\alpha$ from $\beta$ using \eqref{eq:pairing} and hence get the map $[\beta]\mapsto [\alpha]$. The two maps are inverse to each other by definition. So the map $[\alpha]\mapsto [\beta]$ is a permutation of $S/\sim$. 

Now let $\alpha$ be uniformly distributed over $S$. From the discussion above, we see that $[\beta]$, which is uniquely determined by $\alpha$ via \eqref{eq:pairing}, is uniformly distributed over $S/\sim$.
Moreover, the Gabidulin code $\mG_{n,n-k}(\beta_1,\dots, \beta_n)$ depends only on the equivalence class $[\beta]$ of $\beta=(\beta_1,\dots, \beta_n)$, since scaling $\beta$ corresponds to scaling the rows of the generator matrix $\left(\beta_{j}^{q^{i-1}}\right)_{i\in [n-k], j\in [n]}$, which does not change the code.

By \cref{cor:three_equiv}, the Gabidulin code $\mG_{n,k}(\alpha_1,\dots, \alpha_n)$ is $\ldmrd{\ell}$ iff its dual code $\mG_{n,n-k}(\beta_1,\dots, \beta_n)$ is $\operatorname{GKP}(\ell+1)$.
As $[\beta]$ is uniformly distributed over $S/\sim$, by \cref{thm:GM-MRD-finite}, $\mG_{n,n-k}(\beta_1,\dots, \beta_n)$ is $\operatorname{GKP}(\ell+1)$ with probability at least $1-3(n-k)q^{n(n-k)\cdot\min\{\ell+1,n-k\}+n-k-m}\geq 1-\delta$. Therefore, $\mG_{n,k}(\alpha_1,\dots, \alpha_n)$ is $\ldmrd{\ell}$ with the same probability.
\end{proof}
Setting $\ell=\lceil\frac{1-R-\epsilon}{\epsilon}\rceil\leq \frac{1-R}{\epsilon}$, by \cref{thm:list-decodability}, we obtain the following corollary, which shows that random Gabidulin codes achieve list decoding capacity in the rank metric with high probability.
\begin{cor}\label{cor:mm1}
Let $\delta>0$. Let $1\leq k\leq n\leq m$ with $q^m\geq 3(n-k)q^{n(n-k)\cdot\min\{(1-k/n)/\epsilon+1,n-k\}+n-k}/\delta$.
Let $\epsilon>0$ and let $(\alpha_1,\dots,\alpha_n)$ be uniformly distributed over the set of all vectors in $\F_{q^m}^n$ whose coordinates are linearly independent over $\F_q$. Then it holds with probability at least $1-\delta$ that the Gabidulin code $\mG_{n,k}(\alpha_1,\dots, \alpha_n)$ over $\F_{q^m}$ is $\left(1-R-\epsilon,\frac{1-R}{\epsilon}\right)$-average radius list decodable (and hence also $\left(1-R-\epsilon,\frac{1-R}{\epsilon}\right)$-list decodable) in the rank metric, where $R=k/n$ is the rate of the code.  
\end{cor}

We conclude this section with a discussion about the symbolic Gabidulin code $\mG_{n,k}(Z_1,\dots, Z_n)$.
By \cref{thm:GM-MRD-v1-new} and \cref{1equiv}, $\mG_{n,k}(Z_1,\dots, Z_n)$ is $\operatorname{GKP}(\ell)$ and $\operatorname{MRD}(\ell)$ for all $\ell$.
It is also $\ldmrd{\ell}$ for all $\ell$, which can be shown via a proof strategy similar to that of \cref{thm:list-decodability}.
The key difference is that we need to argue that certain elements $Y_1,\dots,Y_n\in\F_q(Z_1,\dots,Z_n)$, which form a dual basis of $Z_1,\dots,Z_n$, are algebraically independent over $\F_q$, so that no nonzero polynomial over $\F_q$ can vanish at $(Y_1,\dots,Y_n)$. For this, we need to pass to the projective space $\mathbb{P}^n$ over $\F_q$ and view the maps $[\alpha]\mapsto [\beta]$ and $[\beta]\mapsto [\alpha]$ as well-defined morphisms that are inverse to each other (over dense open subsets of $\mathbb{P}^n$). 
These morphisms must preserve dimension and, consequently, algebraic independence. The details are omitted.

\section{Conclusions and Future Directions}
In this paper, we have proved that, with high probability, a random Gabidulin code $C\subseteq \F_{q^m}^n$ is list decodable, and even average-radius list decodable, up to the optimal generalized Singleton bound when $m$ is large enough.
Our result requires that $m=\Omega_{\ell}(n^2)$, which is optimal up to a factor depending only on $\ell$, as demonstrated by our lower bound. In achieving this, we have formulated various notions of higher-order MRD codes and established their equivalence, analogous to the work of Brakensiek, Gopi, and Makam \cite{brakensiek2023generic}. We have also proved the GM-MRD theorem, which is an essential ingredient for our main results.

We conclude with the following potential future directions: (1) Are there explicit constructions of Gabidulin codes whose parameters match or come close to those achieved in \cref{mm1}? (2) Is it possible to reduce the parameter $m$ by slightly compromising the rate of the code, in a manner analogous to the work of Guo and Zhang \cite{guo2023randomly}, and Alrabiah, Guruswami, and Li \cite{alrabiah2023randomly}? This does not contradict \cref{thm:informal-lower-bound}
since the list decoding radius in \cite{guo2023randomly, alrabiah2023randomly} is slightly worse than the optimal bound $\frac{\ell}{\ell+1}(1-R)$.
(3) Inspired by the work of \cite{brakensiek2023generalized}, we propose a similar conjecture below, which we call the ``ultimate GM-MRD conjecture.''
\begin{conj}[Ultimate GM-MRD conjecture]
    Let $C\in\F_{q^m}^N$ be any $[N,k]$ MRD code with a generator matrix $G\in\F_{q^m}^{k\times N}$. Let $C^{\prime}$ be the $[n,k]$ code defined by the generator matrix $G^{\prime}:=GA$, where $A\in\F_q^{N\times n}$ is a randomly sampled full rank matrix. Then, with probability $1-o_N(1)$, the code $C^{\prime}$ is $\operatorname{GKP}(\ell)$ for all $\ell\ge 1$.
\end{conj}




\subsection*{Acknowledgments} 
The authors would like to thank Joshua Brakensiek, Manik Dhar, and Sivakanth Gopi for many useful discussions and suggestions that helped this paper. They also appreciate the helpful comments from the anonymous reviewers.
Part of this work was carried out while Zeyu Guo and Zihan Zhang were visiting the Simons Institute for the Theory of Computing at UC Berkeley. They would like to thank the institute for its support and hospitality.

The work of Chaoping Xing was supported in part by the National Key Research and Development Program of China under the Grant 2022YFA1004900, in part by the National Natural Science Foundation of China under the Grants 12031011, 12361141818 and 12271084. The work of Chen Yuan was supported in part by the National Key Research and Development Program of China under the Grant 2023YFE0123900, in part by the National Natural Science Foundation of China under the Grants 12101403.

\bibliographystyle{alpha}
\bibliography{ref}

\appendix
\section{Field Size Lower Bound for $\operatorname{LD-MRD}(\ell)$}

We prove a lower bound on the field size of $\operatorname{LD-MRD}(\ell)$ codes by adapting
the argument in \cite{alrabiah2023ag}. 

\begin{thm}\label{A.1}
Let $\ell\geq 2$.
For any $r\in [0,1]$, any $\operatorname{MRD}$ code $C\subseteq \F_{q^m}^n$ of rate $R$ that is $\left(\frac{\ell\left(1-R\right)}{\ell+1}, \ell\right)$-avearge-radius list-decodable must have $m=\Omega_{\ell}((nR-1)(n-\ell-nR+1))$, which is $\Omega_{\ell}(n^2)$ if the rate $R$ of $C$ is in $[c, 1-c-\ell/n]$ for some constant $c>0$. 
\end{thm}
\begin{proof}
Fix a subspace $V_0\subseteq \F_q^n$ of dimension $\ell$. Choose a subspace $\overline{V}_0$ such that $V_0\oplus \overline{V}_0=\F_q^n$. Assume the minimum distance of $C$ is $d$. As $C$ is MRD, the size of $C$ is $q^{m(n-d+1)}$ \cite[Theorem 5.4]{Del}. 
Let $k=n-d+1=nR$. 
For any two distinct codewords $M_1,M_2\in C$ (viewed as matrices in $\F_q^{m\times n}$), we have $\rank(M_1-M_2)\geq d=n-k+1$.
Let $\mF$ be the collection of subspaces $V\subseteq \overline{V}_0$ of dimension $k-1$. The size of $\mF$ is the number of subspaces of dimension $k-1$ contained in a subspace of dimension $n-\ell$, which is at least $q^{(n-\ell-k+1)(k-1)}$. 
It suffices to prove that $\ell q^{\ell m}\geq |\mF|/2$, as this would imply $m=\Omega_{\ell}((k-1)(n-\ell-k+1))$. 

Assume to the contrary that $\ell q^{\ell m}<|\mF|/2$. Let $M$ be uniformly distributed from $C$. For a fixed subspace $V\in \mF$, let $A\in \F_q^{n\times (k-1)}$ such that $\langle A\rangle=V$. Let $E_V$ be the event that there exists a codeword $M_1\in C$ different from $M$ such that $MA=M_1A$, i.e., $(M-M_1)A=0$. 
If $E_V$ does not hold, then $M$ is uniquely determined by $MA\in \F_{q}^{m\times{(k-1)}}$. As the number of possible values of $M A$ is at most $q^{(k-1)m}$ and $|C|=q^{mk}$, we have
\[
\Pr[\lnot E_V]\leq \frac{q^{m(k-1)}}{q^{mk}}=q^{-m}. 
\]
Therefore, over random $M\in C$, the expected number of $V\in \mF$ such that $E_V$ happens is $\sum_{V\in \mF}(1-\Pr[\lnot E_V])\geq |\mF|/2$. 
Then, we can fix a codeword $M\in C$ such that the size of the set 
\[
\mF_M:=\{V\in \mF: E_V \text{ happens}\}
\]
is at least $|\mF|/2$.

Let $A_0\in \F_q^{n\times \ell}$ such that $\langle A_0\rangle =V_0$. 
By the definition of $\mathcal{F}_M$, for each $V\in\mathcal{F}_M$, there exists a codeword $M_V\neq M$ such that the kernel subspace of $M-M_V$ contains $V$.
Since $M_V A_0\in\F_q^{m\times \ell}$ for any codeword $M_V$ and $\ell q^{\ell m}<|\mF|/2\leq |\mF_M|$, by the pigeonhole principle, there exists distinct $V_1,\ldots,V_\ell\in \mathcal{F}_M$ 
such that $M_{V_1}A_0=\dots=M_{V_\ell}A_0$. Moreover, by the definition of $\mathcal{F}_M$, for $i=1,\dots,\ell$, there exists $A_i\in\F_q^{n\times (k-1)}$ with $\langle A_i\rangle=V_i$ such that $(M-M_{V_i})A_i=0$.


Assume $M_{V_i}=M_{V_j}$ for some $i\neq j$. Then $(M-M_{V_i})A_i=0$ and $(M-M_{V_i})A_j=0$.
Let $A\in\F_q^{n\times \dim(V_i+V_j)}$ such that $\langle A\rangle =V_i+V_j$.
As the columns of $A$ are in $V_i+V_j=\langle A_i\rangle+\langle A_j\rangle$, we have $(M-M_{V_i})A=0$, i.e., $V_i+V_j$ is contained in the kernel subspace of $M-M_{V_i}$. Since $M$ and $M_{V_i}$ are in the MRD code $C$, we have $\rank(M-M_{V_i})\geq n-k+1$. This implies that the kernel subspace of $M-M_{V_i}$ is at most $k-1$. So $\dim (V_i+V_j)\leq k-1$.
However, as $V_i\neq V_j$ and $\dim V_i=\dim V_j=k-1$, we have $\dim (V_i+V_j)\geq k$, which yields a contradiction. 
Thus, we conclude that $M_{V_1},\ldots,M_{V_{\ell}}$ are all distinct. 

Since $\overline{V}_0\cap V_0=\{0\}$, there exists $B_0\in \F_q^{n\times (n-\ell)}$ such that $\langle B_0\rangle =\overline{V}_0$ and $\begin{pmatrix} A_0 & B_0\end{pmatrix}\in \F_q^{n\times n}$ has full rank. Let $Y\in \F_q^{m\times n}$ such that $(M_{V_1}-Y)A_0=\cdots=(M_{V_{\ell}}-Y)A_0=0$ 
and $(M-Y)B_0=0$. This can be achieved by choosing $Y=\begin{pmatrix} M_{V_1}A_0 & MB_0\end{pmatrix}\begin{pmatrix} A_0 & B_0\end{pmatrix}^{-1}$.

For $i\in [\ell]$, we have $(M-Y) A_i=0$ since $\langle A_i\rangle=V_i$, $V_i\subseteq \overline{V}_0$, $\overline{V}_0=\langle B_0\rangle$, and $(M-Y)B_0=0$.
And for $i\in [\ell]$, we know $(M-M_{V_i})A_i=0$, which implies
\[
(M_{V_i}-Y)A_i=(M_{V_i}-M)A_i+(M-Y)A_i=0
\quad\text{and}\quad
(M_{V_i}-Y)A_0=0.\]
Since $V_0\cap \langle V_i\rangle\subseteq V_0\cap \overline{V}_0=\{0\}$ for $i\in [\ell]$, we have $\dim (V_0+V_i)=\dim V_0+\dim V_i=\ell+k-1$ and hence
\[
\rank(M_{V_i}-Y)\leq n-(\ell+k-1)\leq n-k-1,
\]
where we use the fact that $\ell\ge 2.$ As $(M-Y)B_0=0$, we have $\rank(M-Y)\leq n-\dim(\overline{V}_0)=\ell$. 
It follows that 
\[
\rank(M-Y)+\sum_{i=1}^{\ell}\rank(M_{V_i}-Y)\leq \ell+\ell(n-k-1)=\ell(n-k)
\]
which contradicts the claim that $C$ is $\left(\frac{\ell\left(1-k/n\right)}{\ell+1}, \ell\right)$-avearge-radius list-decodable.
\end{proof}

\clearpage 
\includepdf[pages=-]{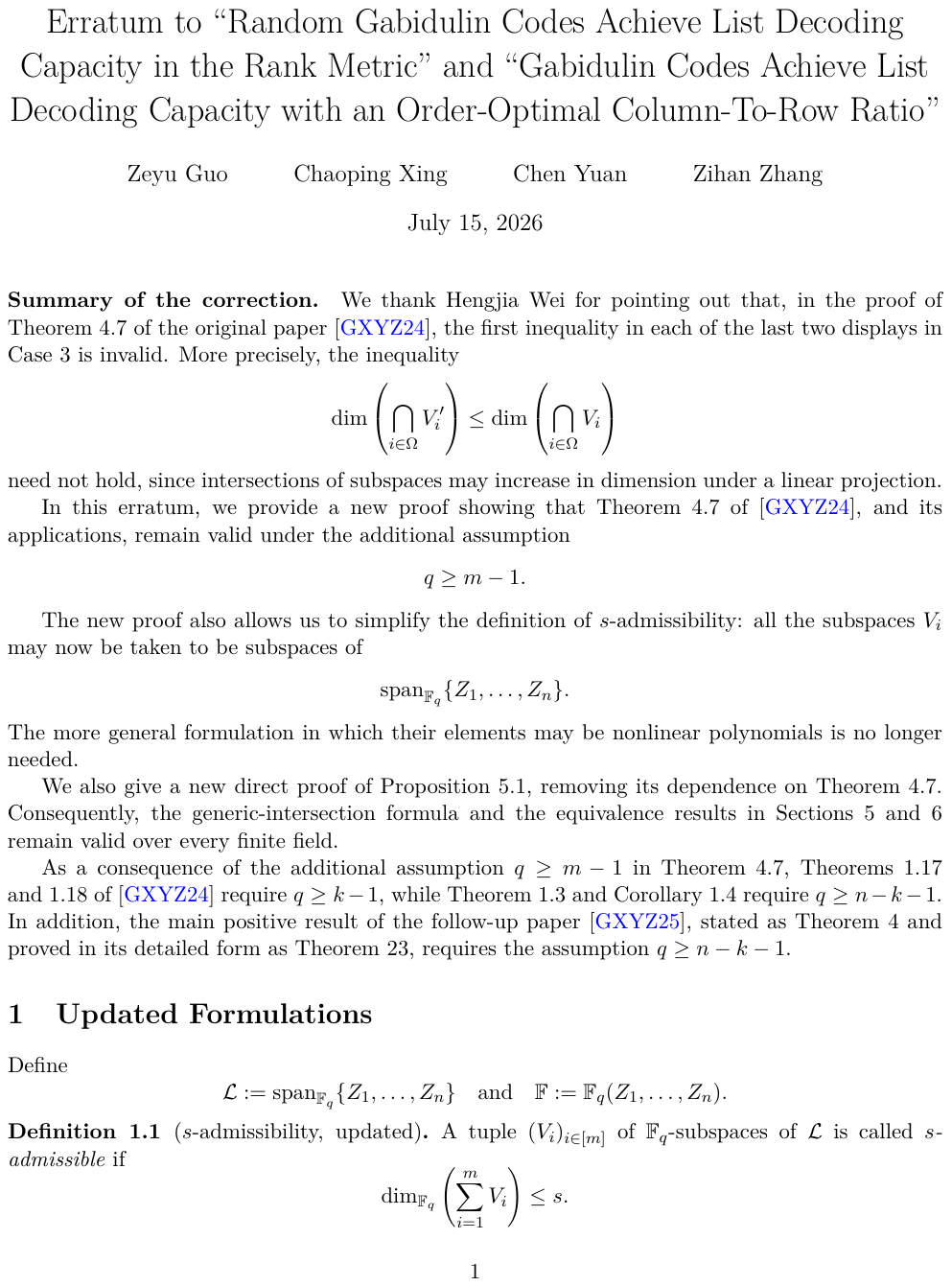}

\end{document}